\newcommand{\setCollector}[2]{\{\ #1\ \mid\ #2 \ \}}
\newcommand{\setCompact}[2]{\{ #1 \mid #2  \}}
\newcommand{\setof}[1]{\{#1\}}
\newcommand{\code}[1]{\texttt{#1}}
\newcommand*{\rom}[1]{\uppercase\expandafter{\romannumeral #1\relax}}

\newcommand{\factorize}[2]{#1/_{#2}}

\newcommand{\toolname}[1]{Merlin}

\newcommand{\bnf}{\mid}
\newcommand{\avar}{\texttt{x}}

\newcommand{\infixconcat}[2]{#1.#2}
\newcommand{\shortReplace}{\texttt{r}}
\newcommand{\shortReplaceOf}[3]{\shortReplace(#1, #2, #3)}
\newcommand{\emptystring}{\epsilon}
\newcommand{\constConf}[1]{\text{"\textvisiblespace\code{Conference}"}}
\newcommand{\constCity}[1]{\text{"\textvisiblespace\code{City}"}}
\newcommand{\reducedstrut}{\vrule width 0pt height .9\ht\strutbox depth .9\dp\strutbox\relax}
\newcommand{\myColorbox}[2]{\begingroup
  \setlength{\fboxsep}{0pt}\colorbox{#1}{\reducedstrut
  $\displaystyle#2$\/}\endgroup}

\newcommand{\oepruneColor}[1]{\myColorbox{red!70!white!80}{#1}}
\newcommand{\oepruneReuseColor}[1]{\myColorbox{blue!85!white!40}{#1}}
\newcommand{\metricpruneColor}[1]{\myColorbox{green!60!white!80}{#1}}
\newcommand{\semisuccessColor}[1]{\myColorbox{cyan!60!white!80}{#1}}
\newcommand{\successColor}[1]{\myColorbox{yellow!60}{#1}}

\newcommand{\sketchhole}{\mathord{?}}

\newcommand{\aval}{a}
\newcommand{\avalp}{b}
\newcommand{\bigconst}{c}

\newcommand{\semanticsOf}[1]{\llbracket #1 \rrbracket}
\newcommand{\semanticsOfAppliedTo}[2]{\semanticsOf{#1}(#2)}
\newcommand{\agrammar}{\mathcal{G}}
\newcommand{\inputexamples}{\mathit{I}}
\newcommand{\inputexamplesp}{\mathit{J}}

\newcommand{\outputexamples}{O}

\newcommand{\ioexamples}{(\inputexamples, \outputexamples)}

\newcommand{\languageOf}[1]{L(#1)}
\newcommand{\termLanguageOf}[1]{\widetilde{L}(#1)}
\newcommand{\completeLanguageOf}[1]{\mathcal{L}(#1)}
\newcommand{\completeTermLanguageOf}[1]{\widetilde{\mathcal{L}}(#1)}
\newcommand{\sygus}[1]{SyGuS}

\newcommand{\dataDomain}{\textit{D}}
\newcommand{\groundTruthFuncs}{\mathit{GT}}
\newcommand{\groundTruthFunc}{\mathit{gt}}
\newcommand{\groundTruthFuncOf}[1]{\groundTruthFunc(#1)}
\newcommand{\groundTruthFuncpOf}[1]{\groundTruthFuncp(#1)}
\newcommand{\allFuncs}{\mathcal{F}}
\newcommand{\nonterms}{\mathcal{N}}
\newcommand{\terminals}{\mathcal{T}}
\newcommand{\anonterm}{\mathcal{A}}

\newcommand{\startSymbol}{\mathcal{S}}
\newcommand{\productions}{\mathcal{P}}

\newcommand{\anoutput}{\mathit{o}}
\newcommand{\aninput}{\mathit{i}}
\newcommand{\kleeneOf}[1]{#1^{*}}

\newcommand{\aspecification}{\Phi}

\newcommand{\anInterpretation}{\mathcal{I}}
\newcommand{\anInterpretationOf}[1]{\anInterpretation(#1)}
\newcommand{\anInterpretationOfOf}[2]{\anInterpretationOf{#1}(#2)}

\newcommand{\variables}{\mathit{Vars}}

\newcommand{\nat}{\mathbb{N}}

\newcommand{\aFunc}{f}
\newcommand{\aFuncOf}[1]{\aFunc(#1)}
\newcommand{\aFuncp}{g}
\newcommand{\aFuncpOf}[1]{\aFuncp(#1)}
\newcommand{\aFuncpp}{h}
\newcommand{\aFuncppOf}[1]{\aFuncpp(#1)}
\newcommand{\anOperator}{\texttt{op}}
\newcommand{\anOperatorOf}[1]{\anOperator(#1)}
\newcommand{\functionDistanceMeasure}{\textit{m}}
\newcommand{\functionDistanceMeasureOf}[2]{\functionDistanceMeasure(#1, #2)}
\newcommand{\functionDistanceMeasurep}{\functionDistanceMeasure'}
\newcommand{\functionDistanceMeasurepOf}[2]{\functionDistanceMeasurep(#1, #2)}
\newcommand{\reals}{\mathbb{R}}
\newcommand{\realsgtz}{\mathbb{R}_{\geq0}}
\newcommand{\groundTruthFuncp}{\groundTruthFunc'}

\newcommand{\aprog}{\texttt{p}}
\newcommand{\aprogp}{\texttt{q}}
\newcommand{\aprogpp}{\aprogp'}
\newcommand{\enumerationOrder}[1]{\preceq_{#1}}
\newcommand{\enumerationOrderNotEq}[1]{\prec_{#1}}

\newcommand{\aprogSet}{\mathit{P}}
\newcommand{\aprogSetp}{\mathit{Q}}

\newcommand{\successor}{\mathit{succ}}

\newcommand{\refineFunc}{\texttt{Refine}}
\newcommand{\refineFuncOf}[3]{\refineFunc(#1, #2, #3)}
\newcommand{\findProgFunc}{\texttt{Search}}
\newcommand{\findProgFuncOf}[4]{\findProgFunc(#1, #2, #3, #4)}
\newcommand{\factorizeFunc}{\texttt{Factorize}}
\newcommand{\factorizeFuncOf}[3]{\factorizeFunc(#1, #2, #3)}
\newcommand{\pruneFunc}{\texttt{Prune}}
\newcommand{\pruneFuncOf}[3]{\pruneFunc(#1, #2, #3)}

\newcommand{\successorIn}[1]{\successor_{#1}}
\newcommand{\successorInOf}[2]{\successorIn{#1}(#2)}
\newcommand{\updateEnumOrder}{\texttt{Learn}}
\newcommand{\updateEnumOrderOf}[3]{\updateEnumOrder(#1, #2, #3)}

\newcommand{\size}{\textit{sz}}
\newcommand{\sizeOf}[1]{\size(#1)}
\newcommand{\distanceMeasure}{\widetilde{\functionDistanceMeasure}}
\newcommand{\distanceMeasureOf}[2]{\distanceMeasure(#1, #2)}
\newcommand{\distanceMeasureForExamples}[1]{\functionDistanceMeasure_{#1}
}\newcommand{\distanceMeasureForExamplesOf}[3]{\distanceMeasureForExamples{#1}(#2,#3)}

\newcommand{\asketch}{\code{s}}
\newcommand{\asketchOf}[1]{\asketch(#1)}

\newcommand{\length}{\mathit{len}}
\newcommand{\lengthOf}[1]{\length(#1)}

\newcommand{\concatOp}{\texttt{concat}}
\newcommand{\concatOpOf}[2]{\concatOp(#1, #2)}
\newcommand{\substrOp}{\texttt{substr}}
\newcommand{\substrOpOf}[3]{\substrOp(#1, #2, #3)}
\newcommand{\replaceOp}{\texttt{replace}}
\newcommand{\replaceOpOf}[3]{\replaceOp(#1, #2, #3)}

\newcommand{\iteop}{\texttt{ite}}
\newcommand{\xorop}{\texttt{xor}}
\newcommand{\andop}{\texttt{and}}
\newcommand{\orop}{\texttt{or}}
\newcommand{\mulop}{\texttt{mul}}
\newcommand{\addop}{\texttt{add}}
\newcommand{\negop}{\texttt{neg}}
\newcommand{\notop}{\texttt{not}}

\newcommand{\asizeThreshold}{\mathit{s}}

\newcommand{\cardinalityOf}[1]{|#1|}

\newcommand{\levenshteinDistance}{\textit{lvst}}
\newcommand{\levenshteinDistanceOf}[2]{\levenshteinDistance(#1, #2)}

\newcommand{\leadingZeros}{\textit{nlz}}
\newcommand{\leadingZerosOf}[1]{\leadingZeros(#1)}
\newcommand{\hammingDistance}{\textit{HDist}}
\newcommand{\hammingDistanceOf}[2]{\hammingDistance(#1, #2)}
\newcommand{\absOf}[1]{\lvert #1 \rvert}

\newcommand{\abitvec}{\mathit{a}}
\newcommand{\abitvecp}{\mathit{b}}
\newcommand{\astring}{\mathit{str}}
\newcommand{\astringp}{\astring'}
\newcommand{\substr}{\textit{substr}}
\newcommand{\substrOf}[2]{\substr(#1,#2)}

\newcommand{\mycheckmark}{\ding{51}}
\newcommand{\myxmark}{\ding{55}}

\documentclass[acmsmall,screen,nonacm,authorversion]{acmart} 
\usepackage{amsmath}
\usepackage{galois}
\usepackage{stmaryrd}
\usepackage{xcolor}
\usepackage{algorithm}
\usepackage{algpseudocode}
\usepackage{cleveref}
\usepackage{xcolor}
\usepackage{pifont}
\usepackage{enumitem}
\usepackage{subcaption}
\usepackage{wrapfig}
\usepackage{tikz}
\usetikzlibrary{arrows.meta, positioning, calc}

\AtBeginDocument{}

\setcopyright{acmlicensed}
\copyrightyear{2018}
\acmYear{2018}
\acmDOI{XXXXXXX.XXXXXXX}
\acmConference[Conference acronym 'XX]{Make sure to enter the correct
  conference title from your rights confirmation email}{June 03--05,
  2018}{Woodstock, NY}
\acmISBN{978-1-4503-XXXX-X/2018/06}

\newcommand{\radius}{\mathit{r}}

\newcommand{\ballof}[2]{\mathcal{B}_{#1}(#2)}
\newcommand{\ballfullof}[3]{\ballof{#1, #2}{#3}}
\newcommand{\cballof}[2]{\mathcal{B}_{#1}[#2]}

\newcommand{\approxmetric}{\mathit{m}^{\sharp}}
\newcommand{\approxmetricof}[2]{\approxmetric(#1, #2)}
\newcommand{\approxmetricpar}[1]{{\mathit{m}_{#1}}^{\sharp}}
\newcommand{\approxmetricparof}[3]{\approxmetricpar{#1}(#2, #3)}

\newcommand{\inducedequivof}[1]{\equiv_{#1}}

\newcommand{\ametric}{\functionDistanceMeasure}
\newcommand{\ametricof}[2]{\functionDistanceMeasureOf{#1}{#2}}

\newcommand{\ametricpar}[1]{\ametric_{#1}}
\newcommand{\ametricparof}[3]{\ametricpar{#1}(#2, #3)}

\newcommand{\ametrictilde}{\tilde \ametric}
\newcommand{\ametrictildeof}[2]{\ametrictilde(#1, #2)}

\newcommand{\concatmetric}{\ametric_{\concatOp}}
\newcommand{\concatmetricof}[2]{\concatmetric(#1, #2)}
\newcommand{\substrmetric}{\ametric_{\substrOp}}

\newcommand{\andmetric}{\ametric_{\andop}}
\newcommand{\andmetricof}[2]{\andmetric(#1, #2)}

\newcommand{\mulmetric}{\ametric_{\mulop}}
\newcommand{\mulmetricof}[2]{\mulmetric(#1, #2)}

\newcommand{\hdmetric}{\ametric_{\mathit{hd}}}
\newcommand{\hdmetricof}[2]{\hdmetric(#1, #2)}

\newcommand{\classof}[1]{[#1]}
\newcommand{\fullclassof}[2]{\classof{#2}_{#1}}

\newcommand{\aset}{\mathit{S}}
\newcommand{\asetx}{\mathit{X}}
\newcommand{\asety}{\mathit{Y}}
\newcommand{\agt}{\mathit{G}}
\newcommand{\anelema}{\mathit{a}}
\newcommand{\anelemb}{\mathit{b}}
\newcommand{\anelemc}{\mathit{c}}

\newcommand{\anelemy}{\mathit{y}}
\newcommand{\anelemg}{\mathit{g}}
\newcommand{\anelemgt}{\mathit{gt}}

\newcommand{\dcof}[1]{#1\mathop{\!\downarrow}}
\newcommand{\subprogrel}{\sqsubseteq}

\newcommand{\buof}[1]{\mathsf{BU}(#1)}
\newcommand{\enumof}[2]{\mathsf{Enum}_{#1}(#2)}

\newcommand{\rset}{\mathit{R}}
\newcommand{\rx}{\code{x}}
\newcommand{\ry}{\code{y}}
\newcommand{\ropx}{\anOperatorOf{\code{x}}}
\newcommand{\ropy}{\anOperatorOf{\code{y}}}
\newcommand{\aclass}{C}
\newcommand{\aclassp}{D}

\newcommand{\factorquasimetric}{\mathop{\mathit{q}}}
\newcommand{\factorquasimetricof}[2]{\factorquasimetric(#1, #2)}
\newcommand{\inducedquasimetric}{\mathop{\mathit{q}_\mathit{m}}}
\newcommand{\inducedquasimetricof}[2]{\inducedquasimetric(#1, #2)}
\newcommand{\inducedorimetric}{\mathit{m}_{\equiv, \factorquasimetric}}
\newcommand{\inducedorimetricof}[2]{\inducedorimetric(#1, #2)}
\crefname{Line}{line}{lines}

\begin{document}

\title{Oriented Metrics for Bottom-Up Enumerative Synthesis}

\author{Roland Meyer}
\email{roland.meyer@tu-braunschweig.de}
\orcid{0000-0001-8495-671X}
\author{Jakob Tepe}
\email{j.tepe@tu-braunschweig.de}
\orcid{0009-0002-8177-4675}
\affiliation{\institution{Technische Universität Braunschweig}
  \city{Braunschweig}
  \country{Germany}
}

\renewcommand{\shortauthors}{Meyer and Tepe}

\begin{abstract}
In syntax-guided synthesis, one of the challenges is to reduce the enormous size of the search space. 
We observe that most search spaces are not just flat sets of programs, but can be endowed with a structure that we call an oriented metric. 
Oriented metrics measure the distance between programs, like ordinary metrics do, but are designed for settings in which operations have an orientation. 
Our focus is on the string and the bitvector domains, where operations like concatenation and bitwise conjunction transform an input into an output in a way that is not symmetric.
We develop several new oriented metrics for these domains. 

Oriented metrics are designed for search space reduction, and we present four techniques: 
(i)~pruning the search space to a ball around the ground truth,
(ii)~factorizing the search space by an equivalence that is induced by the oriented metric, 
(iii)~abstracting the oriented metric (and hence the equivalence) and refining it, 
and 
(iv)~improving the enumeration order by learning from abstract information.
We acknowledge that these techniques are inspired by developments in the literature.
By understanding their roots in oriented metrics, we can substantially increase their applicability and efficiency.
We have integrated these techniques into a new synthesis algorithm and implemented the algorithm in a new solver. 
Notably, our solver is generic in the oriented metric over which it computes. 
We conducted experiments in the string and the bitvector domains, and consistently improve the performance over the state-of-the-art by more than an order of magnitude.
\end{abstract}
%
%

\maketitle

\section{Introduction}\label{sec:intro}
The goal of inductive program synthesis is to synthesize a program from input-output examples that are given as the specification. 
A prime example for inductive synthesis is the FlashFill~\cite{flashfill,flashmeta,flashfillpp} feature in Microsoft Excel enabling millions of end-users to automate data manipulation tasks.
Other applications include superoptimization ~\cite{lens},
program deobfuscation~\cite{syntia,qsynth},
and
synthesizing database queries~\cite{nosdaq,sql_from_examples}.
In this context, the Syntax-Guided Synthesis (\sygus{}) paradigm~\cite{sygus} received considerable attention.
Given a grammar $\agrammar$ and a specification $\aspecification$,
the goal is to find a program $\aprog \in \languageOf{\agrammar}$ that satisfies $\aspecification$. 
Numerous tools~\cite{cvc4,sygus,eusolver,euphony,probe,duet,simba,dryadsynth,synthphonia,nay,nope} have been developed to tackle \sygus{} problems.

Most successful \sygus{} solvers implement a bottom-up enumeration, which constructs larger programs from smaller ones until a solution has been found.
While conceptually similar, the algorithms differ drastically when it comes to two parameters:  
the search space of programs that are considered as possible solutions, and the enumeration order in which the search space is explored.
Most works aim to improve the enumeration order so as to find a solution quickly.  
Learning~\cite{probe,euphony} tries to understand which subprograms are likely to play a role in the solution, and therefore should be enumerated early on.
Deduction~\cite{eusolver,synthphonia,simba,duet,dryadsynth} tries to guide the bottom-up enumeration by information about the target values that has been computed top-down from program sketches.
Strategies to reduce the enormous search space have received less attention. 
The standard technique is to factorize the search space along observational equivalence~\cite{sygus,escher,transit}: 
if two programs have the same outputs on the given inputs, it suffices to keep one of them. 
The comparison has been weakened to an abstraction of the output values, and a refinement loop has been introduced to recover from imprecision and remain complete~\cite{blaze}. 

A drastically new approach to reduce the search space starts from the following consideration~\cite{symetric}. 
Observational equivalence is a discrete judgment: programs may or may not be equivalent. 
Metrics generalize this to a continuous notion of distance~\cite{metricsbook}: programs may be closer to each other or further apart.
Having a metric on the search space suggests a reduction strategy that we refer to as pruning: only consider programs that are close enough to the ground truth. 
These programs form a ball around the ground truth whose radius is the threshold on the distance.
Pruning deliberately gives up completeness and trades it for performance.
The balance between completeness and performance, however, is under the control of the user: it is the radius of the ball. 
Pruning can be combined with the aforementioned factorization techniques.

While metric search spaces and pruning are attractive conceptually, their applicability has been limited so far.
The problem is that metrics require symmetry, $\ametricof{\aninput}{\anoutput}=\ametricof{\anoutput}{\aninput}$. 
The operations on most data domains, however, are oriented.
Consider strings and concatenation.
If the string~$\aninput$ is a substring of~$\anoutput$, then $\aninput$ may help us produce $\anoutput$ by concatenation.
If $\aninput$ is a superstring, then there is no chance to produce $\anoutput$ by concatenation.
The situation is similar with bitvectors and the operations of bitwise conjunction, disjunction, and multiplication. 
In short, the symmetric metrics cannot measure in a meaningful way the impact of operations that are oriented. 

Quasimetrics~\cite{metricsbook} have been proposed as a generalization of metrics that does not require symmetry. 
With quasimetrics, we can assign a meaningful distance to strings that are manipulated by concatenation: 
if $\aninput$ is a substring of $\anoutput$, we take as distance the difference in length $|\anoutput|-|\aninput|$; if $\aninput$ is a superstring, we take infinity. 
Similar quasimetrics can be defined for the operations in the bitvector domain. 
Unfortunately, quasimetrics do not work with example-based specifications.
When two programs have the same outputs on the given inputs, 
their distance should be zero under the quasimetric (lifted to the space of input-output functions).
The definition of quasimetrics, however, requires equality for objects with distance zero.
This will not hold: the fact that programs agree on a number of inputs does not mean they agree on all inputs.

\paragraph{Contribution}
We define a new notion of \emph{oriented metrics (orimetrics)} that targets data domains whose operations are not symmetric.
Orimetrics only require reflexivity in the form of $\ametricof{\anelema}{\anelema}=0$. 
We do not require that $\ametricof{\anelema}{\anelemb}=0$ implies $\anelema=\anelemb$ as in quasimetrics.
Orimetrics only require symmetry at zero, $\ametricof{\anelema}{\anelemb}=0$ implies $\ametricof{\anelemb}{\anelema}=0$.
We do not require $\ametricof{\anelema}{\anelemb}=\ametricof{\anelemb}{\anelema}$ as in metrics and pseudometrics.
We still have the triangle inequality, $\ametricof{\anelema}{\anelemb}\leq \ametricof{\anelema}{\anelemc}+\ametricof{\anelemc}{\anelemb}$. 

We present new orimetrics for the \emph{string}, the \emph{bitvector}, and the \emph{function domain}.
We give a principled way to derive these orimetrics that should carry over to other domains as well. 
The idea is to let the orimetric measure the complexity of the inverse semantics of the operator~\cite{flashmeta,duet}. 

Orimetrics are designed for pruning and factorization in bottom-up enumerative synthesis. 
For \emph{pruning}, we still limit the search to a ball $\ballof{\radius}{\anelemgt}=\setCompact{\aprog}{\ametricof{\aprog}{\anelemgt}< \radius}$, as pioneered in~\cite{symetric}.
The difference, however, is that the direction in which we measure the distance matters. 
For \emph{factorization}, an important insight is that \emph{every orimetric induces an equivalence relation} where $\anelema\inducedequivof{\ametric}\anelemb$ if $\ametricof{\anelema}{\anelemb}=0$. 
This generalizes all factorization strategies discussed above.
Depending on the orimetric, a distance of zero may mean the programs produce the same outputs on the inputs from all examples (observational equivalence), produce the same abstract outputs on the inputs from all examples (the equivalence used by Blaze~\cite{blaze}), or produce similar outputs on some inputs. 

Factorization makes it attractive to work with an orimetric that is rough in that it equates many programs.  
At the same time, the orimetric should faithfully represent the distance to the ground truth. 
To reconcile these desiderata, we work with \emph{approximate orimetrics} $\approxmetric$ and introduce a refinement scheme.
Approximate here means that $\approxmetricof{\aprog}{\groundTruthFunc}=0$ may hold although $\aprog\neq\groundTruthFunc$. 
The reflexivity requirement for orimetrics plays a surprisingly important role for the refinement.
We can refine $\approxmetric$ to any~$\approxmetricpar{\aprog}$ with $\approxmetricparof{\aprog}{\aprog}{\groundTruthFunc}\neq 0$, and reflexivity will guarantee that the program will not be considered a candidate solution again.  
In short, we obtain a stronger factorization at the expense of potentially finding spurious programs and having to refine.

What turned out surprisingly challenging is to harmonize the factorization with the bottom-up enumeration.
To see this, consider the set $\setof{\rx, \ry, \ropy, \ropx}$ on which we have the equivalence $\rx\equiv\ry$ and $\ropx\equiv \ropy$. 
Assume the bottom-up enumeration constructs the programs in the order given by the set, and the factorization maintains the first element in each equivalence class as a representative.
Then we obtain $\setof{\rx, \ropy}$, which is not bottom-up enumerable. 
We give natural and easy to satisfy conditions under which the factorized set remains bottom-up enumerable. 

We implemented our approach in a tool called \toolname{}\footnote{\url{https://github.com/J4K0B/Merlin}} and evaluated its performance against state-of-the-art SyGuS tools and domain-specific solvers. 
In the bitvector domain, \toolname{} is 27 times faster than DryadSynth~\cite{dryadsynth}, the current best SyGuS solver on the bitvector domain.
On Blaze's~\cite{blaze} string benchmarks, it is 75 times faster than Blaze.
Overall, \toolname{} is 42 times faster than a baseline implementation without our proposed techniques.
We give a careful evaluation of the impact of pruning and refinement.

In total, we make three contributions:
\begin{enumerate}
    \item We define orimetrics for bottom-up enumerative synthesis.
    Orimetrics allow us to prune, factorize, and refine the search space (\Cref{sec:overview,sec:preliminaries,sec:framework}).
    \item We define orimetrics for the string, the bitvector, and the function domain. 
    Furthermore, we present a principled way of coming up with orimetrics (\Cref{Section:Instantiation}).
    \item We implemented our approach in a tool called \toolname{} and compared it to the state-of-the-art in the string and the bitvector domain. We win by over an order of magnitude (\Cref{sec:evaluation}).
\end{enumerate}

 \section{Overview}\label{sec:overview}
A SyGuS problem takes as input a specification in the form of a function $\groundTruthFunc$ and a grammar for programs $\agrammar$. 
The task is to find a program $\aprog\in\completeLanguageOf{\agrammar}$ that implements the function, $\semanticsOf{\aprog}=\groundTruthFunc$. 
We call~$\groundTruthFunc$ the ground truth and $\aprog$ a solution to the synthesis task. 
For simplicity, we assume $\groundTruthFunc$ is given as a finite set of input-output examples, but remark that our techniques carry over to more elaborate settings.  
There are various strategies of how to solve a SyGuS problem. 
The most successful solvers implement a form of bottom-up enumeration~\cite{blaze,sygus,eusolver,duet,simba,dryadsynth,synthphonia}, where they try to find a solution by composing subprograms that have already been constructed. 
While every new solver contributes a new technique that makes it faster than the state-of-the-art, there are two parameters that play a role in all solvers. 
The \emph{search space~$\aprogSet$} contains the programs that are considered relevant to solve the synthesis task (they may form solutions or occur as subprograms in solutions). 
The \emph{enumeration order~$\enumerationOrder{}$} defines the order  in which the search space should be explored. 

Our first insight is that all search spaces of practical interest are not unstructured sets, but can be endowed with an
\emph{oriented metric}
(orimetric) $\functionDistanceMeasure$ that gives information about the distance between programs. 
The purpose is to deal with data domains whose operations have an orientation.
Consider a grammar that supports $\concatOp$, the concatenation of strings.
If $\aninput$ is a substring of $\anoutput$, it is easy to find a string $\aninput'$ so that $\concatOpOf{\aninput}{\aninput'} = \anoutput$.
This means $\functionDistanceMeasureOf{\aninput}{\anoutput}$ should be small.
With $\anoutput$ being a superstring of $\aninput$, however, it is impossible to find a string $\anoutput'$ so that $\concatOpOf{\anoutput}{\anoutput'} = \aninput$.
The distance $\functionDistanceMeasureOf{\anoutput}{\aninput}$ should be infinity.
These considerations prompted us to drop the symmetry requirement in metrics. 
The resulting object is an orimetric.

We present four techniques that capitalize on the information given by the orimetric to improve the efficiency of bottom-up enumerative synthesis. 
A remarkable aspect is that our techniques are generic: they only refer to the orimetric and the enumeration order, but do not make assumptions on how these are defined.
This makes it possible to use the four techniques as enhancements in virtually any bottom-up enumerative solver. 

Our first technique is called {\bf pruning}. 
Pruning limits the search to a ball around the ground truth, meaning it tries to build a solution solely from the programs contained in this ball.
Technically, the ball is the set of programs whose distance to the ground truth is smaller than a threshold $\radius$: 
\begin{align*}
\ballfullof{(\aprogSet, \functionDistanceMeasure)}{\radius}{\groundTruthFunc}\quad =\quad \setCollector{\aprog\in\aprogSet}{\functionDistanceMeasureOf{\semanticsOf{\aprog}}{\groundTruthFunc}< \radius}\ .
\end{align*}
We refer to $\radius$ as the radius of the ball, and just write $\ballof{\radius}{\groundTruthFunc}$ when the orimetric (search) space is understood. 
The orimetric is required to have the mathematical properties described in the introduction (which will be made formal in \Cref{sec:framework}). 
What is surprisingly important in the context of SyGuS is reflexivity of the orimetric:
\begin{align}
\semanticsOf{\aprog} = \groundTruthFunc \quad\Rightarrow\quad 
\functionDistanceMeasureOf{\semanticsOf{\aprog}}{\groundTruthFunc}= 0\ .\tag{reflexivity}\label{Implication:Sound}
\end{align}
When read in contraposition, reflexivity says that only programs at distance zero to the ground truth can solve the synthesis task. 
This, however, does not mean we can just work with a ball of radius almost zero.  
The purpose of the ball is to constrain the subprograms that can be used to build candidate solutions. 
Reflexivity then applies to the candidate programs, but not to the subprograms.

The second insight is that the use of an orimetric not only allows us to limit the search to a ball, it also suggests identifying and removing duplicate elements from this ball. 
We call this {\bf factorization}.   
Indeed, from the perspective of the orimetric, two programs are equivalent whenever they have a distance of zero.
The orimetric thus induces the equivalence
\begin{align*}
\aprog_1\inducedequivof{\functionDistanceMeasure} \aprog_2\quad\text{if}
\quad \functionDistanceMeasureOf{\semanticsOf{\aprog_1}}{\semanticsOf{\aprog_2}}=0\ .
\end{align*}
We remove duplicates by factorizing the ball along this equivalence, that is, searching $\factorize{\ballof{\radius}{\groundTruthFunc}}{\inducedequivof{\functionDistanceMeasure}}$.
In our implementation, we remove duplicates by considering representatives of equivalence classes.  
Concretely, we use as representatives the minimal elements wrt.~the enumeration order $\enumerationOrder{}$. 
The reader will observe the similarity between factorization and observational equivalence~\cite{transit,escher}.
There, programs are considered equivalent when they return the same outputs on the inputs from all examples, in which case one program will be discarded.
This significantly reduces the search space while leaving the search complete.
We emphasize that orimetrics are only able to capture factorization because they do not require objects to be equal when they have a distance of zero.

As discussed in the introduction, it would be attractive to have a coarse equivalence that equates many programs.  
Our third technique is to work with \emph{approximate orimetrics} $\approxmetric$.  
We say that an orimetric is \emph{precise}, if the converse of reflexivity holds:
$\approxmetricof{\semanticsOf{\aprog}}{\groundTruthFunc}=0$ implies $\semanticsOf{\aprog}=\groundTruthFunc$.
Otherwise, the orimetric is called \emph{approximate}. 
In approximate orimetrics, we may have $\approxmetricof{\semanticsOf{\aprog}}{\groundTruthFunc}=0$ although $\semanticsOf{\aprog}\neq\groundTruthFunc$. 
We thus have to explicitly check whether a program solves the synthesis task. 
If this is not the case --- we call the program spurious --- we have a {\bf refinement} function that constructs from $\approxmetric$, $\aprog$, and $\groundTruthFunc$ a new orimetric $\approxmetricpar{\aprog}$.
The idea is inspired by and can be combined with the abstraction-refinement approach from \cite{blaze}. 
That work uses predicate abstraction on the output to equate programs.

Our fourth technique modifies the enumeration order on-the-fly by {\bf learning} from spurious programs.
When we find a spurious program $\aprog$, we know that $\aprog$ was promising from the perspective of the previous orimetric.  
While it may not be a solution, $\aprog$ is likely to contain valuable subprograms.
We therefore update the enumeration order to list these subprograms early on.
To be precise, we only list them if they still reside in the ball that is formed with a precise orimetric $\functionDistanceMeasure$. 
It is useful to have a precise orimetric at this point to eliminate as many programs as possible. 
The approach is inspired by \cite{probe}, where operators are preferred that occur frequently in solutions to single examples.

Another approach to modify the enumeration order are deductive methods \cite{duet,simba,dryadsynth,synthphonia, eusolver}. 
We discuss them next, to make clear that they are orthogonal to the idea of using orimetrics and that the approaches can be combined. 
One deductive method is case-splitting~\cite{eusolver}. 
When we enumerated a program for each of the given examples, we can build a decision tree to synthesize an if-then-else program that solves all examples.
Another deductive method
\cite{duet,simba,dryadsynth,synthphonia}
uses the inverse semantics of an operator.
Consider the program sketch $\concatOpOf{\sketchhole}{\sketchhole}$ and the output "\code{POPL}".
Assume we already enumerated a program $\aprog_1$ that outputs "\code{PO}". 
Given "\code{POPL}" and "\code{PO}", the inverse semantics is $\concatOp^{-1}("\code{POPL}", "\code{PO}")=\setof{"\code{PL}"}$. 
We check whether we already enumerated another program $\aprog_2$ that outputs "\code{PL}".
If so, we can complete the sketch to $\concatOpOf{\aprog_1}{\aprog_2}$ and directly solve the synthesis problem.
So instead of having to enumerate all programs up to the size of $\concatOpOf{\aprog_1}{\aprog_2}$, we only need to enumerate programs until we find $\aprog_1$ and $\aprog_2$.

We can leverage our understanding of an operator's inverse semantics to construct orimetrics. 
The idea is that the distance between an input and an output value should be inverse proportional to what may be considered the complexity of the inverse semantics. 
Continuing on the above example, if $\aninput$ is a superstring of "\code{POPL}", the inverse semantics is $\concatOp^{-1}("\code{POPL}", \aninput)=\emptyset$ and hence the distance $\functionDistanceMeasureOf{\aninput}{"\code{POPL}"}$ should be infinity. 
As a less extreme case, a prefix "\code{POP}" would leave us with $\concatOp^{-1}("\code{POPL}", "\code{POP}")=\setof{"\code{L}"}$, which is likely easier to generate than $"\code{PL}"$, 
and therefore $\functionDistanceMeasureOf{"\code{POP}"}{"\code{POPL}"}<\functionDistanceMeasureOf{"\code{PO}"}{"\code{POPL}"}$ should hold. 

\Cref{fig:cegar} presents a generic bottom-up enumerative synthesis algorithm that incorporates our four techniques.
Note that the algorithm is parametric in the enumeration order and in the initial orimetric.
This means a tool will not have to hard code the enumeration order and the orimetric, but can take them as input, together with the grammar
The algorithm implements the counterexample-guided abstraction refinement loop~\cite[Chapter 13]{cegar} explained above.
The input is the SyGuS task of interest.
In the first step, we prune the search space to a ball $\aprogSet$ around the ground truth.
Apart from the initial search space $\completeLanguageOf{\agrammar}$, $\pruneFunc$ takes as input the orimetric $\approxmetric$ and $\groundTruthFunc$. 
In the second step, we factorize this ball and obtain  $\aprogSetp$.
In the third step, we search $\aprogSetp$ in the order~$\enumerationOrder{}$ for a program satisfying the specification. 
If we find a program $\aprog$ that $\approxmetric$ believes satisfies the specification, $\approxmetricof{\semanticsOf{\aprog}}{\groundTruthFunc} = 0$, we hand it over to the next step. 
Since $\approxmetric$ is approximate, 
we have to check whether $\aprog$ solves the SyGuS instance.
If so, the loop stops and returns $\aprog$.
If $\aprog$ turns out to be spurious, we pass it to $\refineFunc$. 
Using also $\approxmetric$ and $\groundTruthFunc$, the refinement determines an updated approximate orimetric $\approxmetricpar{\aprog}$ with $\approxmetricparof{\aprog}{\semanticsOf{\aprog}}{\groundTruthFunc} \neq 0$. 
In the final step, we update the enumeration order using function $\updateEnumOrder$.
It takes as input the current enumeration order $\enumerationOrder{}$,
all programs enumerated by $\findProgFunc$, and the ground truth $\groundTruthFunc$.
By analyzing the enumerated programs, it learns a better enumeration order for the next iteration. 
Then the loop repeats.

\paragraph{Remarks}
The above description is conceptual in that it separates the functions more than we do in our implementation.
The workhorse of our implementation is the search for a candidate solution. 
Pruning and factorization run interleaved with it.
Concretely, $\findProgFunc$ constructs the programs one by one: given a list of programs that have already been constructed, it is able to determine the program $\aprog$ that should be constructed next according to the enumeration order $\enumerationOrder{}$. 
If $\aprog$ does not belong to the ball of interest, $\approxmetricof{\semanticsOf{\aprog}}{\groundTruthFunc}\geq \radius$, it is discarded. 
The same holds if we already have another representative $\aprogp$ in the list, meaning 
$\approxmetricof{\semanticsOf{\aprog}}{\semanticsOf{\aprogp}}=0$. 
We already know $\aprogp\enumerationOrder{}\aprog$, and thus $\aprogp$ should be the representative.  
If the program passes these tests, we append it to the list. 

Why do we keep the enumeration order?
An alternative would be to just have an orimetric and imitate the enumeration order by going through the programs in the order of their distance to the ground truth. 
First, we believe the enumeration order is such an integral part of the solving process that it deserves being a parameter on its own. 
Second, our metrics are very coarse: they are made to define the ball but do not distinguish much between the programs inside the ball. 
This will become clear in the next section where we illustrate our approach on an example. 
Third, the orimetrics are not as flexible as the enumeration order. 
An orimetric has to satisfy a few mathematical properties, whereas the enumeration order just has to be a total order on the search space.
This flexibility makes it easier to adapt the enumeration order based on learned information.

One may also ask whether the concept of orimetrics is needed after all, or whether we could have based our algorithmic improvements on more elementary mathematical notions.
One could try to prune the search space with a quasimetric
and factorize the search space with an equivalence. 
Quasimetrics have the problem that a distance of zero should imply equality of the elements, which does not hold in example-based settings. 
The elements are, however, observationally equivalent. 
The discussion suggests we could endow the search space $\aprogSet$ with an equivalence $\equiv\ \subseteq\aprogSet\times \aprogSet$ that should be used for factorization and with a quasimetric on the now factorized space $\factorquasimetric \subseteq \factorize{\aprogSet}{\equiv}{} \times{} \factorize{\aprogSet}{\equiv}$ that could be used for pruning.
To our surprise, this alternative definition $(\aprogSet, \equiv, \factorquasimetric)$ is equivalent to our notion of orimetric search spaces $(\aprogSet, \functionDistanceMeasure)$:  a combination of an equivalence and a quasimetric is enough to induce an orimetric, and vice versa. 
At the same time, the alternative definition has disadvantages that orimetrics overcome:
(i) working with equivalence classes is cumbersome, we believe the symmetry at zero requirement for orimetrics is simpler,
(ii) one has to understand refinement for two objects, the equivalence and the quasimetric, and maintain both objects during computation, 
(iii) there is no guidance on how to obtain the equivalence, while it is a derived concept for orimetrics.
All this indicates that the notion of orimetric search spaces is somewhat fundamental to bottom-up enumerative synthesis.

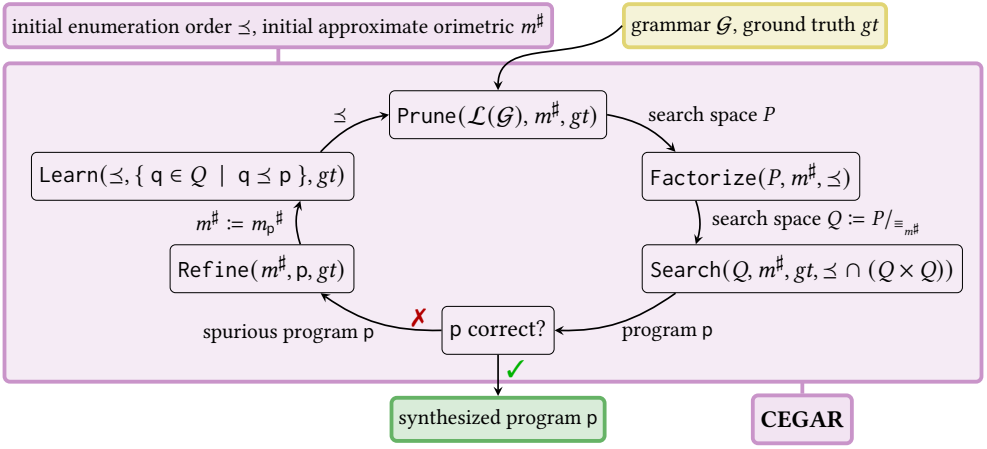
\begin{figure}
    \centering
    \scalebox{0.9}{
    \begin{tikzpicture}
        \tikzset{
rc/.style={rounded corners=0.8mm,line width=0.7pt},
place/.style={draw, rounded corners=0.8mm, minimum width=0.7cm, minimum height=0.7cm, scale=1},
          label/.style={draw, rounded corners=0.8mm, ultra thick, minimum width=0.7cm, minimum height=0.7cm, scale=1},
          inExNode/.style={draw, rounded corners=0.8mm, ultra thick, minimum width=0.7cm, minimum height=0.7cm, scale=0.9},
        }
\coordinate (offset) at (0.5cm, 0cm);
        \coordinate (minioffset) at (0.3cm, 0cm);
\draw [ultra thick, draw=violet!70!white, fill=violet!85!white!15, opacity=0.6, rounded corners=0.8mm]
        (12.15,-0.25) -- (-2.3, -0.25) -- (-2.3,-4.9) -- (12.15,-4.9) -- cycle;

\node[inExNode, draw=violet!70!white!60, fill=violet!70!white!15] (0) at (1.75, 0.3) {initial enumeration order $\enumerationOrder{}$, initial approximate orimetric $\approxmetric$};
        \node[inExNode, draw=olive!50!yellow!60, fill=olive!50!yellow!15] (8) at (8.8, 0.3) {grammar $\agrammar$, ground truth $\groundTruthFunc$};
        \node[inExNode, draw=green!50!black!60, fill=green!50!black!15] (f) at (5,-5.45) {synthesized program $\aprog$};

\node[place] (prune) at (5,-1) {$\pruneFuncOf{\completeLanguageOf{\agrammar}}{\approxmetric}{\groundTruthFunc}$};

        \node[place] (findprog) at (9.5,-3.25) {$\findProgFuncOf{\aprogSetp}{\approxmetric}{\groundTruthFunc}{\enumerationOrder{} \cap \ (\aprogSetp \times \aprogSetp)}$};

        \node[place, above=1.35 of findprog.west, anchor=west] (factorize) {$\factorizeFuncOf{\aprogSet}{\approxmetric}{\enumerationOrder{}}$};

        \node[place, left=9 of findprog.east] (refine) {$\refineFuncOf{\approxmetric}{\aprog}{\groundTruthFunc}$};

        \node[place, above=1.35 of refine.east, anchor=east] (learn) {$\updateEnumOrderOf{\enumerationOrder{}}{\setCollector{\aprogp \in \aprogSetp}{\aprogp \enumerationOrder{} \aprog}}{\groundTruthFunc}$};

        \node[place] (correct) at (5,-4.15) {$\aprog$ correct?};
      
\node[label, draw=violet!70!white!60, fill=violet!70!white!15] (active) at (9.5, -5.45) {\textbf{CEGAR}};
        \draw [ultra thick, draw=violet!70!white, opacity=0.6]
        (active) -- (9.5, -4.925);
        \draw [ultra thick, draw=violet!70!white, opacity=0.6]
        (0) -- (1.75,-0.225);

        \path[-stealth,rc] (8.west) edge[bend left = 20, in=242.7]  
        (prune.north);

        \path[-stealth,rc] (correct) edge  node[scale=1]
        [right, yshift=3pt] {\color{green!70!black}\mycheckmark} (f);

        \path[-stealth,rc] (prune.east) edge[out=315, bend left=18]  node[scale=0.9][right, yshift=5pt, xshift=-2pt] {search space $\aprogSet$} ($(factorize.north west) + (offset)$);

        \path[-stealth,rc] ($(factorize.south west) + (offset) + (minioffset)$) edge[bend left=18]  node[scale=0.9][right, xshift=2pt] {search space $\aprogSetp \coloneq \factorize{\aprogSet}{\inducedequivof{\approxmetric}}$} ($(findprog.north west) + (offset) + (minioffset)$);

        \path[-stealth,rc] ($(findprog.south west) + (offset)$) edge[out=225, bend left=18]  node[scale=0.9][right, yshift=-5pt, xshift=-2pt] {program $\aprog$} (correct.east);

        \path[-stealth,rc] (correct.west) edge[out=135, bend left=18]  
        node[scale=0.9][left, yshift=-5pt, xshift=2pt] {spurious program $\aprog$} 
        node[scale=1.05][right, yshift=2pt, xshift=10pt] {\color{red!70!black}\myxmark} 
        ($(refine.south east) - (offset)$);

        \path[-stealth,rc] ($(refine.north east) - (offset) - (minioffset)$) edge[bend left=18]  node[scale=0.9][left, xshift=-2pt] {$\approxmetric \coloneq \approxmetricpar{\aprog}$} ($(learn.south east) - (offset) - (minioffset)$);

        \path[-stealth,rc] ($(learn.north east) - (offset)$) edge[out=45, bend left=18]  node[scale=0.9][left, yshift=5pt, xshift=2pt] {$\enumerationOrder{}$} ($(prune.west)$);
    \end{tikzpicture}
    }
    \caption{CEGAR loop for synthesis.}
    \label{fig:cegar}
\end{figure}

\subsection{Example}\label{sec:overview:example}

\begin{figure}
\begin{minipage}[b]{0.45\textwidth}
    \small
\centering
\begin{align*}
    \startSymbol \quad &::= \quad 
    \mathit{Init} \bnf 
    \replaceOpOf{\startSymbol}{\startSymbol}{\startSymbol} \bnf
    \concatOpOf{\startSymbol}{\startSymbol}
    \\
    \mathit{Init} \quad &::= \quad
    \avar \bnf \emptystring \bnf \constConf{} \bnf \constCity{}
\end{align*}
\caption{Example grammar $\agrammar$.}\label{fig:ex_overview_grammar}
\end{minipage}
\hfill
\begin{minipage}[b]{0.45\textwidth}
\centering
    \small
\begin{tabular}{|c|c|c|} 
    \hline
    $n$ & Input & Output \\
    \hline
    \hline
    1 & "\code{POPL}\textvisiblespace\code{Conference}" & "\code{POPL}" \\
    \hline
    2 & "\code{Rennes}\textvisiblespace\code{City}" & "\code{Rennes}" \\
    \hline
    3 & "\code{PLDI}\textvisiblespace\code{Conference}" & "\code{PLDI}" \\
    \hline
    4 & "\code{Seoul}\textvisiblespace\code{City}" & "\code{Seoul}" \\
    \hline
\end{tabular}
\caption{Input-output examples $\ioexamples$.}\label{fig:ex_overview_examples}
\end{minipage}
\end{figure}

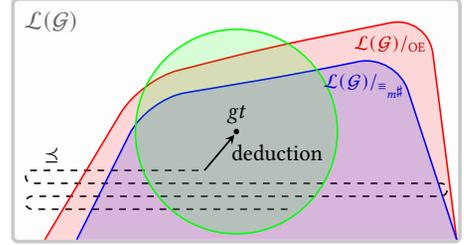
\begin{wrapfigure}{R}{0.45\textwidth}
    \centering
    \scalebox{0.85}{
    \begin{tikzpicture}
        \tikzset{
myline/.style={rounded corners=6mm,line width=0.7pt, fill opacity=0.15},
          enumorder/.style={rounded corners=1mm,line width=0.7pt, dashed},
          deduction/.style={line width=0.8pt},
        }

\draw [ultra thick, draw=gray!70!white, fill=gray!85!white!2, opacity=0.6, rounded corners=0.8mm]
        (6,-0.25) -- (-1, -0.25) -- (-1,-4) -- (6,-4) -- cycle;

        \draw[enumorder](2.5,-2.3) (3.3, -3.5) -- (-0.8, -3.5) -- (-0.8, -3.3) -- (5.8, -3.3) -- (5.8, -3.1) -- (-0.8, -3.1) -- (-0.8, -2.9) -- (2, -2.9) node (b) [right, yshift=6pt, xshift=-75pt]{$\enumerationOrder{}$};

        \draw[myline, draw=red!100!white, fill=red] (5.95,-3.97) -- (5.5,-0.5) -- (3,-1) -- (1,-1.5) -- (-0.5,-3.97);

        \draw[myline, draw=blue!100!white, fill=blue] (5.95,-3.97) -- (5,-1.1) -- (3.5,-1.4) -- (1.1,-1.8) -- (0,-3.97);

        \draw[myline, draw=green!100!white, fill=green](2.5,-2.3) circle (45pt);

        \draw[fill](2.5,-2.3) circle(1pt) node (gt) [above]{$\groundTruthFunc$};

        \draw[-stealth, deduction] (2, -2.9) -- (2.47, -2.35) node (a) [right, yshift=-8pt, xshift=-5pt]{deduction};

        \node[] at (4.9, -0.95) {\color{red!80!black}{\footnotesize $\factorize{\completeLanguageOf{\agrammar}}{\text{OE}}$}};

        \node[] at (4.5, -1.6) {\color{blue!80!black}{\footnotesize $\factorize{\completeLanguageOf{\agrammar}}{\inducedequivof{\approxmetric}}$}};

        \node[] at (-0.4, -0.6) {\color{gray!80!black}$\completeLanguageOf{\agrammar}$};
    \end{tikzpicture}
    }
    \caption{Search space and enumeration order.
    }
    \label{fig:search_space}
\end{wrapfigure}
 
We illustrate our algorithm by solving an example problem in the string domain.
We will repeatedly refer to \Cref{fig:search_space} to make the link to the conceptual development introduced above. 
\Cref{fig:ex_overview_grammar} depicts a context-free grammar $\agrammar$ and \Cref{fig:ex_overview_examples} shows $\groundTruthFunc$ as four input-output examples~$\ioexamples$. 
The nonterminal $\mathit{Init}$ can be rewritten to the input variable $\avar$ or to string constants, where $\emptystring$ is the empty string. 
The operator $\replaceOp$ takes as arguments three strings: the first is the string in which the replacement should happen, namely the first occurrence of the second argument, should it exist, will be replaced by the third argument. 
The operator $\concatOp$ concatenates the strings that are given as arguments. 

In \Cref{fig:ex:execution}, we show the set $\completeLanguageOf{\agrammar}$ of all programs as a collection of sets $\aprogSet_i$ that contain all programs of size $i$.
The set $\aprogSet_1$ contains all programs that can be derived from $\mathit{Init}$. 
Programs of size greater than one are constructed by combining programs of smaller size. 
We use an enumeration order $\enumerationOrder{}$ that orders programs by their size.
Programs that have the same size are ordered from left to right in \Cref{fig:ex:execution}. 
The set $\aprogSet_7$ contains the solution to the synthesis task, highlighted in \successColor{\text{yellow}}. 
The first row in \Cref{fig:ex:sizes} gives the cardinality of each set $\aprogSet_i$. 
Note how the number of programs grows exponentially with the size, and so eliminating programs early on is essential.
In \Cref{fig:search_space}, the set of programs is represented by the gray box and the enumeration order by the dashed line. 

State-of-the-art bottom-up enumerative SyGuS solvers~\cite{blaze,sygus,eusolver,duet,simba,dryadsynth,synthphonia} factorize the search space, often using observational equivalence~\cite{transit,escher}.
In \Cref{fig:ex:execution}, programs that are pruned because they are observationally equivalent to a program that was enumerated earlier are highlighted in \oepruneColor{\text{red}}.
The cardinality of the factorized sets is given in the second row in \Cref{fig:ex:sizes}.
In \Cref{fig:search_space}, this search space is represented by the red area.

Using orimetrics, we can further reduce the search space.
The first step is to define an orimetric on programs. 
The orimetric is chosen based on the data and data manipulations that should be supported.
For strings with replacement, it is beneficial to reward programs that produce superstrings of the outputs given in the examples and punish programs that do not.
We explain in a moment how this plays together with the fact that concatenation prefers substrings.  
We begin by defining an auxiliary quasimetric $\distanceMeasure$ on strings:
\begin{equation*}
    \distanceMeasureOf{\aninput}{\anoutput}\ =\  
    \begin{cases}
        \lengthOf{\aninput} - \lengthOf{\anoutput}  &\ \text{if $\aninput$ is a superstring of $\anoutput$}
        \\
        100 + \absOf{\lengthOf{\aninput} - \lengthOf{\anoutput}}  & \text{otherwise}
    \ .
    \end{cases}
\end{equation*}
For example, $\distanceMeasureOf{"\code{PO}"}{"\code{POPL}"} = 102$ but $\distanceMeasureOf{"\code{POPL}"}{"\code{PO}"} = 2$. 
Note that~$\distanceMeasure$ is not symmetric in general, but it is symmetric at distance zero.
We now lift the auxiliary quasimetric~$\distanceMeasure$ on strings to an orimetric $\functionDistanceMeasure$ on functions over strings.
The definition is as expected, we evaluate the given functions on the inputs from all examples and sum up the distances of the output values.
To make this formal, let $\inputexamplesp\subseteq \inputexamples$ and define 
\begin{equation*}
    \distanceMeasureForExamplesOf{\inputexamplesp}{\aFunc}{\aFuncp}\
    =\ 
    \sum_{\aninput \in \inputexamplesp} \distanceMeasureOf{\aFuncOf{\aninput}}{\aFuncpOf{\aninput}}
    \ .
\end{equation*}
Our orimetric is then $\functionDistanceMeasure=\distanceMeasureForExamples{\inputexamples}$.
We will use other instantiations of $\inputexamplesp$ in a moment. 
Note, that the lifting does not result in a quasimetric.
If two different programs $\aprog$ and $\aprogp$ produce the same outputs on the given inputs, their distance is $0$.
A quasimetric would require $\aprog$ and $\aprogp$ to be equal.
This need not be the case.
In fact, allowing $\aprog$ and $\aprogp$ to have a distance of zero while being different is what enables factorization in the first place. 

With the orimetric at hand, we restrict the search to programs in the ball $\ballof{\radius}{\groundTruthFunc}$, with the radius $\radius$ set to $100$.
This means we only keep programs that return a superstring of the output for every input example. 
We admit that this is discrete, and does not make use of the fact that an orimetric yields continuous values in $\realsgtz$. 
In our experiments, we will see more elaborate instantiations. 

Note that focusing on superstrings does not mean we are not allowed to use the $\concatOp$ operator.  
We can still use $\concatOp$, but on the programs inside the ball. 
The superstring ball is even closed under concatenation: concatenating superstrings yields a superstring. 
One may ask how this relates to the above argument that $\concatOp$ prefers substrings. 
This argument was made for a comparison with the final value. 
Since we use superstrings here, it means the final value cannot be produced by concatenation, and so $\concatOp$ will not be the topmost operator in the solution to the synthesis task. 
It can, however, still be a part of it.

Coming back to \Cref{fig:ex:execution}, the programs that are additionally pruned by this orimetric are highlighted in \metricpruneColor{$\text{green}$}.
\Cref{fig:search_space} depicts the resulting search space as the green ball around $\groundTruthFunc$. 
The cardinality of the program sets is shown in \Cref{fig:ex:sizes}.
Using only pruning is not as effective as factorizing along observational equivalence, but combined we are only left with $81$ programs to explore in $\aprogSet_7$.
In \Cref{fig:search_space}, this search space is the intersection of the green ball and the red area.

\begin{figure}
    \small
 \begin{align*}
    \aprogSet_1 &= 
    \setof{
        \avar, \emptystring, 
        \constConf{},
        \constCity{}
    }
    \\
    \aprogSet_2 &= 
    \emptyset 
    \\
    \aprogSet_3 &= 
    \setof{
        \infixconcat{\avar}{\avar},
        \oepruneColor{\infixconcat{\avar}{\emptystring}},
        \infixconcat{\avar}{\constConf{}},
        \infixconcat{\avar}{\constCity{}},
        \oepruneColor{\infixconcat{\emptystring}{\avar}},
        \oepruneColor{\infixconcat{\emptystring}{\emptystring}},
        \oepruneColor{\infixconcat{\emptystring}{\constConf{}}},
        \oepruneColor{\infixconcat{\emptystring}{\constCity{}}},
        \\
        & \qquad \quad \ \,
        \infixconcat{\constConf{}}{\avar},
        \oepruneColor{\infixconcat{\constConf{}}{\emptystring}},
        \metricpruneColor{\infixconcat{\constConf{}}{\constConf{}}},
        \\
        & \qquad \quad \ \,
        \metricpruneColor{\infixconcat{\constConf{}}{\constCity{}}},
        \infixconcat{\constCity{}}{\avar},
        \oepruneColor{\infixconcat{\constCity{}}{\emptystring}},
        \metricpruneColor{\infixconcat{\constCity{}}{\constConf{}}},
        \\
        & \qquad \quad \ \,
        \metricpruneColor{\infixconcat{\constCity{}}{\constCity{}}}
    }
    \\
    \aprogSet_4 &=
    \setof{
        \ldots,
        \oepruneColor{\shortReplaceOf{\avar}{\avar}{\constCity{}}},
        \oepruneColor{\shortReplaceOf{\avar}{\emptystring}{\emptystring}},
        \oepruneReuseColor{\shortReplaceOf{\avar}{\constCity{}}{\constConf{}}},
        \ldots,
        \semisuccessColor{\shortReplaceOf{\avar}{\constConf{}}{\emptystring}},
        \ldots
    }
    \\
    \aprogSet_5 &= \setof{\ldots}
    \qquad
    \aprogSet_6 = \setof{\ldots}
    \qquad
    \aprogSet_7 = \setof{
        \ldots,
        \successColor{\shortReplaceOf{\shortReplaceOf{\avar}{\constConf{}}{\emptystring}}{\constCity{}}{\emptystring}}
        ,
        \ldots
        }
\end{align*}
    \caption{Example execution of a bottom-up enumeration algorithm
    with highlights for 
    \oepruneColor{\text{OE factorization}},
    \oepruneReuseColor{\text{OE factorization with learning}},
    \metricpruneColor{\text{orimetric pruning}},
    \semisuccessColor{\text{partial correctness}},
    and 
    \successColor{\text{correctness}}.
    The $\concatOp$ operator is abbreviated by an infix "$\infixconcat{}{}$"
    and 
    the $\replaceOp$ operator is abbreviated by $\shortReplace$.
    }\label{fig:ex:execution}
\end{figure}

\begin{figure}
    \small
 \begin{tabular}{|c|c|c|c|c|c|c|c|} 
    \hline
    Method & $\aprogSet_1$ & $\aprogSet_2$ & $\aprogSet_3$ &$\aprogSet_4$ &$\aprogSet_5$ &$\aprogSet_6$ &$\aprogSet_7$ \\
    \hline
    \hline
    No Pruning or Factorization & 4 & - & 16 & 64 & 128 & 1280 & 4352 \\
    \hline
    OE Factorization & 4 & - & 9 & 6 & 27 & 56 & 119 \\
    \hline
    Orimetric Pruning (OP) & 4 & - & 7 & 18 & 56 & 323 & 929 \\
    \hline
    OE Factorization + OP & 4 & - & 5 & 6 & 19 & 50 & 81 \\
    \hline
    $\inducedequivof{\approxmetric}$ Factorization + OP + Learning & 4 / 5 & - & 5 / 12 & 3 / 10 & - & - & - \\
    \hline
\end{tabular}   
\caption{Number of programs at each size per solving method.}\label{fig:ex:sizes}
\end{figure}

\begin{figure}
    \small
 \begin{align*}
    \aprogSet_1 &= 
    \setof{
        \shortReplaceOf{\avar}{\constConf{}}{\emptystring},
        \avar, \emptystring, 
        \constConf{},
        \constCity{}
    }
    \\
    \aprogSet_2 &= 
    \emptyset 
    \\
    \aprogSet_3 &= 
    \setof{
        \infixconcat{\avar}{\avar},
        \oepruneColor{\infixconcat{\avar}{\emptystring}},
        \infixconcat{\avar}{\constConf{}},
        \infixconcat{\avar}{\constCity{}},
        \infixconcat{\avar}{\shortReplaceOf{\avar}{\constConf{}}{\emptystring}},
        \ldots,
        \oepruneColor{\infixconcat{\emptystring}{\constCity{}}},
        \ldots
    }
    \\
    \aprogSet_4 &=
    \setof{
        \ldots,
        \oepruneColor{\shortReplaceOf{\avar}{\avar}{\constCity{}}},
        \oepruneColor{\shortReplaceOf{\avar}{\emptystring}{\emptystring}},
        \ldots,
        \successColor{\shortReplaceOf{\shortReplaceOf{\avar}{\constConf{}}{\emptystring}}{\constCity{}}{\emptystring}}
        \ldots
    }
\end{align*}
    \caption{Example execution of the second iteration of our algorithm
    with highlights for 
    \oepruneColor{\text{OE factorization}}
    and 
    \successColor{\text{correctness}}.
    The $\concatOp$ operator is abbreviated by an infix "$\infixconcat{}{}$"
    and 
    the $\replaceOp$ operator is abbreviated by $\shortReplace$.
    }\label{fig:ex:execution_second_it}
\end{figure}

With approximate orimetrics $\approxmetric$, we can reduce the search space even further.  
The approximate orimetrics should be rough so that the induced equivalence $\inducedequivof{\approxmetric}$ eliminates many programs. 
We will later recover precision by adding a refinement loop. 
In our example, we use $\approxmetric = \distanceMeasureForExamples{\inputexamplesp}$ with $\inputexamplesp=\setof{\aninput_1}$.
This approximate orimetric only takes into account the first example when comparing functions, $\approxmetricof{\aFunc}{\aFuncp} = \distanceMeasureOf{\aFuncOf{\aninput_1}}{\aFuncpOf{\aninput_1}}$.  
The induced equivalence equates programs at distance zero---actually, this is why we wanted the orimetric to be symmetric at distance zero. 
Here, we equate programs that have the same output on the first input. 
To give an example, $\avar\inducedequivof{\approxmetric}\oepruneReuseColor{\shortReplaceOf{\avar}{\constCity{}}{\constConf{}}}$. 
The equivalence holds as $\constCity{}$ does not occur in "\code{POPL}\textvisiblespace\code{Conference}", and so no replacement happens. 
When we factorize the search space along $\inducedequivof{\approxmetric}$, the program in \oepruneReuseColor{\text{blue}} will be eliminated since we already have $\avar$. 
In \Cref{fig:search_space}, the factorized search space is the blue area. 
When applying $\pruneFunc$ as well as $\factorizeFunc$, we are left with the green ball intersected with the blue area.
We now use $\findProgFunc$ to find a solution candidate for the synthesis problem in this search space.

Function $\findProgFunc$ enumerates the programs along the order $\enumerationOrder{}$. 
Observe that the output of the program $\aprog$ in \semisuccessColor{\text{cyan}}, when executed on the first example,
is "\code{POPL}". 
Therefore, under $\approxmetric$, the distance to the ground truth is zero, and $\findProgFunc$ returns $\aprog$ as a candidate solution.
Notably, to find the candidate solution, we only needed to enumerate programs of size up to $4$.

Unfortunately, program $\aprog$  does not work on some of the other examples: the solution is spurious. 
For instance, on the second example the program yields 
"\code{Rennes}\textvisiblespace\code{City}" which is not equal to the specified output "\code{Rennes}". 
We use $\refineFunc$ to generate a new orimetric for the next iteration. 
To also take into account the second example, we set $\approxmetric$ to $\distanceMeasureForExamples{\setof{\aninput_1, \aninput_2}}$.
While the approximate orimetric and its refinement are simple in this example, our approach can also be instantiated with predicate abstraction as introduced in~\cite{blaze}.
We discuss the details below.

We now use $\updateEnumOrder$ to update the enumeration order for the next iteration.
To do so, we use the precise metric $\distanceMeasureForExamples{\inputexamples}$ to analyze the candidate solution $\aprog$.
We see that $\aprog$ produces a superstring of the outputs for all examples.
This means, 
$\distanceMeasureForExamplesOf{\inputexamples}{\semanticsOf{\aprog}}{\groundTruthFunc}<100$ and therefore the program lies within the ball around $\groundTruthFunc$.
We consider $\aprog$ valuable and promote it in the enumeration order.
Recall that the enumeration order in our example is defined by the size.
Function $\updateEnumOrder$ redefines the size of $\aprog$ to be $1$. 
Then $\aprog$ and programs that use $\aprog$ as a subprogram will be considered early on. 

\Cref{fig:ex:execution_second_it} gives the program sets constructed in the second iteration of the refinement loop. 
Note how $\aprog$ has size one. 
With this change in size, we already find the \successColor{\text{solution}} in $\aprogSet_4$. 
In the last row of \Cref{fig:ex:sizes}, we see the cardinalities of the sets again. 
On the left of the slash are the cardinalities for the first iteration of the refinement loop, and on the right the cardinalities for the second iteration. 
In total, we only consider $39$ programs to solve the problem. 

\paragraph{Selecting an Orimetric}
Coming up with a good orimetric is key to the success of our method.
In our example, had we defined an orimetric that rewards functions producing substrings of the output and punish those that produce superstrings,
we would have failed to generate the solution using bottom-up enumeration.

To construct good orimetrics, we derive them from the semantics of the operators in the grammar.
We already discussed how concatenation needs substrings to produce the desired output while replacement favors superstrings. 
In practice, we 
concurrently run a portfolio of solvers employing a different orimetric each. 
We developed three orimetrics for the string domain and four orimetrics for the bitvector domain.
In \Cref{sec:metrics_strings,sec:metrics_bv}, we introduce a principled way of designing orimetrics. 

\paragraph{Deduction}
An aspect the above example did not highlight is the use of deduction.
Conceptually, deduction accelerates the search by modifying the enumeration order, as depicted in \Cref{fig:search_space}.
Our tool \toolname{} implements the deduction technique from DryadSynth~\cite{dryadsynth}.
To explain it, consider the sketch $\shortReplaceOf{\sketchhole}{\sketchhole}{\sketchhole}$.  After we found the programs
$\constCity{}$, $\emptystring$, and $\shortReplaceOf{\avar}{\constConf{}}{\emptystring}$, we would enumerate 
$\shortReplaceOf{\shortReplaceOf{\avar}{\constConf{}}{\emptystring}}{\constCity{}}{\emptystring}$ next,
which solves the synthesis problem. 
Deduction thus understands  a program sketch.
This relies on the inverse semantics of a sketch.
If one finds arguments such that the sketch filled with these arguments is a solution to the synthesis problem, the solution will be enumerated next.
This means, for any supported sketch, we only need to find the correct arguments. 
Most orimetrics we define in this paper optimize the search for viable arguments of a sketch by approximating its inverse semantics. \subsection{State-of-the-Art}\label{sec:state_of_the_art}

\begin{figure}
    \small
 \begin{tabular}{|c|c|c|}
 \hline
 Solver & Search Space Pruning / Factorization & Enumeration Order (Deduction) \\ 
 \hline
 \hline
 ESolver~\cite{sygus} & Observational Equivalence (OE) & constant (\myxmark) \\
 \hline
 EUSolver~\cite{eusolver} & OE & constant (\mycheckmark) \\
 \hline
 Blaze~\cite{blaze} & OE + Abstraction Refinement + Automata & constant (\myxmark) \\
 \hline
 EUPhony~\cite{euphony} & weak OE & constant/offline learning (\mycheckmark) \\
 \hline
 Probe~\cite{probe} & OE & learning (\myxmark) \\
 \hline
 Duet~\cite{duet} & OE & constant (\mycheckmark) \\
 \hline
 Simba~\cite{simba} & OE & constant (\mycheckmark) \\
 \hline
 DryadSynth~\cite{dryadsynth} & OE & constant (\mycheckmark) \\  
 \hline
 Synthphonia~\cite{synthphonia} & OE & constant (\mycheckmark) \\
 \hline
 \textcolor{blue}{\toolname{}} & OE + Abstraction Refinement + \textcolor{blue}{Orimetrics} & learning (\mycheckmark) \\
 \hline
\end{tabular}   
\caption{Comparison of enumerative SyGuS solvers.}
\label{fig:comparison}
\end{figure}

We discuss to what extent the state-of-the-art SyGuS solvers~\cite{blaze,sygus,eusolver,euphony,duet,simba,dryadsynth,synthphonia} can be seen as instances of the generic solver in~\Cref{fig:cegar}. 
The discussion shows that rather different techniques can be understood as (i) assuming an \emph{orimetric} on the search space and (ii) improving an \emph{enumeration order} through learning. 
These two ingredients are often left implicit in the related work. 
By making them explicit and giving them rigorous definitions, we provide a framework in which SyGuS technology can be developed.

The discussion of the solvers is summarized in \Cref{fig:comparison}.
The second column describes the pruning technique that is applied to reduce the search space. 
The third column categorizes the enumeration order as constant or learning, and indicates whether deductive elements are used.

ESolver~\cite{sygus} is the basis of enumerative SyGuS solvers.  
It enumerates programs by size until it finds a program that works for all examples. 
ESolver applies observational equivalence to factorize the search space. 
The enumeration order is constant and does not use deduction.
EUSolver~\cite{eusolver} advances ESolver by adding deduction.
If for every example a suitable program has been found, EUSolver tries to construct an if-then-else program that solves all examples. 
Behind the construction is a decision tree classification of the examples.

Blaze~\cite{blaze} was the first solver based on abstraction refinement~\cite[Chapter 13]{cegar}. 
The focus is on matrix transformation and string problems, and 
Blaze expects to have access to three pieces of information about the domain: the cost for each production in the grammar, a set of predicates that can be used to abstract data values~\cite[Chapter 15]{cegar}, and an abstract semantics for each operator. 
Blaze then solves the synthesis problem in the abstract, and refines the abstract domain by adding predicates if the solution turns out to be spurious. 
Blaze can be seen as an instance of the generic algorithm in~\Cref{fig:cegar}.  
For an orimetric based on predicate abstraction, let $\aFunc^{\sharp}$ be the abstraction of~$\aFunc$ and let $\alpha(\aninput)$ be the abstraction of the input value $\aninput$. 
Both can be computed with the information assumed by Blaze.
One then defines 
\begin{equation*}
    \distanceMeasureForExamplesOf{\mathit{Blaze}}{\aFunc}{\aFuncp}\
    =\ 
    \sum_{\aninput \in \inputexamples} \distanceMeasureOf{\aFunc^{\sharp}(\alpha(\aninput))}{\aFuncp^{\sharp}(\alpha(\aninput))}
    \ .
\end{equation*}
Another novelty in Blaze is that the search space is represented and explored using automata-theoretic techniques. 
This important algorithmic aspect is not reflected in our generic solver, which is more on the semantic level.  
The enumeration order is constant, namely defined by the costs. 

EUPhony~\cite{euphony} is the first solver that uses probabilistic information. 
It translates the given grammar into a probabilistic variant, and then traverses the language using A$^*$.
These probabilities have been learned from solutions to a large collection of synthesis tasks. 
The learning is offline, it happens before the translation and the probabilities are not adapted during the search. 
The enumeration order, which is determined by A$^*$ based on the probabilities, is therefore constant in our terminology. 
EUPhony applies a weak version of observational equivalence that works on sentential forms to factorize the search space.
It adopts the deduction techniques from EUSolver to synthesize programs for benchmarks that require case-splitting.

Probe~\cite{probe} uses just-in-time learning to train the grammar.
It assigns costs to each production, and the cost of a program is then the sum of the costs of the productions needed to derive it.
Probe enumerates programs in the order of increasing costs, and stops when a cost limit is reached or a solution is found.
When the cost limit is reached, the costs are updated as follows.
Operators which occur frequently in programs solving at least one example are assigned a lower cost than the other operators. 
The technique immediately fits, and actually inspired, our refinement of the enumeration order. 
Probe applies observational equivalence factorization and does not do deduction.

Duet~\cite{duet}, Simba~\cite{simba}, DryadSynth~\cite{dryadsynth}, and Synthphonia~\cite{synthphonia} are similar when it comes to the following characteristics: they all use observational equivalence to factorize the search space, have a constant enumeration order, enumerate based on the program size (with sharing in~\cite{dryadsynth}), and use a variant of the deduction technique from EUSolver to deal with synthesis tasks that need case-splitting.
What distinguishes them are their own new deduction techniques. 

Duet~\cite{duet} solves SyGuS problems in the theories of bitvectors, Booleans, strings,  and integers. 
Its deduction is based on an inverse semantics for the operators in the grammar. 
When one argument for an operator has been fixed, the inverse semantics provides subproblems for the remaining arguments. 
These are solved by inserting an already enumerated term or by further decomposition.

Simba~\cite{simba} concentrates on bitvectors and advances Duet's deduction process. 
The deduction computes necessary preconditions for the arguments of a program sketch.
If a program satisfying the preconditions of an argument has been found, it is inserted at the corresponding place,  the preconditions for the remaining arguments are refined, and the process repeats. 

DryadSynth~\cite{dryadsynth} solves SyGuS problems in the theory of bitvectors.
DryadSynth keeps a set of sketches that are hardcoded into the algorithm.
This results in a drastic increase in performance compared to Simba.
DryadSynth enumerates programs in order of their size (with sharing).
For each enumerated program it checks whether there are other previously enumerated programs with which a sketch can be completed.
To perform this search efficiently, DryadSynth maintains viable programs for each sketch in a separate data structure.

Synthphonia solves synthesis problems in the string domain and works on more expressive grammars than specified in the SyGuS format.
Synthphonia introduces a framework to perform deduction and enumeration concurrently instead of in an interleaved fashion, and has specialized data structures for the communication between the threads.
Moreover, the case-splitting deduction technique from EUSolver is also implemented in a concurrent fashion.

The general solver we propose in \Cref{sec:framework} can be instantiated to the above solvers. While the above techniques mostly manipulate the enumeration order, our tool \toolname{} focuses on the search space (but can also update the enumeration order).
\toolname{} solves SyGuS problems in the bitvector and in the string domain.
Furthermore, \toolname{} uses abstraction refinement to factorize the search space more effectively.
\toolname{} also incorporates the deduction techniques from DryadSynth and extends them to more sketches as well as to sketches in the string domain.
The basis for these deduction techniques is to have an inverse semantics.
Orimetrics can capture the complexity of the inverse semantics and prune programs accordingly.  \section{Preliminaries}\label{sec:preliminaries}
\begin{definition}
A \emph{SyGuS problem} $(\agrammar, \aspecification)$ consists of a context-free grammar $\agrammar$ and a specification~$\aspecification$.
The task is to find a program $\aprog\in \languageOf{\agrammar}$
that satisfies the specification, $\aprog\models\aspecification$.
\end{definition}
 
We briefly recall the ingredients. 
A \emph{context-free grammar} $\agrammar = (\nonterms, \terminals, \productions, \startSymbol)$ consists of finite sets of nonterminals~$\nonterms$, terminals~$\terminals$, and
productions $\productions \subseteq \nonterms \times \kleeneOf{(\terminals \cup \nonterms)}$, together with a start nonterminal $\startSymbol \in \nonterms$. 
The term language of a nonterminal, denoted by $\termLanguageOf{\anonterm}$, consists of all words over $\terminals$ and $\nonterms$ that can be derived from $\anonterm$ using the productions. 
The language of $\anonterm$ is limited to the terminal words, 
$\languageOf{\anonterm} = \termLanguageOf{\anonterm} \cap \kleeneOf{\terminals}$. 
The language of the grammar is $\languageOf{\agrammar} = \languageOf{\startSymbol}$.
The complete language of the grammar considers all nonterminals, $\completeLanguageOf{\agrammar}= \bigcup_{\anonterm \in \nonterms} \languageOf{\anonterm}$.
Similarly, the complete term language of the grammar is  
$\completeTermLanguageOf{\agrammar}= \bigcup_{\anonterm \in \nonterms} \termLanguageOf{\anonterm}$.

In SyGuS, the terminals are variables or operators. 
Variables have arity zero, operators may have an arity greater than zero.
We expect the grammar to respect the arity, meaning the productions have to provide the expected number of arguments.
With this requirement, we can simply call terminal words programs, and write them as $\aprog\in \languageOf{\startSymbol}$.
A word over terminals and nonterminals is a program sketch $\asketch \in \completeTermLanguageOf{\agrammar}$. 
The holes of a sketch are the unresolved nonterminals.
The arity of a sketch is the number of its holes, say $n$. 
The program that results from replacing the holes by programs $\aprog_1$ to $\aprog_n$ is $\asketchOf{\aprog_1, \ldots, \aprog_{n}}$. 
We use $\subprogrel\ \subseteq \completeLanguageOf{\agrammar}\times\completeLanguageOf{\agrammar}$ for the subprogram relation where $\aprog\subprogrel\aprogp$, if there is a sketch $\asketch$ so that $\aprogp=\asketchOf{\aprog}$. 
We use $\dcof{\aprogp}=\setCompact{\aprog\in\completeLanguageOf{\agrammar}}{\aprog\subprogrel\aprogp}$ for the set of all subprograms of $\aprogp$, including $\aprogp$ itself. 
This is the downward closure wrt. $\subprogrel$. 

A domain $(\dataDomain, \anInterpretation)$ consists of a set of data values~$\dataDomain$ and an interpretation $\anInterpretation$. 
The interpretation assigns a function 
$\anInterpretationOf{\anOperator} = \dataDomain^{k} \rightarrow \dataDomain$  to each operator $\anOperator$ of arity $k$.  
The domain gives rise to a semantics for programs $\aprog\in \languageOf{\agrammar}$. 
The semantics is the function $\semanticsOf{\aprog}$ of type $\allFuncs=(\variables \rightarrow \dataDomain) \rightarrow \dataDomain$. 
It takes as input a variable assignment $\aninput:\variables \rightarrow \dataDomain$ and returns a data value.
The definition is as expected:  $\semanticsOfAppliedTo{\code{x}}{\aninput}=\aninput(\code{x})$ and $\semanticsOfAppliedTo{\anOperator(\aprog_1, \ldots, \aprog_k)}{\aninput}=\anInterpretationOfOf{\anOperator}{\semanticsOfAppliedTo{\aprog_1}{\aninput}, \ldots, \semanticsOfAppliedTo{\aprog_k}{\aninput}}$. 
It is common in SyGuS to take the domain and the semantics of programs as defined by the SMT-LIB~\cite{smtlib} standard. 
We follow this convention. 

We consider example-based specifications where $\aspecification \subseteq (\variables \rightarrow \dataDomain) \times \dataDomain$ consists of a finite set of input-output examples.  
A program satisfies the specification,  $\aprog\models \aspecification$, if  
$\semanticsOfAppliedTo{\aprog}{\aninput}=\anoutput$ for all $(\aninput, \anoutput)\in\aspecification$.
The \emph{ground truth} $\groundTruthFuncs\subseteq \allFuncs$ consists of all functions that satisfy the specification in this sense. 
\paragraph{CEGIS}
SyGuS problems that are not example based can still be solved with example-based techniques.
The idea, known as counterexample-guided inductive synthesis~\cite{cegis}, is to let an SMT solver generate new examples should a candidate program not yet satisfy the specification.
With this argument, we focus on example-based specifications.
\paragraph{Equivalences}
An equivalence $\equiv\ \subseteq \aset\times \aset$ on a set $\aset$ is a relation that is reflexive, symmetric, and transitive.
We write $\fullclassof{\equiv}{\anelema}=\setCompact{\anelemb}{\anelemb\equiv\anelema}$ for the equivalence class of $\anelema\in\aset$. 
We just write $\classof{\anelema}$ if the equivalence relation is understood. 
We lift the notation to sets $\agt\subseteq\aset$ and define $\classof{\agt}=\bigcup_{\anelemg\in\agt}\classof{\anelemg}$.
We call this the closure of $\agt$ under $\equiv$. 
The equivalence is precise wrt.\ $\agt$, if the closure does not add any elements, $\classof{\agt}\subseteq\agt$.
Note that the reverse inclusion always holds. 
If the equivalence is not precise wrt.\ $\agt$, it is called approximate. 
The equivalence is unambiguous wrt.\ $\agt$, if $\classof{\agt}=\classof{\anelemg}$ for all $\anelemg\in\agt$.  
This means all elements in $\agt$ are equivalent.
Let $\aFunc:\aset\rightarrow \aset$ be a transformer on $\aset$. 
The equivalence is a congruence wrt.\ $\aFunc$, if $\anelema\equiv \anelemb$ implies $\aFuncOf{\anelema}\equiv \aFuncOf{\anelemb}$. 
The definition generalizes to functions in several arguments in the expected way. 
Factorizing $\aset$ along~$\equiv$ yields the set of equivalence classes $\factorize{\aset}{\equiv}\ = \setCompact{\classof{\anelema}}{\anelema\in\aset}$. 
A representative system for $\factorize{\aset}{\equiv}$ is a set~$
R\subseteq \aset$ that contains precisely one element $\anelemc\in C$ for every class $C\in\factorize{\aset}{\equiv}$.   \section{Contribution \rom{1} -- Oriented Metric Search Spaces}\label{sec:framework}
We first define the main object of this paper, oriented metrics, and then turn to the search space and the enumeration order.
Finally, we describe the components of our CEGAR loop in detail. 
\subsection{Oriented Metrics}
It will be convenient to give the definition for arbitrary sets $\aset$.\begin{definition}
A function $\ametric: \aset \times \aset \rightarrow \realsgtz$ is an \emph{oriented metric (orimetric)} if, for all $\anelema, \anelemb, \anelemc\in\aset$, 
    \begin{align*}
        &\ametricof{\anelema}{\anelema}\ =\ 0 
        \tag{reflexivity}
\\
        \ametricof{\anelemb}{\anelema}\ =\ 0\ \ \Rightarrow\ \ 
        &\ametricof{\anelema}{\anelemb}\ =\ 0
        \tag{symmetry at zero}
\\
        \ametricof{\anelema}{\anelemc}\ \ \leq\ \ &\ametricof{\anelema}{\anelemb}\ +\ \ametricof{\anelemb}{\anelemc}\ .
        \tag{$\triangle$-inequality}
    \end{align*}
The \emph{equivalence induced by $\ametric$} is the relation $\inducedequivof{\ametric}\ \subseteq \aset\times\aset$ where $\anelema\inducedequivof{\ametric}\anelemb$ if $\ametricof{\anelema}{\anelemb}=0$, for all~$\anelema, \anelemb\in\aset$. 
The orimetric is \emph{precise}, \emph{approximate}, resp. \emph{unambiguous} wrt. $\agt\subseteq\aset$, if the equivalence $\inducedequivof{\ametric}$ has these properties.
The orimetric is a \emph{congruence} wrt. $\aFunc:\aset\rightarrow\aset$, if this holds for $\inducedequivof{\ametric}$.
We generalize the congruence requirement to functions in several arguments where needed.
The \emph{quasimetric induced by $\ametric$} is $\inducedquasimetric: \factorize{\aset}{\inducedequivof{\ametric}}\times \, \factorize{\aset}{\inducedequivof{\ametric}}\rightarrow  \realsgtz$ with $\inducedquasimetricof{\classof{\anelema}}{\classof{\anelemb}}=\ametricof{\anelema}{\anelemb}$. 
\end{definition}
We also suggested an alternative to orimetrics.
\begin{definition}
A \emph{factorization and pruning structure} $(\aset, \equiv, \factorquasimetric)$ consists of an equivalence relation $\equiv\ \subseteq\aset\times\aset$ and a quasimetric $\factorquasimetric: \factorize{\aset}{\equiv}\times \, \factorize{\aset}{\equiv}\rightarrow  \realsgtz$ on the factorized set. Recall that a quasimetric requires  the triangle inequality and $\factorquasimetric(\classof{\anelema}, \classof{\anelemb})=0$ if and only if $\classof{\anelema}=\classof{\anelemb}$~\cite{metricsbook}. 
The \emph{orimetric induced by the factorization and pruning structure} is $\inducedorimetric:\aset\times\aset\rightarrow\realsgtz$ with $\inducedorimetricof{\anelema}{\anelemb}=\factorquasimetricof{\classof{\anelema}}{\classof{\anelemb}}$. 
\end{definition}
The main finding is that these concepts are equivalent.
\begin{theorem}
If $(\aset, \ametric)$ is an orimetric space, then $(\aset, \inducedequivof{\ametric}, \inducedquasimetric)$ is a factorization and pruning structure.
If $(\aset, \equiv, \factorquasimetric)$ is a factorization and pruning structure, then $(\aset, \inducedorimetric)$ is an orimetric space. 
\end{theorem}
We split the proof into two lemmas.
\begin{lemma}\label[lemma]{Lemma:MainProperty}
Let $\ametric$ be an orimetric on $\aset$.
(i)\ \ $\inducedequivof{\ametric}$ is an equivalence.
(ii)\ \ The orimetric is invariant under the induced equivalence,   
 $\anelema_1\inducedequivof{\ametric}\anelema_2$ and $\anelemb_1\inducedequivof{\ametric}\anelemb_2$ imply $\ametricof{\anelema_1}{\anelemb_1}=\ametricof{\anelema_2}{\anelemb_2}$, which means the induced quasimetric is well-defined.
 (iii) $\inducedquasimetric$ is a quasimetric.  
\end{lemma}
\begin{proof}
(i) Reflexivity of the induced equivalence is by reflexivity for the orimetric. For symmetry, it suffices to have symmetry at zero. 
Transitivity is by the triangle inequality.

(ii) Recall that  $\anelema_1\inducedequivof{\ametric}\anelema_2$ means $\ametricof{\anelema_1}{\anelema_2}=0=\ametricof{\anelema_2}{\anelema_1}$. 
We can therefore calculate as follows:
\begin{align*}
&\ \ \ametricof{\anelema_1}{\anelemb_1}\\
\text{($\triangle$-inequality)}\quad \leq&\ \ \ametricof{\anelema_1}{\anelema_2}\ +\ \ametricof{\anelema_2}{\anelemb_1}\\
\text{($\ametricof{\anelema_1}{\anelema_2}=0$)}\quad =&\ \ \ametricof{\anelema_2}{\anelemb_1}\\
\text{($\triangle$-inequality)}\quad \leq&\ \ \ametricof{\anelema_2}{\anelemb_2}\ +\ \ametricof{\anelemb_2}{\anelemb_1}\\
\text{($\ametricof{\anelemb_2}{\anelemb_1}=0$)}\quad =&\ \ \ametricof{\anelema_2}{\anelemb_2}\ .
\end{align*}
Repeating the argument with the indices swapped yields $\ametricof{\anelema_2}{\anelemb_2}\leq \ametricof{\anelema_1}{\anelemb_1}$.
Together, the desired equality follows.

(iii) The triangle inequality immediately follows from the orimetric.
For equality at zero, we have
\begin{equation*} 
\inducedquasimetric(\classof{\anelema}, \classof{\anelemb})=0\quad\Leftrightarrow\quad \ametricof{\anelema}{\anelemb}=0\quad\Leftrightarrow\quad\anelema\inducedequivof{\ametric}\anelemb\quad\Leftrightarrow\quad\classof{\anelema}=\classof{\anelemb}\ .\qedhere
\end{equation*}
\end{proof}
\begin{lemma}\label[lemma]{Lemma:Structure2Ori}
Let $(\aset, \equiv, \factorquasimetric)$ be a factorization and pruning structure. 
Then $\inducedorimetric$ is an orimetric.
\end{lemma}
\begin{proof}
Reflexivity and the triangle inequality carry over from the quasimetric.
Symmetry at zero is a consequence of the equality at zero property of quasimetrics:
\begin{equation*}
\inducedorimetricof{\anelema}{\anelemb} = 0 \ \ \Rightarrow\ \
    \factorquasimetricof{\classof{\anelema}}{\classof{\anelemb}} = 0 \ \ \Rightarrow\  \
    \classof{\anelema} = \classof{\anelemb}  \ \ \Rightarrow\ \
    \factorquasimetricof{\classof{\anelemb}}{\classof{\anelema}} = 0 \ \ \Rightarrow\  \
    \inducedorimetricof{\anelemb}{\anelema} = 0 \ .\qedhere
\end{equation*}
\end{proof} 
With this equivalence in place, we concentrate on orimetrics and illustrate the above definitions. 
We visualize an orimetric by a labeled and directed graph.
The nodes are the classes in the induced equivalence. 
They are labeled by the elements they contain. 
The graphs are complete, meaning we have a directed edge between every pair of nodes. 
The edge is labeled by the distance between the nodes in the corresponding classes.
Here we use 
\Cref{Lemma:MainProperty}~(ii). 
The orientation of the edges matters, 
which is what makes orimetrics more expressive than metrics. 
The triangle inequality says that one cannot reduce the distance from one class to another by traveling a detour. 

Consider the set $\rset=\setof{\rx, \ry, \ropx, \ropy}$.
The following graphs define orimetrics $\ametric_1$ and $\ametric_2$ on~$\rset$ that will serve as running example: 
\begin{center}
    \scalebox{0.8}{
    \begin{tikzpicture}
        \tikzset{
mycircle/.style={minimum size=10pt, fill opacity=0.15, text opacity=1},
          red/.style={color=red, text=black, fill=red},
          blue/.style={color=blue, text=black, fill=blue},
          rc/.style={rounded corners=0.8mm,line width=0.7pt},
        }

        \node (b) at (-7, 1.5) {$\ametric_1$};

        \node[circle, draw, label={left:$\aclass=\setof{\rx, \ry}$}, align=center, blue, mycircle] (left) at (-5, 0) {};
        \node[circle, draw, label={right:$\aclassp=\setof{\ropx, \ropy}$}, align=center, blue, mycircle] (right) at (-2.5, 0) {
            
        };
        \path[-stealth,rc,blue] (left) edge[bend left=20]  node[scale=1]
        [above] {$1$} (right);
        \path[-stealth,rc,red] (right) edge[bend left=20]  node[scale=1]
        [below] {$2$} (left);

        \node (b) at (3, 1.5) {$\ametric_2$};

        \node[circle, draw, align=center,  label={left:$\aclass=\setof{\rx, \ry}$},blue, mycircle] (leftb) at (5, 0) {};
        \node[circle, draw, align=center,  label={right:$\aclassp_1=\setof{\ropx}$}, blue, mycircle] (righttop) at (7.5, 1.25) {};
        \node[circle, draw, align=center, label={right:$\aclassp_2=\setof{\ropy}$},blue, mycircle] (rightbot) at (7.5, -1.25) {};

        \path[-stealth,rc,red] (leftb) edge[bend left=15]  node[scale=1]
        [above] {$1$} (righttop);
        \path[-stealth,rc,blue] (righttop) edge[bend left=15]  node[scale=1]
        [below] {$1$} (leftb);

        \path[-stealth,rc,blue] (leftb) edge[bend left=15]  node[scale=1]
        [above] {$2$} (rightbot);
        \path[-stealth,rc,red] (rightbot) edge[bend left=15]  node[scale=1]
        [below] {$1$} (leftb);

        \path[-stealth,rc,blue] (righttop) edge[bend left=15]  node[scale=1]
        [right] {$3$} (rightbot);
        \path[-stealth,rc,red] (rightbot) edge[bend left=15]  node[scale=1]
        [left] {$2$} (righttop);
    \end{tikzpicture}
    }
\end{center}
The leftmost node in the graph defining $\ametric_1$ represents the equivalence class $\aclass=\setof{\rx, \ry}$, which means $\ametric_1(\rx, \ry)=0$. 
The directed edge from $\aclass$ to the class $\aclassp=\setof{\ropx, \ropy}$ is labeled by $1$.
This means $\ametric_1(\anelema, \anelemb)=1$ for all $\anelema\in\aclass$ and $\anelemb\in\aclassp$.
That this distance is the same for all elements in the two classes is \Cref{Lemma:MainProperty}~(ii). 
We do not have symmetry, the edge from $\aclassp$ back to $\aclass$ has distance $2$ instead of $1$. 
The orimetric $\ametric_2$ splits up the equivalence class $\aclassp$.  
Consider the blue edges. 
We have the distance $\ametric_2(\ropx, \ropy)=3$. 
The triangle inequality says that we cannot take a shortcut by traveling from $\ropx$ to $\ropy$ via $\aclass$. 
If we assigned the blue edge from $\aclass$ to~$\aclassp_2$ a distance of $1$, the triangle inequality would fail and the graph would not represent an orimetric.

Orimetric $\ametric_1$ is a congruence wrt. $\anOperator$, but $\ametric_2$ is not. 
We have $\rx\inducedequivof{\ametric_2} \ry$ but $\ropx\not\inducedequivof{\ametric_2}\ropy$. 
Consider the set $\setof{\rx}$.
Neither $\ametric_1$ nor $\ametric_2$ is precise wrt. $\setof{\rx}$, because closing the set under the induced equivalence would add the element $\setof{\ry}$. 
Both orimetrics $\ametric_1$ and $\ametric_2$ are precise wrt. $\setof{\ropx, \ropy}$, but only $\ametric_1$ is unambiguous wrt.\ this set. 
Under $\ametric_2$, the set falls apart into two equivalence classes.

We discuss the motivation behind the properties we require of orimetrics.
Reflexivity is essential for the refinement.
It says that a distance different from zero is enough to rule out a program as a candidate solution. 
Metrics and also the weaker pseudo-metrics are typically symmetric~\cite{metricsbook}.  
By avoiding symmetry, we can assign a meaningful distance to relations that are oriented, 
as explained in the overview. 
We need symmetry at zero for symmetry of the induced equivalence. 
The triangle inequality is important for transitivity of the induced equivalence and for \Cref{Lemma:MainProperty}~(ii). 
The three properties seem to be what is needed in the context of synthesis.
We have not found this definition in the literature, and chose to name the object oriented metrics.

In the overview, we considered the ground truth to be a single function.
This assumption is actually justified, but needs some discussion. 
For now, all we know is that $\groundTruthFuncs$ is a set. 
The point is that an example-based specification defines the output only for some of the inputs, but leaves freedom for the remaining inputs.  
Consider now a candidate solution $\aFunc$ and $\groundTruthFunc_1, \groundTruthFunc_2\in\groundTruthFuncs$.  
It could be the case that $\ametricof{\aFunc}{\groundTruthFunc_1}=0$ but $\ametricof{\aFunc}{\groundTruthFunc_2}\neq 0$. 
This means if we picked the wrong ground truth function, we would not be able to show that we solved the synthesis task. 
The problem disappears if we assume that the orimetric is \emph{unambiguous} wrt. the ground truth. 
Not only will we have $\ametricof{\aFunc}{\groundTruthFunc_1}=0$ if and only if $\ametricof{\aFunc}{\groundTruthFunc_2}=0$. 
Lemma~\ref{Lemma:MainProperty}~(ii) even guarantees that the distance is the same, 
$\ametricof{\aFunc}{\groundTruthFunc_1}=\ametricof{\aFunc}{\groundTruthFunc_2}$, whether it is zero or not!

\paragraph{Remark}
We expect all orimetrics we work with to be unambiguous wrt.\ the ground truth $\groundTruthFuncs$.
With the above discussion, we can then concentrate on an arbitrary element $\groundTruthFunc\in\groundTruthFuncs$. \subsection{Search Spaces and Enumeration Order}
The two main parameters in bottom-up enumerative SyGuS solvers are the search space and the enumeration order. 
We now give them formal definitions. 
\begin{definition}
A \emph{search space} is a set of programs $\aprogSet \subseteq \completeLanguageOf{\agrammar}$.
We call $\completeLanguageOf{\agrammar}$ the full search space.
The semantics of a program is a function from $\allFuncs=(\variables \rightarrow \dataDomain) \rightarrow \dataDomain$. 
If we have an orimetric~$\ametric$ on this set $\allFuncs$, we speak of an \emph{oriented metric search space}. 
We also write  $(\aprogSet, \ametric)$ to make both components explicit. 
Since we are interested in bottom-up enumerative solving, the search space should be \emph{bottom-up enumerable}:  for every program $\aprog\in\aprogSet$ we should also find all subprograms in the set, $\dcof{\aprog} \subseteq\aprogSet$. 
We use $\buof{\aprogSet}$ to denote the largest bottom-up enumerable set contained in~$\aprogSet$. 
\end{definition}
\begin{lemma}\label[lemma]{lem:bu_closed}
The largest bottom-up enumerable set contained in $\aprogSet$ is well-defined:
the bottom-up enumerable sets are closed under arbitrary unions, and so
$\buof{\aprogSet}\subseteq\aprogSet$ exists and is unique.
\end{lemma}
\noindent To give an example, consider the set $\rset=\setof{\rx, \ry, \ropx, \ropy}$ from above.  
This set is bottom-up enumerable,  $\buof{\rset}=\rset$. 
If we remove the element $\ry$, we have $\buof{\setof{\rx, \ropx, \ropy}}=\setof{\rx, \ropx}$. 
We fail to construct the program $\ropy$ by bottom-up enumeration.

When is it sound to restrict the search to a subset of programs? 
The following definition gives a sufficient condition. 
\begin{definition}
We call a search space $\aprogSetp$ \emph{complete} wrt. another search space $\aprogSet$ and orimetric $\ametric$, 
if for every program $\aprog\in\aprogSet$ there is a program $\aprogp\in\aprogSetp$ with $\semanticsOf{\aprog}\inducedequivof{\ametric}\semanticsOf{\aprogp}$. 
\end{definition}
\noindent 
In \Cref{lem:cegar:termination}, we give termination guarantees for working with complete search spaces.
Pruning methods based on observational equivalence satisfy completeness.
Also, the abstraction method in~\cite{blaze} is complete. 
We will deliberately use \emph{incomplete} search spaces.
\begin{definition}
An \emph{enumeration order} is a well-founded total order $\enumerationOrder{}\ \subseteq \aprogSet \times \aprogSet$ on the search space. 
It is a \emph{bottom-up} enumeration, if it is compatible with the subprogram relation,  for all $\aprog, \aprogp\in \aprogSet$, $\aprog\subprogrel\aprogp$ implies $\aprog\enumerationOrder{} \aprogp$. 
For the factorization, it will be important that the enumeration order is a \emph{precongruence}.
For every operation $\anOperator$, say of arity $1$, we expect that $\aprog\enumerationOrder{} \aprogp$ implies $\anOperator(\aprog)\enumerationOrder{} \anOperator(\aprogp)$. 
The definition generalizes to higher arities.
\end{definition}
\noindent We define three enumeration orders on the set $\rset$: 
\begin{align*}
\rx\enumerationOrder{0}\ropy\enumerationOrder{0}\ropx\enumerationOrder{0}\ry\quad\qquad
\rx\enumerationOrder{1}\ry\enumerationOrder{1}\ropy\enumerationOrder{1}\ropx\quad\qquad
\rx\enumerationOrder{2}\ropx\enumerationOrder{2}\ry\enumerationOrder{2}\ropy\ .
\end{align*}
The enumeration order $\enumerationOrder{0}$ is not bottom-up, because $\ropy\enumerationOrder{0}\ry$. 
It is also not a precongruence, because $\rx\enumerationOrder{0}\ry$ but $\ropy\enumerationOrder{0}\ropx$.
The enumeration order $\enumerationOrder{1}$ is bottom-up, but fails to be a precongruence for the same reason.
The enumeration order $\enumerationOrder{2}$ is bottom-up and a precongruence.

Given an enumeration order, we can construct (a subset of) the search space $\aprogSet$ algorithmically. 
We go through the programs as prescribed by the order, and only keep a program if the subprograms have already been listed. 
Let this procedure return the set $\enumof{\enumerationOrder{}}{\aprogSet}$. 
On the running example, this yields $\enumof{\enumerationOrder{0}}{\rset}=\setof{\rx, \ropx, \ry}$, $\enumof{\enumerationOrder{1}}{\rset}=\rset=\enumof{\enumerationOrder{2}}{\rset}$.
\begin{lemma}\label[lemma]{Lemma:BottomUpEnumeration}
If $\enumerationOrder{}$ is bottom-up, then $\enumof{\enumerationOrder{}}{\aprogSet}=\buof{\aprogSet}$. 
\end{lemma}

 \subsection{Components of the CEGAR Loop}\label{sec:cegar}
\paragraph{$\pruneFunc$}
The function computes a 
ball around the ground truth element we have chosen. 
We now define this ball.
We again give the definition in the abstract  for flexibility, meaning we consider an arbitrary set $\aset$ with orimetric $\ametric$.
Let $\anelemgt\in\aset$ be the element that should serve as the center of the ball.
Let $\radius\in\realsgtz$ be the radius.
The \emph{ball} and the \emph{closed ball} around~$\anelemgt$ of radius~$\radius$ are defined by 
\begin{align*}
\ballof{\radius}{\anelemgt}\ \ =\ \ \setCollector{\anelema\in\aset}{\ametricof{\anelema}{\anelemgt}< \radius}\qquad\text{resp.}\qquad
\cballof{\radius}{\anelemgt}\ \ =\ \ \setCollector{\anelema\in\aset}{\ametricof{\anelema}{\anelemgt}\leq \radius}
\ .
\end{align*}
\noindent The ball is of course monotonic in the radius: the larger the more. 
To give an example, consider the set $\rset$ with orimetric $\ametric_1$ from above. 
Then $\ballof{0}{\rx}=\emptyset$, $\cballof{0}{\rx}=\setof{\rx, \ry}=\ballof{1}{\rx}=\cballof{1}{\rx}$, and $\cballof{2}{\rx}=\rset$.  Moreover, $\ballof{1}{\ropx}=\setof{\ropx, \ropy}$ and $\cballof{1}{\ropx}=\rset$. 

Function $\pruneFunc$ constructs the elements in the ball bottom-up, using \Cref{Lemma:BottomUpEnumeration}.
This means the enumeration order should be bottom-up, and the result of the enumeration will not be the ball itself but $\buof{\ballof{\radius}{\anelemgt}}$, the largest bottom-up enumerable set that lives inside the ball.
\paragraph{$\factorizeFunc$}
We factorize the search space $\aprogSet=\buof{\ballof{\radius}{\anelemgt}}$ that we have just determined along the equivalence induced by the orimetric, $\factorize{\aprogSet}{\inducedequivof{\ametric}}$. 
This means programs are put into an equivalence class if their distance is zero. 
How do we represent the equivalence classes in a way that can be manipulated algorithmically? 
The idea is to use a representative system that only keeps the $\enumerationOrder{}$-least program from each equivalence class. 
Formally, the $\factorizeFunc$ function is defined as follows: \begin{align*}
        \factorizeFuncOf{\aprogSet}{\functionDistanceMeasure}{\enumerationOrder{}}\
        = \ 
        \setCollector{\aprog \in \aprogSet}
        {\forall \aprogp \in \aprogSet: \aprogp \enumerationOrderNotEq{} \aprog \Rightarrow \semanticsOf{\aprog}\not\inducedequivof{\ametric}\semanticsOf{\aprogp}}\ .
\end{align*}

We return to the set $\rset=\setof{\rx, \ry, \ropx, \ropy}$.
Recall that we defined the orimetrics $\ametric_1$ and~$\ametric_2$, but only the former is a congruence.
We also have the enumeration orders $\enumerationOrder{1}$ and $\enumerationOrder{2}$, but only the latter is a precongruence.
Now $\factorizeFuncOf{\rset}{\ametric_1}{\enumerationOrder{1}}=\setof{\rx, \ropy}$, $\factorizeFuncOf{\rset}{\ametric_2}{\enumerationOrder{2}}=\setof{\rx, \ropx, \ropy}$, and $\factorizeFuncOf{\rset}{\ametric_1}{\enumerationOrder{2}}=\setof{\rx, \ropx}$. 
All sets are complete wrt. $\rset$ and the given orimetric. 
However, only the last set is bottom-up enumerable. 
We see,
the correct use of $\factorizeFunc$ is intricate.
If we use $\factorizeFunc$ with an orimetric that is not a congruence or an enumeration order that is not a precongruence, we can lose bottom-up enumerability.
The following lemma states sufficient conditions 
for $\factorizeFunc$ to return a search space that is bottom-up enumerable.
\begin{lemma}\label[lemma]{Lemma:Enumeration}
(i) Let $\enumerationOrder{}$ be an enumeration order on $\aprogSet$. 
Then $\factorizeFuncOf{\aprogSet}{\functionDistanceMeasure}{\enumerationOrder{}}$ is complete wrt.~$\aprogSet$ under $\functionDistanceMeasure$. 
(ii) If $\aprogSet$ is bottom-up enumerable, $\enumerationOrder{}$ is additionally a precongruence, and $\ametric$ is a congruence, then  $\factorizeFuncOf{\aprogSet}{\functionDistanceMeasure}{\enumerationOrder{}}$ is bottom-up enumerable.
\end{lemma}
\noindent The examples we have just given show that \Cref{Lemma:Enumeration}~(ii) does not hold without the precongruence and the congruence requirements. 
The proof of the result can be found in the appendix. 
For \Cref{Lemma:Enumeration}~(i), we use the well-foundedness of the enumeration order, which guarantees the existence of a least element in each equivalence class. 

\paragraph{$\findProgFunc$}

The $\findProgFunc$ function iterates through the search space along the enumeration order until it finds the first program that, from the perspective of the orimetric, solves the synthesis task. 
It starts with the $\enumerationOrder{}$-minimal program $\aprog\in\aprogSet$. 
Then it sets $\aprog$ to the successor of $\aprog$ until 
$\semanticsOf{\aprog} \inducedequivof{\ametric} \groundTruthFunc$ holds, 
upon which $\aprog$ is returned.
To compute the successor of $\aprog$, we use the procedure $\enumof{\enumerationOrder{}}{\aprogSet}$.
Assuming that $\enumerationOrder{}$ is bottom-up, 
Lemma~\ref{Lemma:BottomUpEnumeration} shows that $\findProgFunc$ explores $\buof{\aprogSet}$.
Hence, to make sure we inspect the entire search space, $\aprogSet$ must be bottom-up enumerable, $\buof{\aprogSet}=\aprogSet$.
This is where we will use Lemma~\ref{Lemma:Enumeration}(ii).

We illustrate how our theory of enumeration orders can be combined with deduction, more precisely, the deduction technique in DryadSynth~\cite{dryadsynth}.
Consider a program sketch $\asketch$ of arity $n+1$. 
If we find the last parameter of $\asketch$ that is needed to solve the synthesis task, then the instantiation of the sketch should be the next program to enumerate.
To make this formal, assume we currently explore program $\aprog$ and we have already enumerated $\aprog_1, \ldots, \aprog_{n}\enumerationOrder{} \aprog$.
If $\aprogp=\asketchOf{\aprog_1, \ldots, \aprog, \ldots, \aprog_n}$ 
solves the synthesis task, $\semanticsOf{\aprogp}=\groundTruthFunc$, then it should be the immediate successor of $\aprog$, meaning $\aprogp=\successorInOf{\enumerationOrder{}}{\aprog}$. 
If we relax the requirement so that the filled sketch should be the successor of the last program of the same size as $\aprog$, then we get the instantiations from Duet~\cite{duet} and Simba~\cite{simba}.
Checking whether there are $\aprog_1,\ldots, \aprog_n \enumerationOrder{} \aprog$ so that $\semanticsOf{\aprogp}=\groundTruthFunc$ 
can be done efficiently for some operators.
For the $\xorop$ operator for example, only a lookup in a hashmap storing all programs with their values as keys is needed~\cite{dryadsynth}.
For the $\iteop$ operator, a decision tree can be learned, as done in EUSolver \cite{eusolver}.

This modification of the enumeration order may ruin the precongruence property,
and one may be concerned that it jeopardizes the guarantees given by \Cref{Lemma:Enumeration}. 
The point is that we only see the subset of programs up to and including the first deduction step.
On these programs, precongruence and the guarantees hold.
Afterwards, $\findProgFunc$ terminates and so the guarantees are not needed for the remaining programs.

\paragraph{$\refineFunc$}
We leave the exact instantiation of the function to the user and present our instantiation in \Cref{Section:Instantiation}.
Here, we only give two properties that the refinement scheme must satisfy. 
Let $\functionDistanceMeasurep = \refineFuncOf{\functionDistanceMeasure}{\aprog}{\groundTruthFunc}$ be the refined orimetric.
We want
\begin{align}
    \inducedequivof{\functionDistanceMeasurep}\ \subseteq\ \inducedequivof{\functionDistanceMeasure}
    \qquad\qquad
    \mathit{and}
    \qquad\qquad
    \semanticsOf{\aprog} \not\inducedequivof{\functionDistanceMeasurep}\groundTruthFunc\ . \tag{refinement}\label{Equation:Refinement}   \end{align}
The first property states that the new orimetric is more precise in that the induced equivalence relates fewer programs. 
The second says that the new orimetric differentiates between $\semanticsOf{\aprog}$ and $\groundTruthFunc$. 
\begin{lemma}[Progress]\label[lemma]{lem:progress}
If function $\refineFunc$ guarantees the properties in \eqref{Equation:Refinement}, then the algorithm will never explore the same program in two different loop iterations. 
\end{lemma}
\paragraph{$\updateEnumOrder$}
The enumeration order is updated by $\updateEnumOrder$ based on the programs seen in $\findProgFunc$.
Programs and operators which promise to be more valuable for a solution to $\groundTruthFunc$ are preferred in the updated enumeration order.
The exact instantiation is left to the user, and we present our instantiation in \Cref{Section:Instantiation}.

Putting everything together, we obtain the above CEGAR loop. 
The following correctness guarantee immediately follows from the termination condition of the loop. 
\begin{lemma} 
Consider the \sygus{} problem $(\agrammar, (\inputexamples, \outputexamples))$. 
If the CEGAR loop (\Cref{fig:cegar}) terminates, it returns a program $\aprog \in \completeLanguageOf{\agrammar}$ that satisfies the specification, $\semanticsOfAppliedTo{\aprog}{\inputexamples}=\outputexamples$. 
\end{lemma}

We also have a termination guarantee. It puts together Lemmas~\ref{Lemma:BottomUpEnumeration},~\ref{Lemma:Enumeration}, and~\ref{Lemma:MainProperty}(ii). 
\begin{lemma}\label[lemma]{lem:cegar:termination}
    Assume $(\agrammar, (\inputexamples, \outputexamples))$ is solvable, 
    $\enumerationOrder{}$ is bottom-up and a precongruence,
    $\functionDistanceMeasure$ is a congruence and precise as well as unambiguous wrt. $\groundTruthFuncs$, and
    $\pruneFunc$ returns a bottom-up enumerable search space that is complete wrt.\ $\completeLanguageOf{\agrammar}$ and $\functionDistanceMeasure$.
    Then CEGAR will terminate in the first iteration.  \end{lemma}

  \section{Contribution \rom{2} -- Constructing Oriented Metric Search Spaces}\label{Section:Instantiation}
We define new orimetrics and explain how these orimetrics can be refined. 
Our focus will be on strings and bitvectors. 
We do not have new instantiations for the enumeration order, but work with a standard size-based definition. 
However, we enhance this order with new deduction strategies for strings and bitvectors. 
As it turns out, the definition of orimetrics and the design of deduction strategies are somewhat related, so we discuss deduction together with the orimetrics.

\subsection{Lifting Oriented Metrics to Function Spaces}\label{sec:lifting}
For SyGuS, we are interested in orimetrics on the function space $\allFuncs=(\variables \rightarrow \dataDomain) \rightarrow \dataDomain$.
However, as we have seen in the overview, it is convenient to construct such orimetrics on functions by lifting orimetrics on the data domain.
This lifting is actually the reason why we defined orimetrics for arbitrary sets.  
We now make the lifting explicit. 

Let $\ametrictilde$ be an orimetric on the data domain~$\dataDomain$. 
Our goal is to lift $\ametrictilde$ to the function space $\asetx\rightarrow \dataDomain$.   
The set $\asetx$ should be understood as the set of (inputs from) all examples, although the definition does not rely on this understanding.
We define the lifting parametric in a set $\asety\subseteq \asetx$. 
Intuitively, this is the subset of examples we currently consider, and the refinement loop will then make $\asety$ larger and larger.
We define $\ametricpar{\asety}:(\asetx\rightarrow \dataDomain)\times(\asetx\rightarrow\dataDomain)\rightarrow \realsgtz$ by
\begin{align*}
\ametricparof{\asety}{\aFunc}{\aFuncp}\ \ = \ \ \sum_{\anelemy\in\asety}\ametrictildeof{\aFuncOf{\anelemy}}{\aFuncpOf{\anelemy}}\ .
\end{align*}
\begin{lemma}\label[lemma]{lem:lifting}
If $\ametrictilde$ is an orimetric on~$\dataDomain$, then 
$\ametricpar{\asety}$ is an orimetric on $\asetx\rightarrow \dataDomain$, for all $\asety\subseteq\asetx$. 
\end{lemma}

Besides the triangle inequality, 
a quasimetric satisfies $\ametrictildeof{\aval}{\avalp} = 0$ if and only if $\aval = \avalp$~\cite{metricsbook}.
When we lift from a quasimetric on the data domain using the set $\inputexamples$ of inputs from all examples, then the result will be precise and unambiguous wrt. $\groundTruthFuncs$. 
\begin{lemma}\label[lemma]{lem:lift_quasi}
For $\inputexamplesp \subseteq \inputexamples$ and $\ametrictilde$ an orimetric on~$\dataDomain$, the orimetric 
$\ametricpar{\inputexamplesp}$ on $\allFuncs=(\variables \rightarrow \dataDomain) \rightarrow \dataDomain$ is unambiguous wrt. $\groundTruthFuncs$. 
If $\ametrictilde$ is a quasimetric and $\inputexamplesp = \inputexamples$, the orimetric
$\ametricpar{\inputexamplesp}$
is also precise wrt.\ $\groundTruthFuncs$.
\end{lemma}
\paragraph{$\refineFunc$}
As explained in the overview, we use $\approxmetric=\ametricpar{\inputexamplesp}$ with $\inputexamplesp \subseteq \inputexamples$ as our approximate orimetric.
We refine the orimetric by adding an example for which the spurious program $\aprog$ does not return the expected value.
Let $\aninput \in \inputexamples$ be an input example for which 
$\semanticsOfAppliedTo{\aprog}{\aninput} \neq \groundTruthFuncOf{\aninput}$ holds. We define
$
    \refineFuncOf{\ametricpar{\inputexamplesp}}{\aprog}{\groundTruthFunc}
    = 
    \ametricpar{\inputexamplesp \cup \setof{\aninput}}
    \ .
$
This is a refinement indeed. 
Blaze's~\cite{blaze} abstraction refinement can also be instantiated in $\refineFunc$. 
For this, we only need to update the abstraction $\alpha$ and the abstract transformers $\aFunc^{\sharp}$. 
Then we can calculate the orimetric as described in \Cref{sec:state_of_the_art}.

\paragraph{$\updateEnumOrder$}
When $\updateEnumOrderOf{\enumerationOrder{}}{\aprogSet}{\groundTruthFunc}$ is called,
the spurious program $\aprog$ is $\enumerationOrder{}$-maximal in $\aprogSet$.
We use the precise metric $\distanceMeasureForExamples{\inputexamples}$ to judge all subprograms $\aprogp$ of $\aprog$.
For every subprogram $\aprogp$
whose semantics reside in
$\ballfullof{(\allFuncs, \distanceMeasureForExamples{\inputexamples})}{\radius}{\groundTruthFunc}$,
we update the enumeration order by pretending $\aprogp$
has no children and size $1$.
 \subsection{Recipe}
All orimetrics we work with will be liftings from orimetrics on the data domain. 
We therefore only discuss orimetrics on strings and bitvectors, but not on functions over these data domains.  
Actually, the orimetrics on strings and bitvectors we define will be quasimetrics.
This guarantees us preciseness of the orimetrics that result from lifting, \Cref{lem:lift_quasi}. 
We stress, however, that the lifted objects are no longer quasimetrics, they are orimetrics. The definitions in this section are guided by the following
\begin{quote}
{\bfseries Rule of Thumb:}
\\
If we understand deduction for the operators on the data domain, then we understand which values should be close in the orimetric for the data domain. 
\end{quote} \subsection{Oriented Metrics for Strings}\label{sec:metrics_strings}
For strings, we have three orimetrics, two of them defined with deduction in mind.
\subsubsection{An Oriented Metric for $\concatOp$}
For deduction, we start from the sketch $\concatOp(\aprog, \aprog')$ that consists of the concatenation operation.  
If $\aprog$ returns $\aninput$ and $\aprog'$ returns $\aninput'$, then $\concatOp(\aprog, \aprog')$ will return $\anoutput=\aninput.\aninput'$. 
Deduction infers missing arguments from already given arguments and target values. 
Concretely, if we want to obtain value $\anoutput$ and we already have $\aninput$, then the missing input must be $\aninput'$.
This has been captured by the inverse semantics in~\cite{duet}.

We define an orimetric that is optimized for concatenation.  
We think of the arguments of the orimetric as the values that are given to the  deduction engine, namely the output~$\anoutput$ and the input~$\aninput$ that is already known.  
The value $\concatmetricof{\aninput}{\anoutput}$ should then capture how far~$\aninput$ is away from $\anoutput$ when concatenation is applied next. 
This amounts to judging what could be called the complexity of the inverse semantics.
How many values does the inverse semantics contain, and how easy is it to generate these values.
The distance should then be inverse proportional to this complexity.
In particular, when there are no values left in the inverse semantics, the distance should be infinity.
For concatenation, this would be the case when $\aninput$ is not a prefix or a suffix of $\anoutput$. 

While these considerations are useful as guidance, implementing them strictly does not lead to a useful definition. 
First, we want to determine the ball not only for the concatenation operation, but it should also work for the remaining operations.
Second, trying to define the complexity of the inverse semantics quickly gets out of hand mathematically. 
How do we make sure the definition still satisfies the triangle inequality?

Our definition relaxes the requirement that the input should be a prefix or a suffix to just an infix.
We approximate the complexity of the inverse semantics as follows.
The number of elements in the inverse semantics is always one, we are just missing one string. 
The string, however, is easier to generate the shorter it is. 
With this, we use the orimetric
\begin{align*}
\concatmetricof{\aninput}{\anoutput}\ =\ 
\begin{cases}
    \lengthOf{\anoutput} - \lengthOf{\aninput}
    & \text{if $\aninput$ is an infix of $\anoutput$}
    \\
    \bigconst + \absOf{\lengthOf{\anoutput} - \lengthOf{\aninput}}& \text{otherwise} \ .
\end{cases}
\end{align*}
We select the constant $\bigconst \in \reals$ to be larger than the sum over the lengths of all output values in the examples, and use this constant to define the ball. 
Alternatively, we could extend the reals with infinity and define $\concatmetricof{\aninput}{\anoutput}=\infty$, if $\aninput$ is not an infix of $\anoutput$. 
However, this would add infinity as an undesirable technicality to many places.
\subsubsection{An Oriented Metric for $\substrOp$}
The second sketch our deduction should support consists of the substring operation. 
It takes three arguments: the original string $\aninput$, the starting index of the substring that should be determined as the output $\anoutput$, and the maximal length of that substring.
Our orimetric focuses on the relationship between $\aninput$ and~$\anoutput$.
The only precondition we can derive is that the input must be a superstring of the output. 
We define $\substrmetric$ just like $\concatmetric$, but with the infix requirement flipped.
Actually, within \toolname, we have implemented a full new deduction strategy for the substring operation. 
It can be found in the appendix.
\subsubsection{Levenshtein}
Our synthesizer runs with a portfolio of different orimetrics, and it has turned out beneficial to have a fallback that is not optimized for a deduction strategy.
We use the Levenshtein distance for this purpose~\cite{levenshtein}.
The Levenshtein distance between two strings is the number of substitutions, deletions, and insertions that is necessary to convert one string into the other. 
To efficiently check whether a program lies inside a ball,  we use Levenshtein automata~\cite{levenshteinAutomata}.
A Levenshtein automaton for a target string $\anoutput$ and a radius $\radius$ is a finite automaton that accepts all strings $\aninput$ with $\levenshteinDistanceOf{\aninput}{\anoutput} < \radius$.
In our implementation, we set the maximum radius to $\radius = 4$.
 \subsection{Oriented Metrics for Bitvectors}\label{sec:metrics_bv}
For bitvectors, we have three orimetrics tailored towards sketches, and
one additional orimetric.
\subsubsection{An Oriented Metric for $\andop$}
We consider the bitwise conjunction $\andop(\aprog, \aprog')$.
It takes as input bitvectors $\aninput$ and $\aninput'$ and produces the bitvector $\anoutput=\aninput\&\aninput'$.
We again derive an orimetric by working out a deduction method.  
The deduction method reasons over the output $\anoutput$ and the input $\aninput$ that is already known. 
The first step is to check whether $\anoutput$ is bitwise smaller than $\aninput$, denoted by $\anoutput\sqsubseteq\aninput$. 
If the check fails, the deduction aborts. 
If the check succeeds, we calculate the requirements on $\aninput'$, for each bit $b$, as follows. 
(R1) If $\anoutput[b]=1$, then $\aninput'[b]=1$. 
(R2) If  $\anoutput[b]=0$ and $\aninput[b]=1$, then $\aninput'[b]=0$.
There are no requirements if $\anoutput[b]=0$ and $\aninput[b]=0$.

The goal of the orimetric $\andmetricof{\aninput}{\anoutput}$ is to estimate the complexity of the inverse semantics. 
The idea is to sum up all situations (R1) and (R2) in which the inverse semantics has no choice for the value of a bit: $\sum_{b\in[0, 63]}\aninput[b] \vee \anoutput[b]$.
This function, however, will not be reflexive.
If $\aninput=\anoutput=1^{64}$, it will return $64$ rather than $0$.
The solution is to drop requirement (R1).
The idea is to approximate the search for $\aninput'$ from above. 
We take for granted that we are in a space of values that are bitwise larger than the output. 
Only summing up situations (R2) then leads to the final definition:
\begin{align*}
\andmetricof{\aninput}{\anoutput}\ =\ 
\begin{cases}
    \sum_{b\in[0, 63]} \neg \anoutput[b]\wedge \aninput[b]
    &\text{if $\anoutput\sqsubseteq\aninput$}
    \\
    \bigconst
    & \text{otherwise}\ .
\end{cases}
\end{align*}
The large constant is again for emptiness of the inverse semantics.
It is worth noting that under the assumption $\anoutput\sqsubseteq\aninput$, the sum $\sum_{b\in[0, 63]} \neg \anoutput[b]\wedge \aninput[b]$ is just the Hamming distance. 
The Hamming distance is even a metric. 
A similar reasoning yields an orimetric for bitwise disjunction. 
\subsubsection{An Oriented Metric for $\mulop$}
For the multiplication among bitvectors $\mulop(\aprog, \aprog')$, deduction has to start from $\anoutput$ and $\aninput$ and determine $\aninput'$ so that $\anoutput=\aninput\times \aninput'$. 
This amounts to finding the multiplicative inverse in the integer ring represented by the bitvectors. 
Note that the inverse does not exist if the output is smaller than the input.

The corresponding orimetric again approximates the complexity of the inverse semantics, here the multiplicative inverse element.
We take this complexity to be the difference in the number of leading zeros between input and output: 
\begin{align*}
\mulmetricof{\aninput}{\anoutput}\ =\ 
\begin{cases}
    0 &
    \text{if $\aninput = \anoutput$} 
    \\
    1 + \leadingZerosOf{\anoutput}-\leadingZerosOf{\aninput} &\text{if $\leadingZerosOf{\aninput}\leq \leadingZerosOf{\anoutput} $}
    \\
    \bigconst
    & \text{otherwise}\ .
\end{cases}
\end{align*}

\subsubsection{Hamming}
To relate two values directly, we use a modification of the Hamming distance.
The Hamming Distance between two bitvectors is the number of bits where the two bitvectors do not match.
Taking $\hammingDistance$ directly as the quasimetric prunes programs that are good. 
Consider a 64-bit bitvector $\aninput$ with $\hammingDistanceOf{\aninput}{\anoutput} = $ 64.
A single $\notop$ operation would yield $\anoutput$.
With that in mind, we define a new quasimetric $\hdmetric$ that relates two bitvectors $\abitvec$ and $\abitvecp$ as follows:
\begin{equation*}
\hdmetricof{\abitvec}{\abitvecp} = 
\begin{cases}
    \hammingDistanceOf{\abitvec}{\abitvecp}
    &, \hammingDistanceOf{\abitvec}{\abitvecp} \leq 32
    \\
    64 - \hammingDistanceOf{\abitvec}{\abitvecp} + 1
    &, \textit{otherwise}\ .
\end{cases}
\end{equation*}

 \subsection{Hacks}

\paragraph{Sharing Programs Between Concurrent Solver Instances}
In our implementation, we run several solvers concurrently that utilize different orimetrics to prune the search space. When a solver with orimetric $\ametric$ finds a candidate solution, we also analyze it in the context of the other orimetrics~$\functionDistanceMeasurep$.  
The goal is to find subprograms that are valuable wrt.\ $\functionDistanceMeasurep$, and modify the enumeration order of the corresponding solver accordingly.
This way, valuable programs can be shared between the concurrent solvers.

 \paragraph{Keeping Programs up to a Size Threshold}
We do not want to apply pruning to small programs, otherwise the bottom-up enumerable portion of the ball may become too small to be useful. 
We slightly modify the given orimetric $\functionDistanceMeasure$ to create $\functionDistanceMeasurep$, which keeps programs of size up to a threshold~$\asizeThreshold$ in the open ball of radius $\radius$.

  \newlist{questions}{enumerate}{2}
\setlist[questions,1]{label=\textbf{Q\arabic*:},ref=\textbf{Q\arabic*}}
\setlist[questions,2]{label=(\alph*),ref=\thequestionsi(\alph*)}

\section{Contribution \rom{3} -- Implementation and Evaluation}\label{sec:evaluation}
We implemented our approach in a SyGuS solver called \toolname{}.
It is written in C++ and uses Z3~\cite{z3}
as an SMT solver backend for the CEGIS loop.
We evaluate the performance of \toolname{} to answer the following research questions:
\begin{questions}
    \item How does \toolname{} perform on SyGuS tasks of a variety of domains? \label{rq:performance}
    \item How does \toolname{} compare against state-of-the-art SyGuS and domain-specific solvers? \label{rq:compare}
    \item What is the effect of pruning using orimetrics? \label{rq:pruning}
    \item Does abstraction and learning from spurious programs enhance performance? \label{rq:learning}
    \item What impact does the radius of the ball have? \label{rq:radius}
    \item Which orimetrics are successful? \label{rq:successful}
    \item What is the benefit of running multiple instances with different orimetrics concurrently? \label{rq:concurrency}
\end{questions}

We ran all experiments on an Apple M3 Max with 64 GB of RAM and used a 10-minute timeout.

\subsection{Implementation Details}
For strings, we use a size threshold $\asizeThreshold$ of $3$ and
start with orimetrics that consider one example only. 
To implement Levenshtein automata, we use the Mata finite automaton library \cite{automatonLib}.
For Bitvectors, we use a size threshold $\asizeThreshold$ of $7$
and start with orimetrics that consider two examples. 
We concurrently run several instances of the solver using different orimetrics.
We also run an instance that does not prune and will refer to it as $\ametric_{\infty}$. 
For each orimetric designed for a specific sketch, the solver only uses deduction for this sketch. 
The other threads apply deduction on sketches for which we did not design a specific orimetric, e.g.\ the $\addop(\sketchhole, \sketchhole)$ sketch.
That means, for each benchmark featuring strings (bitvectors),
Merlin runs solver instances for all string (bitvector) orimetrics in a portfolio.
In particular, even when a benchmark features multiple oriented operations, e.g. $\concatOp$ and $\substrOp$, Merlin runs four threads each employing one of 
$\concatmetric{}$, $\substrmetric{}$, $\levenshteinDistance$, or $\ametric_{\infty}$.
If the learning mechanism finds a program that is valuable using a specific orimetric, we only change the enumeration order of the thread employing this orimetric and of the thread that does not prune.

The choice of the radius depends on the evaluation.
For the comparison with other tools, we conduct experiments with varying radii for $\levenshteinDistance$ and $\hdmetric$.
For the remaining orimetrics, we use $\bigconst$, which means we evaluate the given conditions and prune if they fail. 
For the ablation studies, we use varying radii for all orimetrics.

 \subsection{Setup}
\paragraph{Benchmarks}
We use benchmarks from three domains: SyGuS bitvector benchmarks without conditionals, SyGuS string benchmarks without conditionals, and the Blaze string benchmark set.

In the Bitvector domain, we have 549 benchmarks:
We include 44 Hacker's delight~\cite{hackersdelight} benchmarks from the SyGuS competition suite and 5 additional Hacker's Delight benchmarks from Probe.
The specification of these benchmarks is not in the form of examples, thus, we use a CEGIS loop.
The CEGIS loop introduces non-determinism.
Therefore, we ran the benchmarks 3 times and report the mean of the results.
We also include the 500 deobfuscation benchmarks from Simba~\cite{simba}.
Here, the specification is given in the form of input-output examples.

In the String domain, we use 181 tasks from Duet~\cite{duet} that return a string.
These include 108 benchmarks from the SyGuS competition,
32 benchmarks designed from Stack Overflow questions,
and 41 benchmarks designed from Exceljet articles.
We also use 108 benchmarks from the Blaze string benchmark set \cite{blaze}.
Blaze uses a custom DSL which is not in the SyGuS format.

Since we did not implement the case-splitting deduction methods from EUSolver~\cite{eusolver}, we omit this class of benchmarks.
However, it is well-understood how to handle conditionals, namely with a decision tree construction.
This was first proposed by EUSolver~\cite{eusolver} and adapted by several other SyGuS solvers. 
Moreover, we stress that the decision tree construction \emph{also fits into our framework}: it is another method of deduction that changes the enumeration order. 
To allow for a fair comparison of the core solving strategy, we remove conditionals from the String benchmark set and from 5 benchmarks of the benchmark set Hacker's Delight. 
The other benchmarks of the benchmark set Hacker's Delight, as well as the Blaze and Deobfuscation benchmark sets do not have conditionals.

\paragraph{Operators Featured in the Benchmarks}
Apart from two benchmarks of the Hacker's Delight benchmark set, \emph{all} benchmarks use oriented operators for which we defined orimetrics.
We give a detailed overview on the operators in each benchmark set.
The benchmark set String has the following string operators:  
$\texttt{str.++}$ (concatenation), $\texttt{str.replace}$, $\texttt{str.at}$ (selecting a character by index), $\texttt{int.to.str}$ (converting an integer to a string), and $\texttt{str.substr}$.
The Blaze benchmarks use the following string operators:  
$\texttt{Concat}$ and $\texttt{SubStr}$. 
The benchmark set Deobfuscation uses the following bitvector operators:  
$\texttt{bvnot}$ (flips all bits), $\texttt{bvxor}$ (bitwise xor), $\texttt{bvand}$ (bitwise and), $\texttt{bvor}$ (bitwise or), $\texttt{bvneg}$ (negation), $\texttt{bvadd}$ (addition), $\texttt{bvmul}$ (multiplication), and $\texttt{bvsub}$ (subtraction).
The Hacker’s Delight benchmark set has three categories that reflect the difficulty of the benchmarks. 
The most difficult benchmarks contain all of the following operators:
the operators from Deobfuscation,   
$\texttt{bvudiv}$ (unsigned division), $\texttt{bvurem}$ (unsigned remainder), $\texttt{bvlshr}$ (logical shift right), $\texttt{bvashr}$ (arithmetic shift right), $\texttt{bvshl}$ (shift left), $\texttt{bvsdiv}$ (signed division), and $\texttt{bvsrem}$ (signed remainder).
The category with medium difficulty typically features 5-10 operators from above.
The simplest benchmarks have 2-4 of the above operators.

\paragraph{Baseline Solvers}
We compare \toolname{} against the general SyGuS tools Probe~\cite{probe} and Duet~\cite{duet}.
In the bitvector domain, we additionally compare \toolname{} against Simba~\cite{simba} and DryadSynth~\cite{dryadsynth}.
In the string domain, we add the recent Synthphonia~\cite{synthphonia} tool for comparison.
For the Blaze benchmarks, we only compare against Blaze~\cite{blaze} since their DSL is not directly compatible with the other solvers.
\Cref{sec:state_of_the_art} offers a description on how all mentioned tools work.
For DryadSynth, we did not enable the ChatGPT feature, which configures an initial enumeration order using ChatGPT. \subsection{Effectiveness of \toolname{}}

\begin{figure}
\centering
\begin{subfigure}{.45\textwidth}
  \centering
  \scalebox{0.9}{
  \includegraphics[width=1\linewidth]{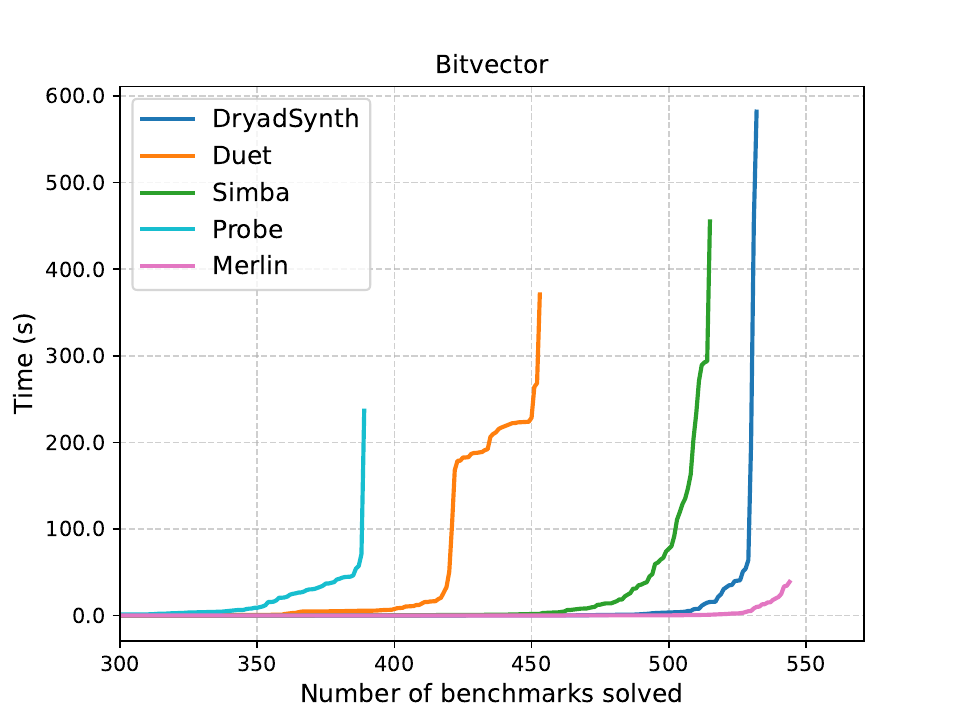}}
  \caption{Time comparison for the bitvector domain.}
  \label{fig:bv_times_300}
\end{subfigure}\begin{subfigure}{.45\textwidth}
  \centering
    \scalebox{0.9}{
\includegraphics[width=1\linewidth]{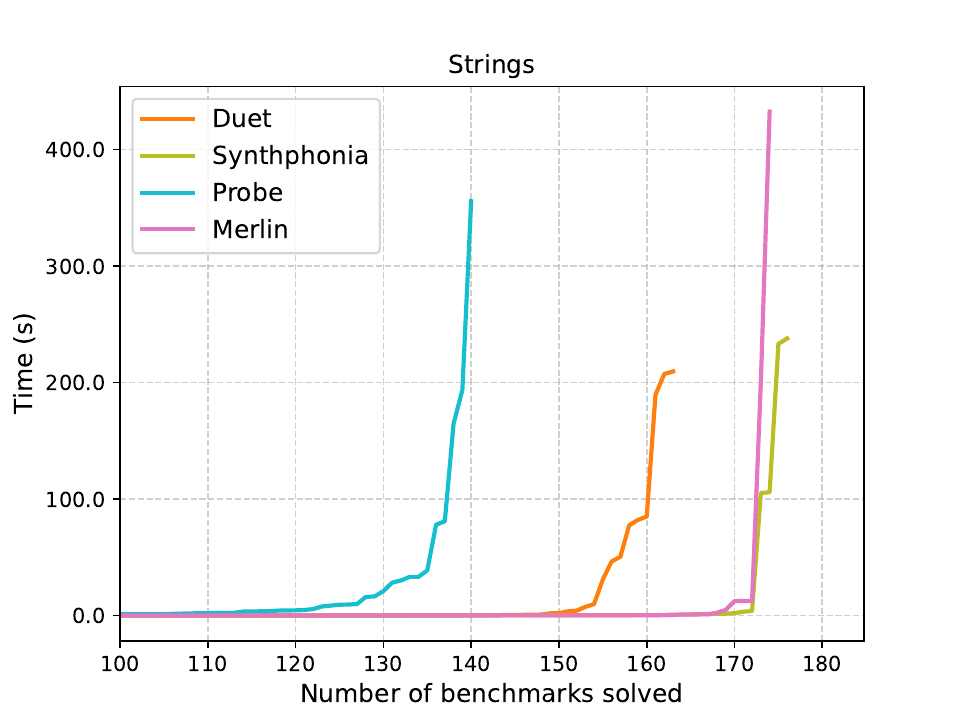}}
  \caption{Time comparison for the string domain.}
  \label{fig:string_times_100}
\end{subfigure}
\begin{subfigure}{.45\textwidth}
  \centering
    \scalebox{0.9}{
\includegraphics[width=1\linewidth]{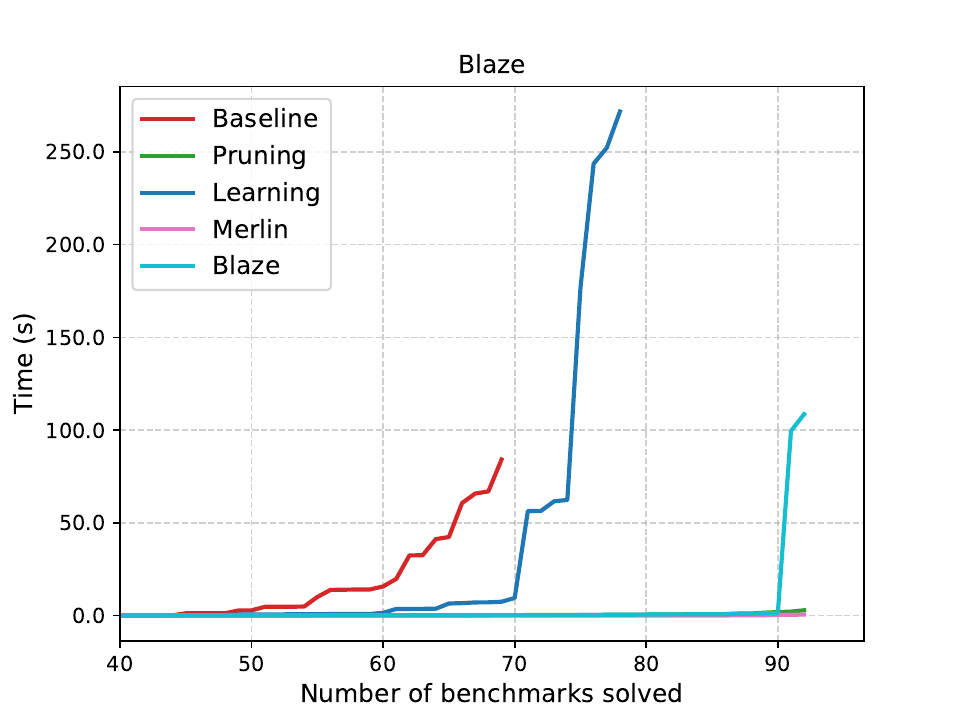}}
  \caption{Time comparison for the Blaze benchmarks.}
  \label{fig:blaze_times_40}
\end{subfigure}\begin{subfigure}{.45\textwidth}
\centering
\scalebox{0.77}{
\footnotesize
\begin{tabular}{|c|c|c|c|} 
    \hline
    Benchmark, $\radius$ & \#Bench & Time & $\cardinalityOf{\aprogSet}$ \\
    \hline
    \hline
Blaze, 25 & 48 (4) & 55.00 & 4126\\ 
\hline
Blaze, 50 & 59 (8) & 58.51 & 8264\\ 
\hline
Blaze, 75 & 81 (17) & 68.42 & 39283\\ 
\hline
Blaze, 100 & 89 (23) & 16.72 & 36433\\ 
\hline
Deobfusc, 25 & 275 (0) & 4.32 & 42\\ 
\hline
Deobfusc, 50 & 288 (1) & 8.00 & 60\\ 
\hline
Deobfusc, 75 & 322 (1) & 4.11 & 24\\ 
\hline
Deobfusc, 100 & 324 (0) & 3.37 & 6\\ 
\hline
Strings, 25 & 130 (4) & 2.32 & 2\\ 
\hline
Strings, 50 & 119 (2) & 2.12 & 2\\ 
\hline
Strings, 75 & 116 (6) & 45.46 & 2\\ 
\hline
Strings, 100 & 108 (2) & 16.01 & 3\\ 
\hline
\end{tabular}}
  \caption{Ablation studies for different $\radius$.}
  \label{table:radius}
\end{subfigure}\caption{Solving time for all benchmarks and ablation studies for different $\radius$.}
\end{figure}

To answer \ref{rq:performance} and \ref{rq:compare}, we evaluate \toolname{} on all benchmarks.

For the bitvector domain,
\Cref{fig:bv_times_300} summarizes the solving times of each solver.
Note, that the x-axis starts at $300$.
The benchmarks before take negligible time.
Probe solves 389 benchmarks.
Duet solves 453 benchmarks,
Simba finds solutions for 515 benchmarks,
and 
DryadSynth solves 532 benchmarks.
Lastly, our new tool, \toolname{}, solves 544 of the total 549 bitvector benchmarks.
If we compare the running time of DryadSynth and \toolname{} on all benchmarks where at least one tool has a solution, 
\toolname{} is 27 times faster than DryadSynth.

We want to note again that we tested DryadSynth without its ChatGPT feature.
This feature applies to the $49$ Hacker's Delight problems.
While the running time reported in their paper~\cite{dryadsynth} with ChatGPT enabled is comparable to the running time without it in most benchmarks,
there are $4$ benchmarks where the ChatGPT feature has a great effect:
DryadSynth is able to solve these 
only with ChatGPT.
\toolname{} is not able to solve these benchmarks.
However, we stress that the ChatGPT feature can be seen as a different initial enumeration order.

For the string domain, 
\Cref{fig:string_times_100} summarizes the solving times for each solver.
Again, note that the x-axis starts at $100$.
Probe solves 140 benchmarks.
Duet solves 163 benchmarks.
\toolname{} can almost compete with the newly proposed, domain-specific string synthesis tool Synthphonia:
\toolname{} solves 174 benchmarks while Synthphonia finds a solution for 176 benchmarks.

For the benchmarks from Blaze, 
\Cref{fig:blaze_times_40} summarizes the solving times for Blaze and \toolname{}.
Again, note that the x-axis starts at $40$.
Both tools solve 92 benchmarks
and have nearly the same running time across the majority of benchmarks.
Still, \toolname{} is the fastest on all benchmarks.
In the hardest benchmarks, \toolname{} is vastly superior.
Overall,
\toolname{} is 75 times faster than Blaze.

 \subsection{Ablation Studies}
To answer \ref{rq:pruning} and \ref{rq:learning},
we now evaluate the effect of pruning and learning on the synthesis performance.
For this, we implemented three additional different versions of \toolname{}:
\begin{itemize}
    \item Pruning: An implementation only using pruning. 
    Here, learning and abstraction is disabled, and therefore we do not use $\refineFunc$ and $\updateEnumOrder$ functions.
    \item Learning: An implementation that uses the orimetrics only for abstraction and learning.
    \item Baseline: This implementation uses neither the pruning nor 
    the learning features.
\end{itemize}
All implementations use the same deduction methods and have the same initial enumeration order.
Moreover, all implementations at least factorize the search space using observational equivalence.

\begin{figure}
\centering
\begin{subfigure}{.45\textwidth}
  \centering
    \scalebox{0.9}{
\includegraphics[width=1\linewidth]{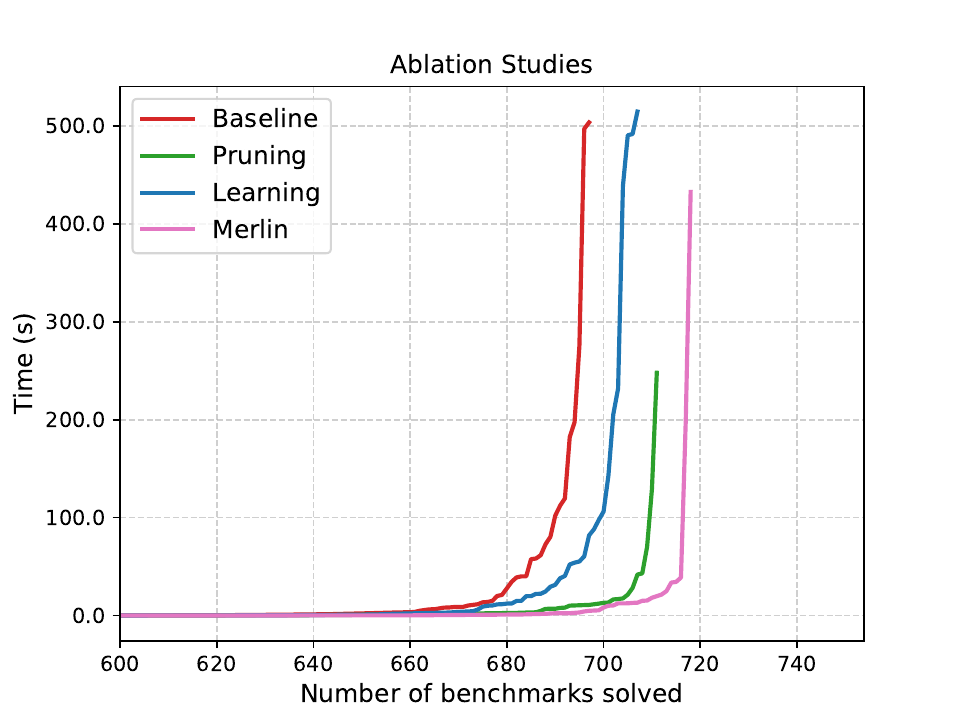}}
  \caption{Solving time across all 730 SyGuS benchmarks.}
  \label{fig:ablation_all_cactus_600}
\end{subfigure}\begin{subfigure}{.45\textwidth}
  \centering
    \scalebox{0.9}{
\includegraphics[width=1\linewidth]{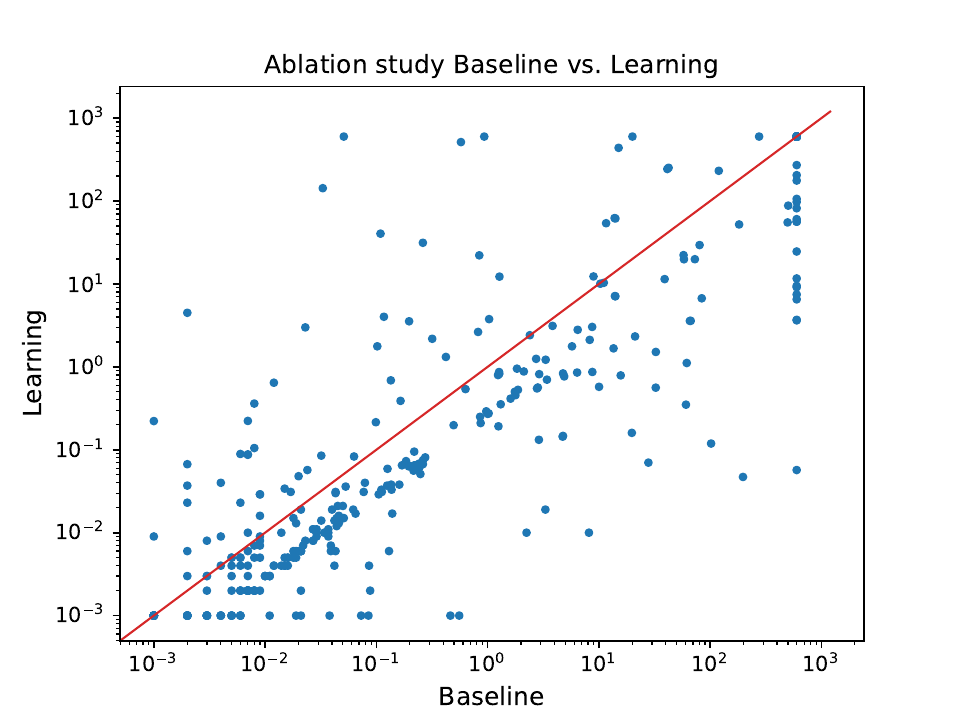}}
  \caption{Scatter plot Baseline vs Learning.}
  \label{fig:scatter_baseline_learning}
\end{subfigure}
\begin{subfigure}{.45\textwidth}
  \centering
    \scalebox{0.9}{
\includegraphics[width=1\linewidth]{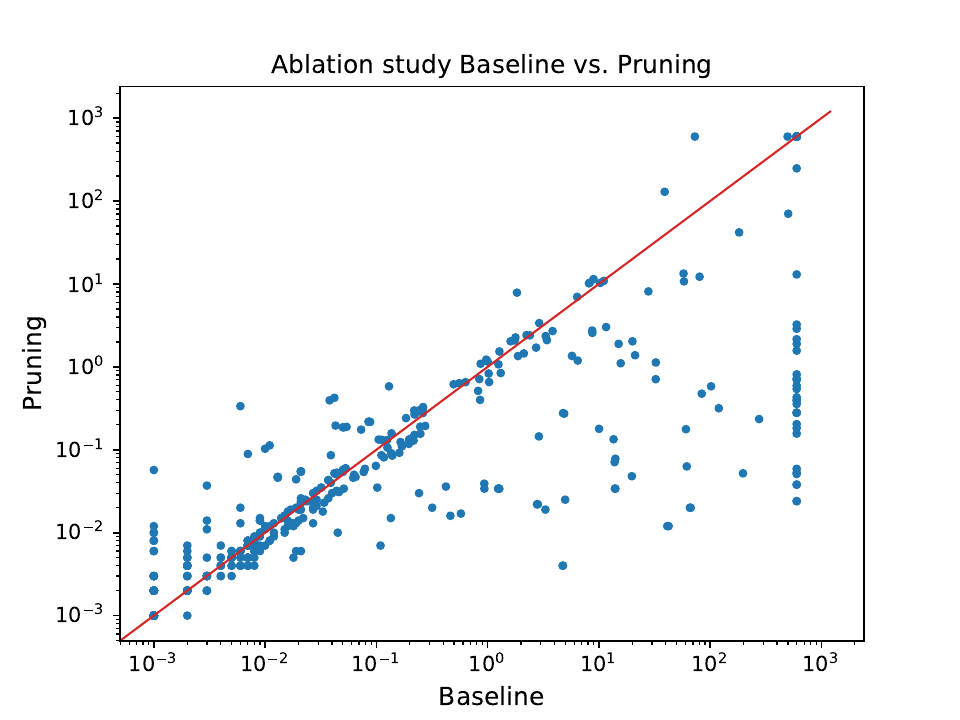}}
  \caption{Scatter plot Baseline vs Pruning.}
  \label{fig:scatter_baseline_pruning}
\end{subfigure}\begin{subfigure}{.45\textwidth}
  \centering
    \scalebox{0.9}{
\includegraphics[width=1\linewidth]{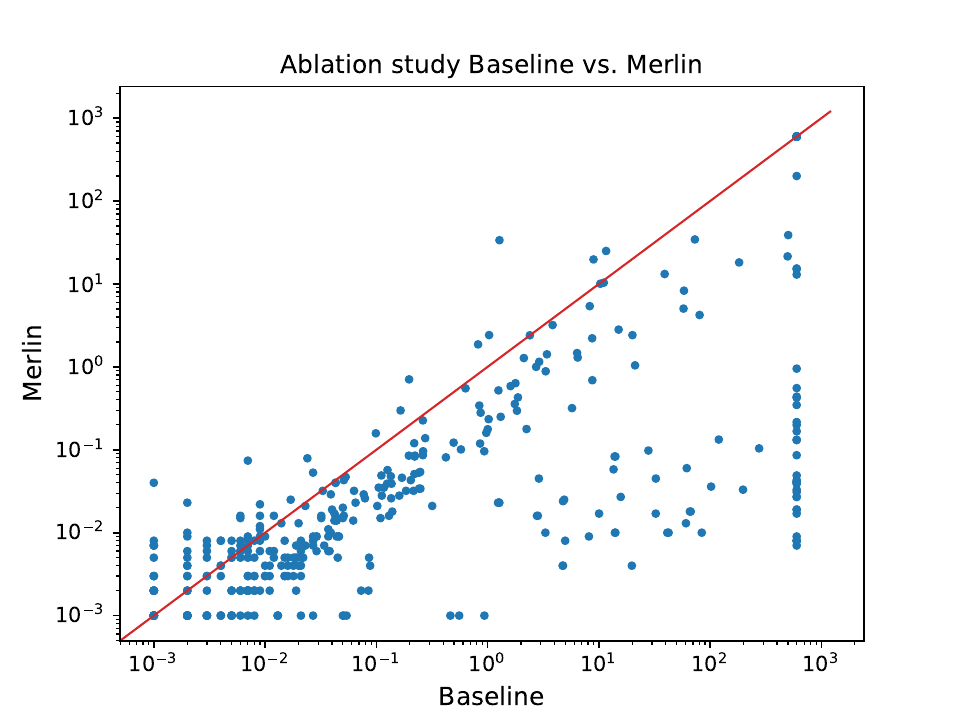}}
  \caption{Scatter plot Baseline vs \toolname{}.}
  \label{fig:scatter_baseline_mysolver}
\end{subfigure}
\caption{Ablation Studies.}
\label{fig:ablation}
\end{figure}

\Cref{fig:ablation} summarizes the ablation studies for the SyGuS benchmarks.
\Cref{fig:ablation_all_cactus_600} shows the running time across all SyGuS benchmarks.
Note, that the x-axis starts at $600$.
The benchmarks before take negligible time.
The Baseline solver solves 697 benchmarks and is the slowest overall.
The Pruning solver solves 711, and
the Learning solver solves 707 benchmarks.
\toolname{} solves 718 benchmarks and is the fastest overall.

\Cref{fig:blaze_times_40} shows the running time across the Blaze benchmarks.
Again, note that the x-axis starts at $40$.
Baseline solves 69 benchmarks and is the slowest overall,
while Learning solves 78 benchmarks.
Pruning and \toolname{} both solve 92 benchmarks.
Remarkable is that our abstraction alone, represented by Learning, is inferior to the abstraction of Blaze.
However, our pruning method is so powerful that it compensates for the quality of our abstraction and even surpasses Blaze's performance.
In future work we want to adapt Blaze's abstraction.

\Cref{fig:scatter_baseline_learning,fig:scatter_baseline_pruning,fig:scatter_baseline_mysolver} compare Baseline against Learning, Pruning, and \toolname{} in a scatter plot.
The axes show the solving time in seconds.
Note that the axes are in log-scale.
\Cref{fig:scatter_baseline_learning} compares Baseline against Learning.
It shows that our learning method has a positive effect on the majority of benchmarks.
However, there are also benchmarks where learning negatively impacts performance.
Since the enumeration order is altered, learning might shift the solution to a later point in the enumeration order.
The positive effect is two-fold.
First, due to the approximate orimetrics, factorization has a greater effect.
Second, learning from previously solved 
examples
enables solving some complex benchmarks where Baseline times out.
\Cref{fig:scatter_baseline_learning} compares Baseline against Pruning.
For trivial benchmarks, the baseline solver outperforms the solver that uses pruning.
The reason may be the overhead due to concurrency. 
For more complex benchmarks,
Baseline takes significantly longer or even times out.
This behavior is expected because we only start pruning the search space after a certain size threshold.
For benchmarks that can be solved below or just above the threshold, the search space is not pruned enough to observe a noticeable improvement.

\Cref{fig:scatter_baseline_mysolver} compares Baseline against \toolname{}.
Combining learning and pruning in Merlin results in a significant speed up on almost all instances.
\toolname{} also solves more instances than Pruning and Learning.
This shows that learning and pruning both have a significant impact on their own and can be combined to provide an even greater benefit.
The negative impact of learning on some instances that we observed in \Cref{fig:scatter_baseline_learning} is also averted.
This is because when using multiple threads, we are less likely to learn a bad program:
We only update the enumeration order of the thread with the corresponding orimetric.
Thus, the enumeration order of the other threads is not impacted.
In the Learning implementation, we only had one thread and one enumeration order; 
we were at risk of learning a bad program.
If we compare the total running time of Baseline and \toolname{} on all benchmarks where at least one tool has a solution, 
\toolname{} is 42 times faster than Baseline.

To answer \ref{rq:radius}, we conduct experiments on the deobfuscation, strings, and Blaze benchmarks.
We did not see a great effect of pruning on the Hacker's Delight benchmarks, thus we omit these here.
We take 100, 75, 50, and 25 percent of the maximum distance as the radius for each orimetric.
By maximum distance, we refer to the maximum distance an orimetric can return given an output $\anoutput$, excluding $\bigconst$.
For example, if $\anoutput$ is a 4-bit bitvector with two bits set, the maximum distance excluding~$\bigconst$ that $\andmetric$ can return is $2$.
With several outputs, we take the sum.
For the orimetric $\substrmetric$, we use discrete pruning at 100\%.
Since for the orimetrics $\concatmetric$ and $\mulmetric$ values smaller than 50\% do not make sense, we use 
$100$, $87.5$, $75$, and $62.5$ percent instead.
For the orimetric~$\levenshteinDistance$,
100\% corresponds to allowing a distance of $4$ for each output, 75\% corresponds to 3, and so on.
\Cref{table:radius} compares Baseline with instances of Pruning.
The second column shows the number of benchmarks that were solved using an underapproximation thread.
In parentheses is the number of these benchmarks, that the baseline solver did not solve.
The third column describes the gained speed-up factor on the benchmarks that were solved by an underapproximation thread and by the baseline.
The last column compares the size of the explored search space.
For example, the first row states that when instantiating Merlin with a 25\% radius and running the Blaze benchmarks,
48 benchmarks were solved by a thread that prunes the search space.
4 of these benchmarks were not solved by the baseline solver.
On the benchmarks that were also solved by the baseline solver, the 25\% instantiation was 55 times faster and considered 4126 times fewer programs than the baseline solver.
The Blaze benchmarks show the biggest gain in performance.
This is not surprising because the top operator of the grammar is concatenation.
Thus, concentrating on substrings is a very good heuristic.
There also seems to be a sweet spot for taking 75\% of the maximum radius, although on the blaze benchmarks this is still too coarse for many benchmarks.
A prime example is the following benchmark:
Convert the string \code{"Launa Withers"} into \code{"L. Withers"}.
Extracting the \code{"L"} is essential for this benchmark, but the output falls outside the 75\% radius.
We leave fine-tuning the radius up to future work. \subsection{Suitability of the Presented Orimetrics and Impact of Concurrency}
To answer \ref{rq:successful}, we analyze \emph{how} Merlin solved the benchmarks.
\Cref{fig:num_threads} shows for how many benchmarks a solver instance was the fastest. 
For example, in Deobfuscation (\Cref{fig:num_thread_deobfusc}), the pie chart says that 
in 49 benchmarks, which make up 9.8\% of the solved instances of the benchmark set, 
the $\andmetric{}$
solver was faster than the solvers with the other orimetrics. 
The entries $\ametric_{\infty}$ refer to the solver instance that does not prune. \begin{figure}
\centering
\begin{subfigure}{.25\textwidth}
  \centering
  \includegraphics[width=1\linewidth]{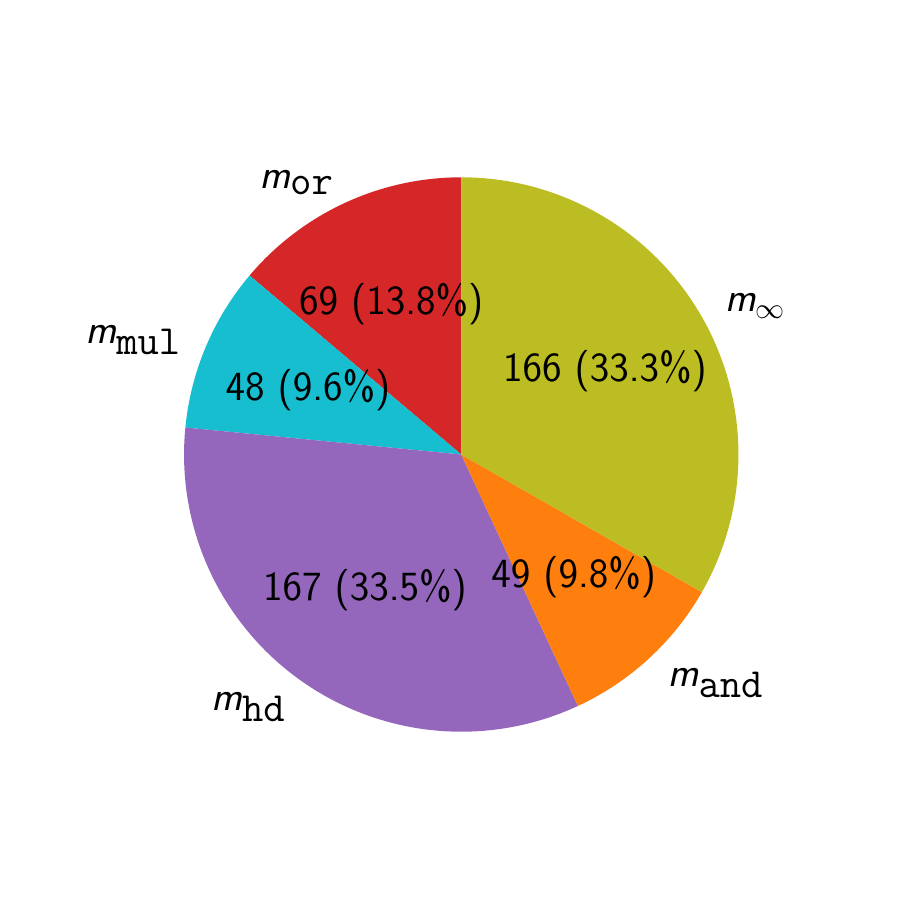}
  \vspace{-1cm}
  \caption{Deobfuscation.}
  \label{fig:num_thread_deobfusc}
\end{subfigure}\begin{subfigure}{.25\textwidth}
  \centering
  \includegraphics[width=1\linewidth]{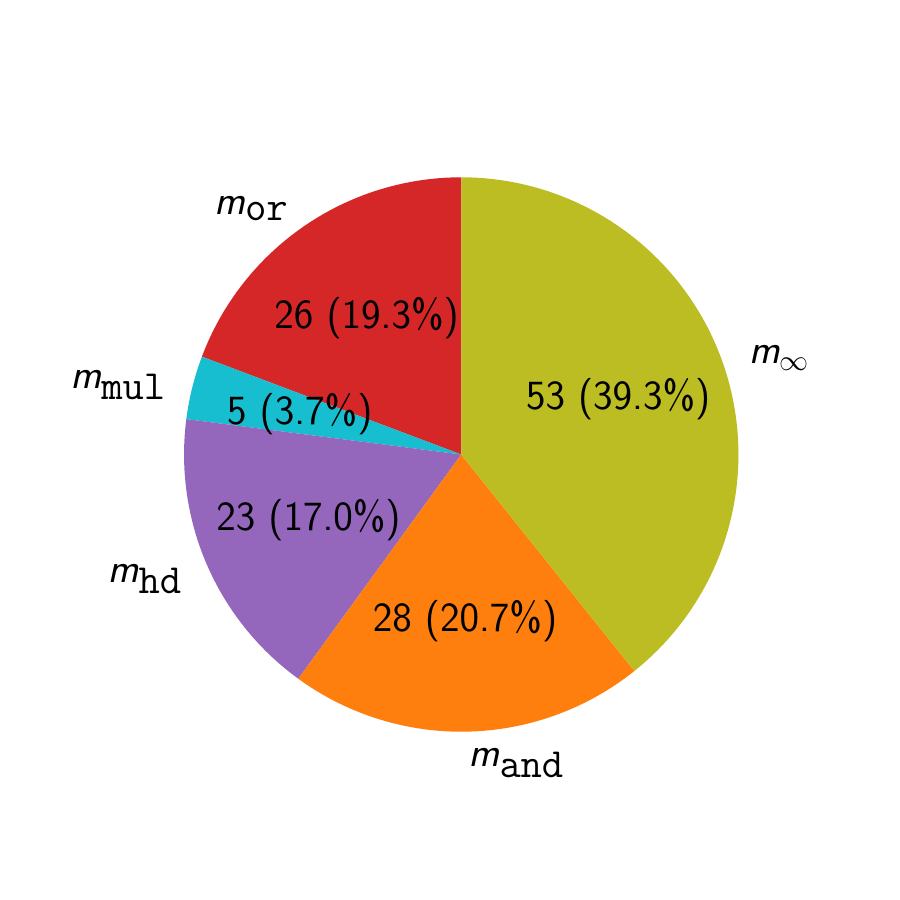}
  \vspace{-1cm}
  \caption{Hacker's Delight.}
  \label{fig:num_thread_hd}
\end{subfigure}\begin{subfigure}{.25\textwidth}
  \centering
  \includegraphics[width=1\linewidth]{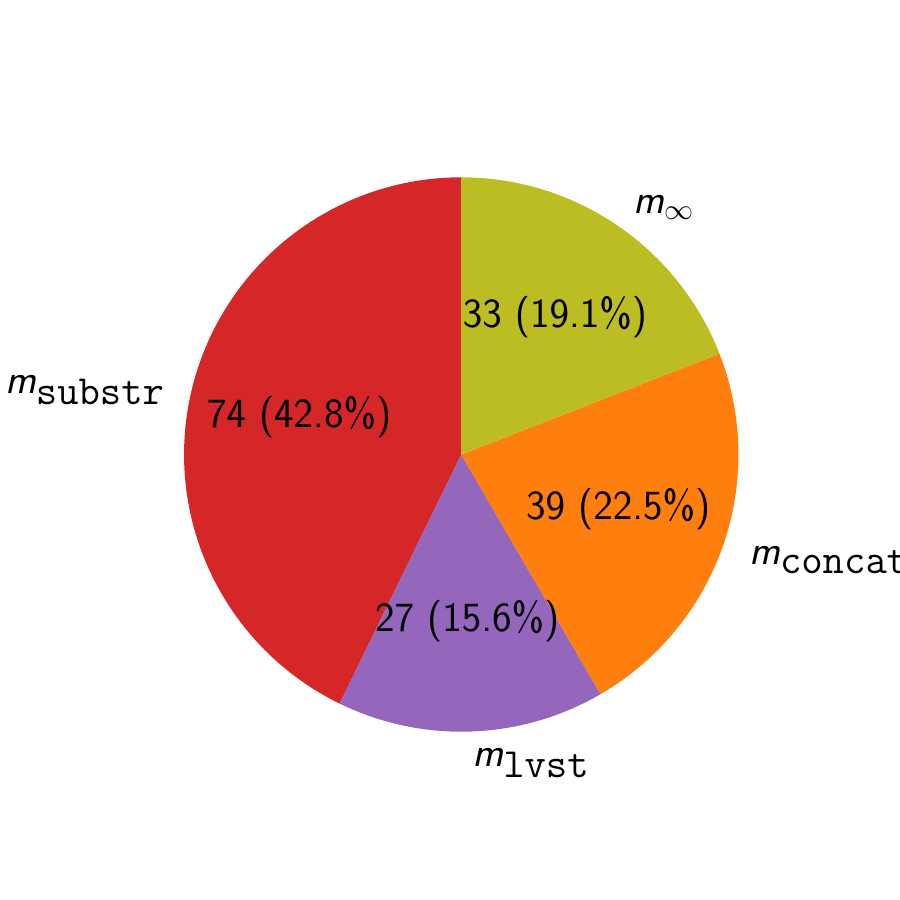}
  \vspace{-1cm}
  \caption{String.}
  \label{fig:num_thread_string}
\end{subfigure}\begin{subfigure}{.25\textwidth}
  \centering
  \includegraphics[width=1\linewidth]{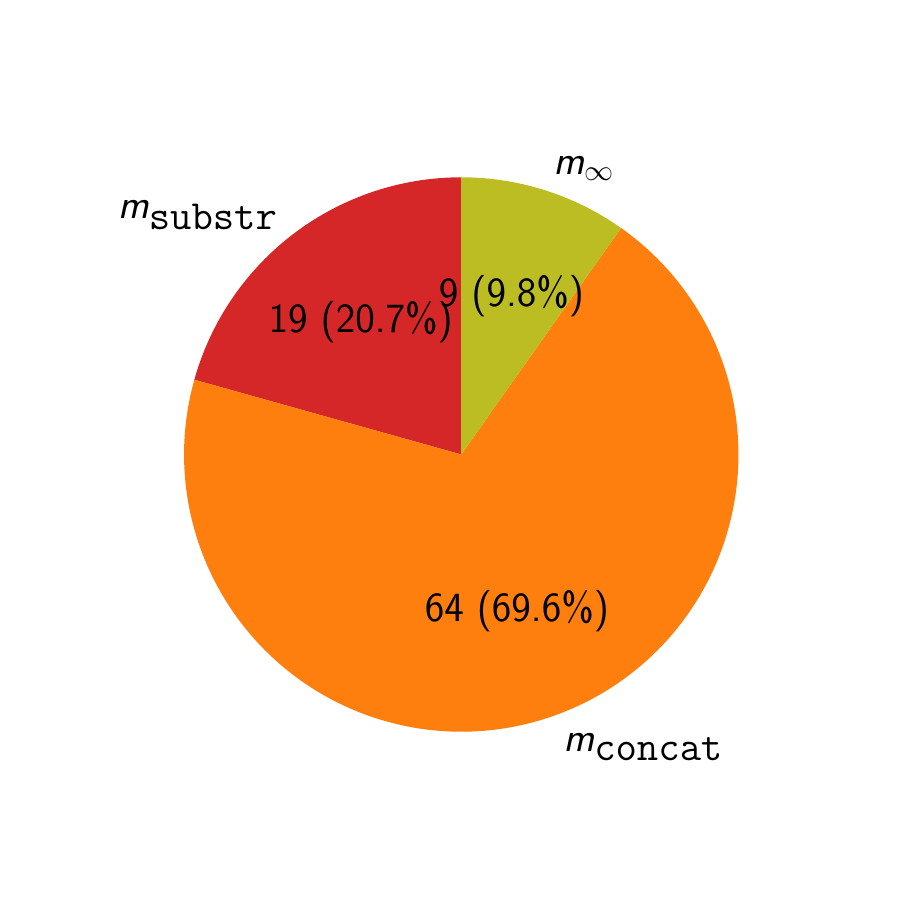}
  \vspace{-1cm}
  \caption{Blaze.}
  \label{fig:num_thread_blaze}
\end{subfigure}
\caption{Number of instances solved using each orimetric for each benchmark set.}
\label{fig:num_threads}
\end{figure}
All in all, the orimetric solvers were successful on the majority of benchmarks, although they prune away a part of the search space.
The results are better for the string than for the bitvector benchmarks. 
For bitvectors (\Cref{fig:num_thread_deobfusc,fig:num_thread_hd}), the success of the $\ametric_{\infty}$ solver indicates that 
pruning may have a negative impact, and there is room for better orimetrics.
Furthermore, it very much depends on the benchmark set which orimetric will fit. 
For example, the orimetric $\mulmetric{}$ fits better the deobfuscation benchmark set than Hacker's Delight. 
The orimetrics designed for strings cover a broader spectrum of the benchmarks.
Here, only $19\%$ of the String benchmark set and $10\%$ of the Blaze benchmark set were solved fastest by the $\ametric_{\infty}$ instance.
Remarkably, the Levenshtein distance was less suitable for the String and Blaze benchmarks than the Hamming distance was for the bitvector benchmarks. \begin{figure}
\centering
\begin{subfigure}{.45\textwidth}
  \centering
    \scalebox{0.9}{
\includegraphics[width=1\linewidth]{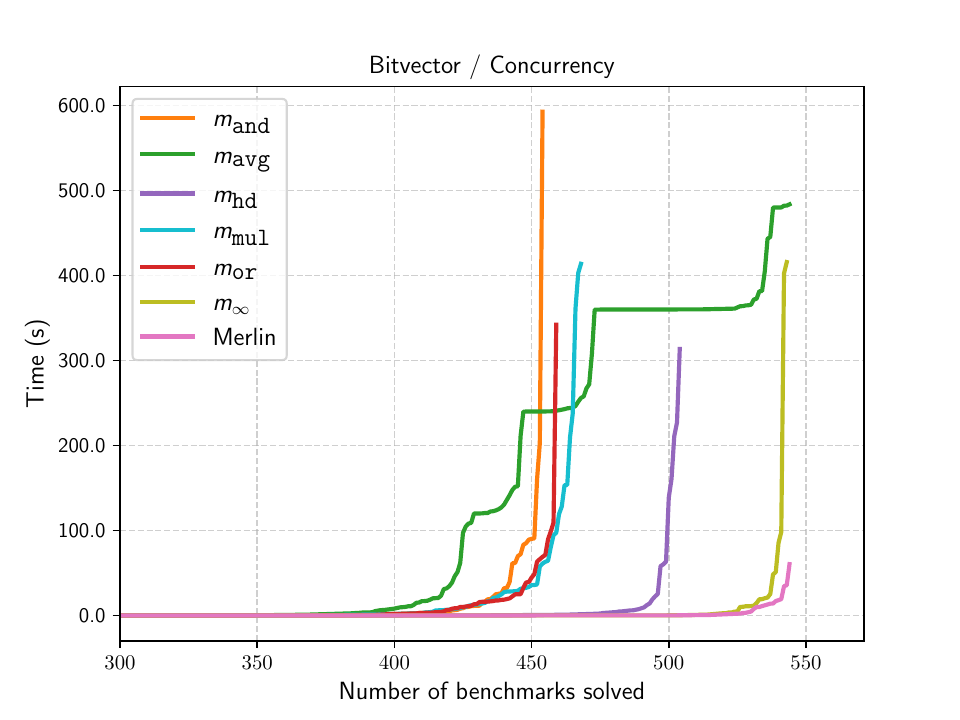}}
  \caption{Time comparison for the bitvector domain.}
  \label{fig:bv_concurrency_times_300}
\end{subfigure}\begin{subfigure}{.45\textwidth}
  \centering
    \scalebox{0.9}{
\includegraphics[width=1\linewidth]{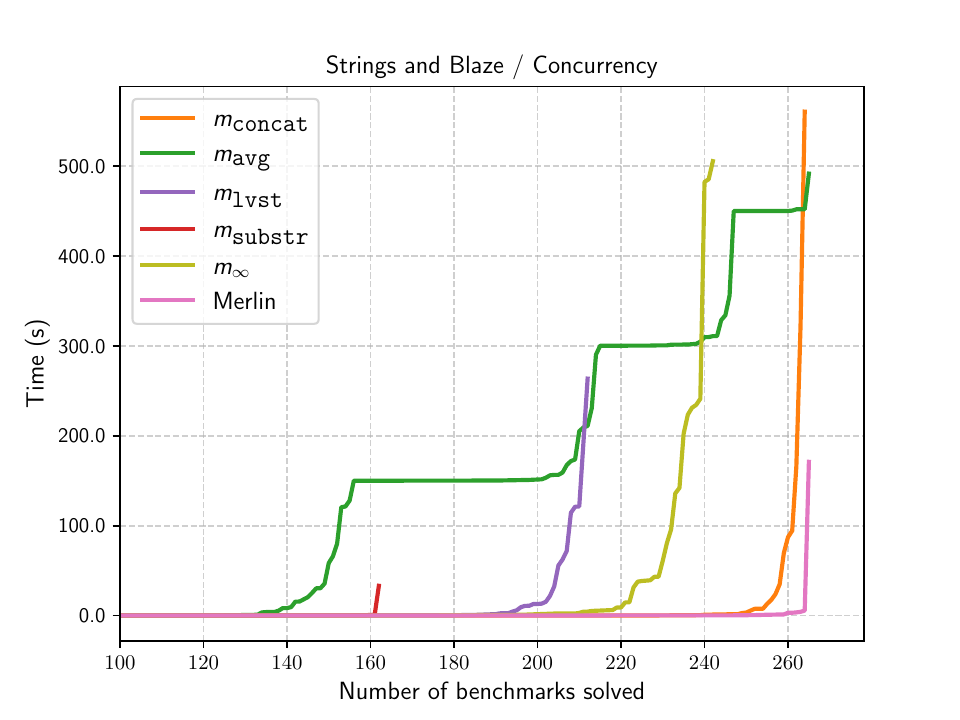}}
  \caption{Time comparison for the string domain.}
  \label{fig:string_concurrency_times_100}
\end{subfigure}
\caption{Solving times for bitvector and string benchmarks.}
\end{figure}

Next, we study the impact of concurrency.
Before execution, we do not know which orimetric will prune the search space best. 
For this reason, we run the solver instances in parallel.
\Cref{fig:bv_concurrency_times_300,fig:string_concurrency_times_100} show the runtime for each instance in isolation, and for Merlin.
Here, $\ametric_{\texttt{avg}}$ represents the average runtime across the instances.
For bitvectors, no instance performs as well as \toolname{}.
This means committing to one orimetric does not work well.
Picking an orimetric at random will, on average, result in $\ametric_{\texttt{avg}}$.
While this random choice solves almost as many benchmarks as \toolname{}, the runtime is 115x slower.
Always using $\ametric_{\infty}$ is 6x slower than \toolname{}.
For strings, the $\concatmetric$ instance comes closest to \toolname{} while being 10x slower.
To answer \ref{rq:concurrency}, concurrently running multiple solver instances employing different orimetrics is crucial for \toolname{}'s performance.
  \section{Related Work}\label{sec:related_work}
\paragraph{Metrics in Synthesis}
SyMetric~\cite{symetric} uses distance metrics for approximate observational equivalence.
They cluster programs whose output is within a given radius inside an equivalence class. 
SyMetric also prunes the search space using a ball around $\groundTruthFunc$.
Their resulting equivalence is not a congruence.
They rely on similarity between the inputs and the outputs for most operators.
Using this similarity, they are able to repair a spurious program given a predefined set of rewrite rules.

Syntia~\cite{syntia} incorporates metrics in top-down search.
For a given sketch, they randomly fill its holes and execute the resulting program.
Using metrics they assign each sketch a score based on similarity of the program's output and the specified output.
This score is then used to guide search.

\paragraph{Pruning and Factorization}
Most tools employ some form of OE factorization~\cite{transit,escher}.
For this,
Version Spaces Algebras are used in FlashFill~\cite{flashfill,flashmeta,flashfillpp} and 
finite tree automata are used by Dace~\cite{wang_fta}.
Absynthe~\cite{absynthe} enables the use of OE for programs with local variables. 
Morpheus~\cite{morpheus} and Neo~\cite{neo} use logical reasoning to prune parts of the search space for which they can prove that it does not contain a solution.
For this, Morpheus requires over-approximative specifications of program components.
It then generates sketches and uses an SMT solver to discard infeasible program sketches.
Neo extends this approach with conflict driven learning.
Having ruled out one sketch, it extracts the root cause of its infeasibility.
Then other sketches for which this root cause applies can also be pruned.
FlashFill++\cite{flashfillpp} proposes cuts to enable middle-out synthesis.
A cut creates a set of subproblems whose solutions can be combined to solve the whole synthesis problem.
This set of subproblems may be incomplete, i.e.\ 
not every possible way to dissect the synthesis problem is explored.
This effectively prunes the search space.

\paragraph{Approximation}
To further prune the search space, a common approach is to use abstractions or type information to represent programs.
With this, one can decide the feasibility of sketches in order to prune infeasible ones.
NOSDAQ~\cite{nosdaq} synthesizes database queries.
Through abstraction of database collections, they can prove that a sketch cannot be completed to yield the correct output.
Similarly, Synquid~\cite{synquid}, its extension~\cite{synquid_extension}, and $\lambda^2$~\cite{lambda2} use type information to prune infeasible partial programs.
Scythe~\cite{sql_from_examples} first searches for sketches that satisfy the specification in the abstract setting.
Then, it uses this set of sketches to search for an instantiation satisfying the specification in the concrete setting.
Absynthe~\cite{absynthe} is similar in that it uses user-defined abstract semantics to create viable sketches. 
For each a viable sketch, it fills the holes and executes the concrete program using a given interpreter.
If the program satisfies the specification, it is returned as the solution.

\paragraph{Learning}
Similar to DryadSynth~\cite{dryadsynth}, HySynth~\cite{hysynth} and DeepCoder~\cite{deepcoder} use Large Language Models to instantiate an enumeration order.
LaSy~\cite{lasy} and ReGuS~\cite{regus} blend upfront and just-in-time learning.
They solve a suite of related synthesis tasks.
The solutions for easier tasks can then be used as a callable component to solve harder tasks.
In our framework, we can simulate these learning approaches by appropriately setting the initial enumeration order.
Atlas~\cite{atlas} extends Blaze~\cite{blaze} by learning useful predicates for abstraction from a set of training problems upfront.
As we have shown for Blaze, predicate abstraction can be seen as an instantiation of oriented metrics.
Bester~\cite{bester} motivates just-in-time learning by showing that for complex benchmarks, partial solutions are often part of the intended solution.
FrAngel~\cite{frangel} modifies partial solutions to solve the synthesis problem.

\section{Conclusion and Future Work}\label{sec:conclusion}
We presented oriented metrics as a foundation for pruning, factorization, refinement, and learning in syntax-guided synthesis. 
We defined a synthesis algorithm that has these features and is parametric in the enumeration order and the initial orimetric. 
We invented new orimetrics for the string and the bitvector domain that occur frequently in \sygus{} problems.
We implemented our approach in a tool called \toolname{}, and obtained a speed-up of an order of magnitude compared to the state-of-the-art.
In the future, we will further explore the design of orimetrics and the choice of when to prune.

\bibliographystyle{ACM-Reference-Format}
\bibliography{bibliography}

\clearpage{}\appendix
\section{More on Evaluation}
\Cref{fig:bv_times_all,fig:string_times_all,fig:blaze_times_all} show the full graph comparing the solving times.
\Cref{fig:ablation_all_cactus} shows the ablation studies for all SyGuS benchmarks.
\begin{figure}
\centering
\begin{subfigure}{.5\textwidth}
  \centering
  \includegraphics[width=1\linewidth]{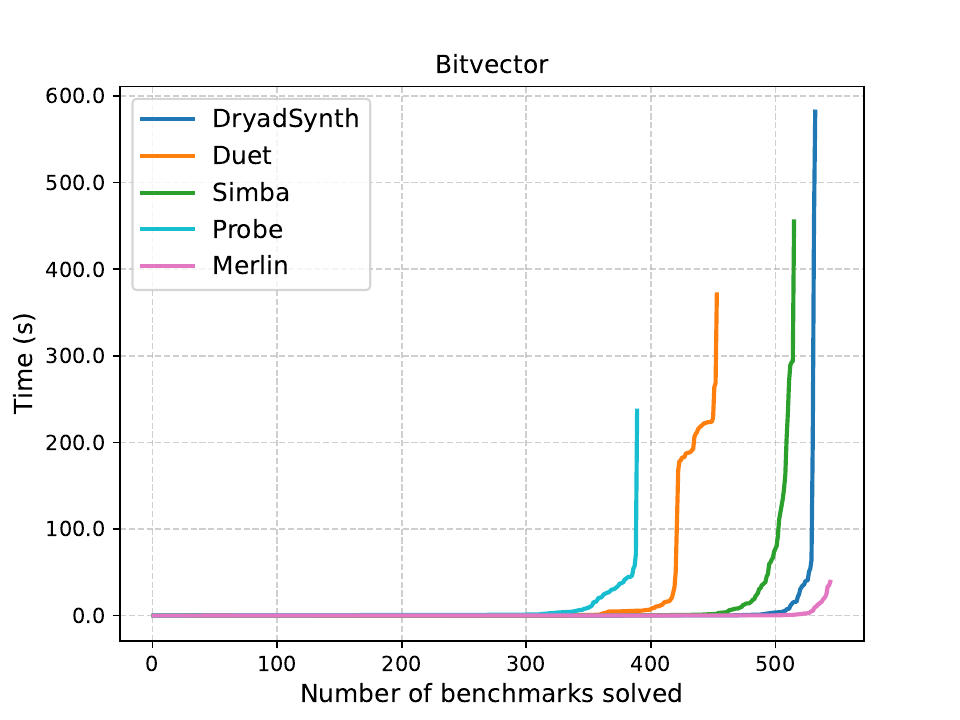}
  \caption{Time comparison for the bitvector domain.}
  \label{fig:bv_times_all}
\end{subfigure}\begin{subfigure}{.5\textwidth}
  \centering
  \includegraphics[width=1\linewidth]{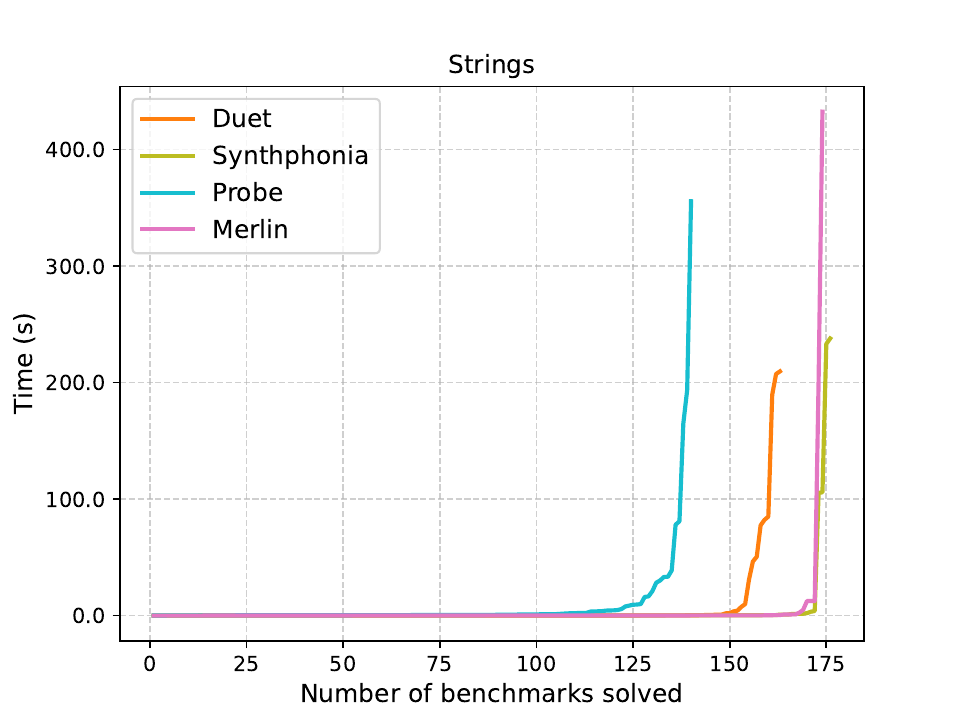}
  \caption{Time comparison for the string domain.}
  \label{fig:string_times_all}
\end{subfigure}
\begin{subfigure}{.5\textwidth}
  \centering
  \includegraphics[width=1\linewidth]{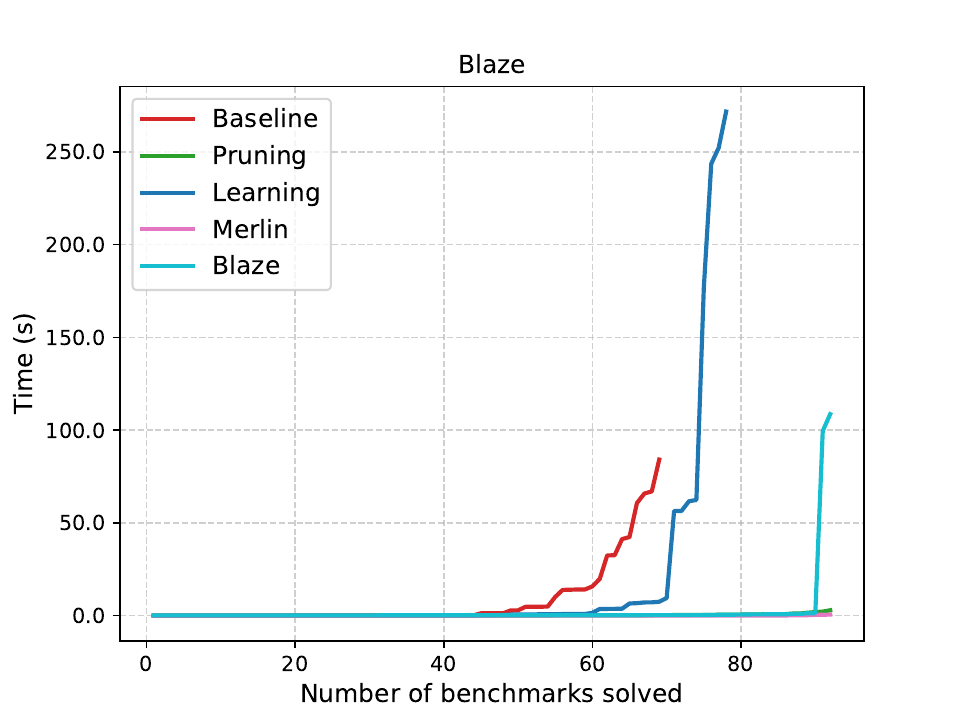}
  \caption{Time comparison for the Blaze benchmarks.}
  \label{fig:blaze_times_all}
\end{subfigure}\begin{subfigure}{.5\textwidth}
  \centering
  \includegraphics[width=1\linewidth]{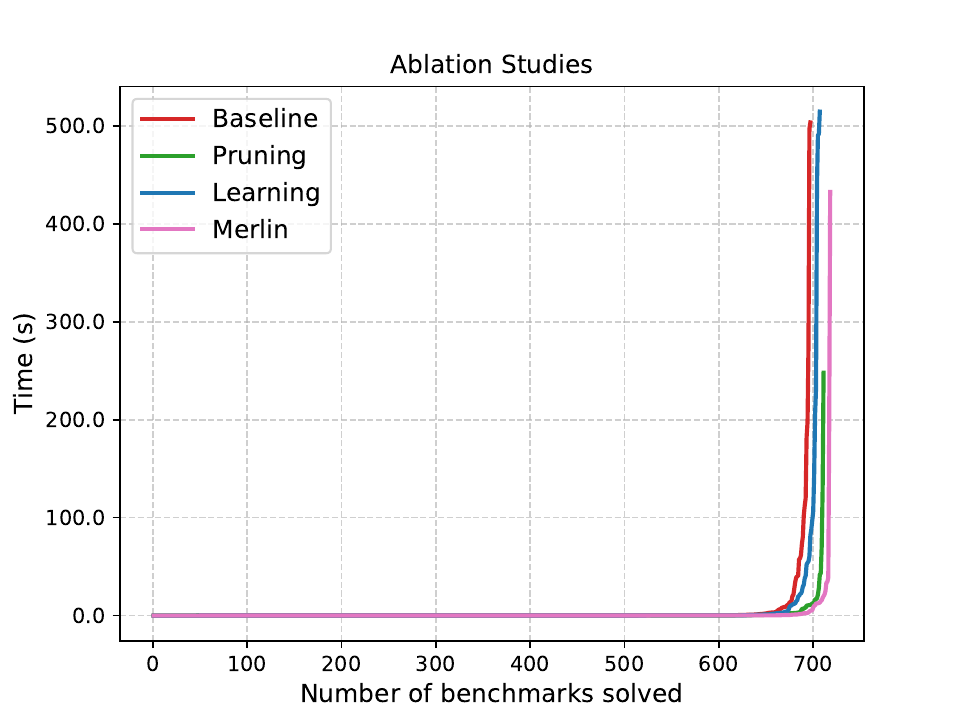}
  \caption{Time comparison for all SyGuS benchmarks.}
  \label{fig:ablation_all_cactus}
\end{subfigure}
\caption{Solving time for the bitvector and string benchmarks and ablation studies.}
\end{figure}

\Cref{fig:bv_fastest,fig:string_fastest} show the number non-trivial benchmarks solved by each tool.
We consider a benchmark to be non-trivial if it takes one of the solvers at least 100 milliseconds to produce a solution.
\Cref{fig:blaze_fastest} shows the number of fastest solved non-trivial benchmarks from the Blaze benchmark set.
\begin{figure}
\begin{subfigure}{.5\textwidth}
  \includegraphics[width=1\linewidth]{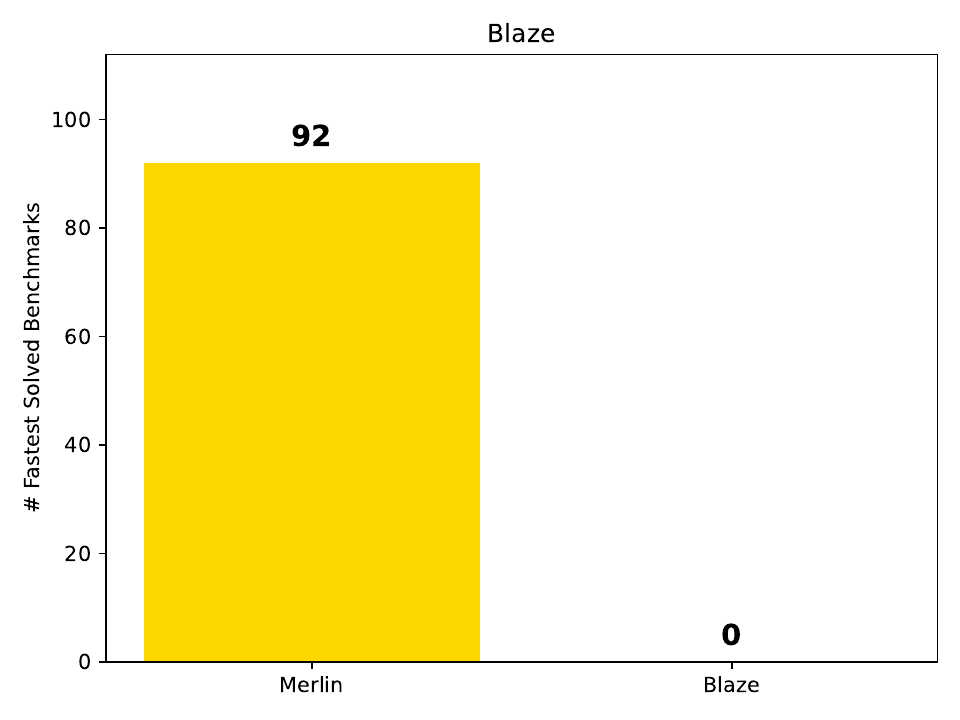}
  \caption{Number of fastest solved non-trivial Blaze benchmarks.}
  \label{fig:blaze_fastest}
\end{subfigure}\begin{subfigure}{.5\textwidth}
  \centering
  \includegraphics[width=1\linewidth]{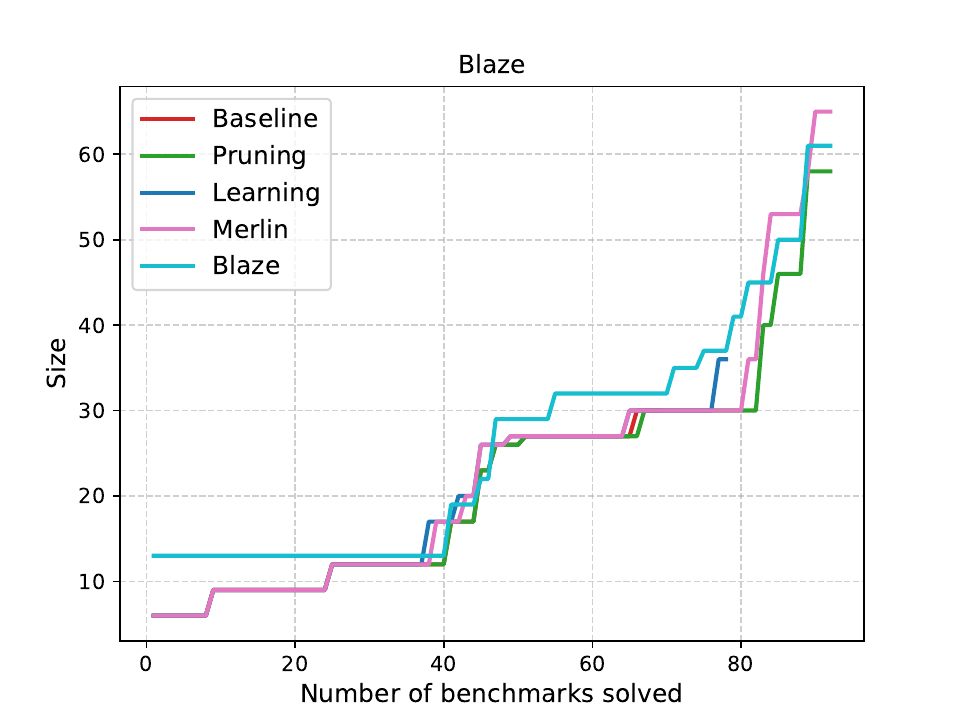}
  \caption{Size comparison for the Blaze benchmarks.}
  \label{fig:blaze_sizes}
\end{subfigure}
\caption{
  Fastest solved benchmarks and size comparison for Blaze Benchmarks.
}
\end{figure}

\begin{figure}
\centering
\begin{subfigure}{.5\textwidth}
  \centering
  \includegraphics[width=1\linewidth]{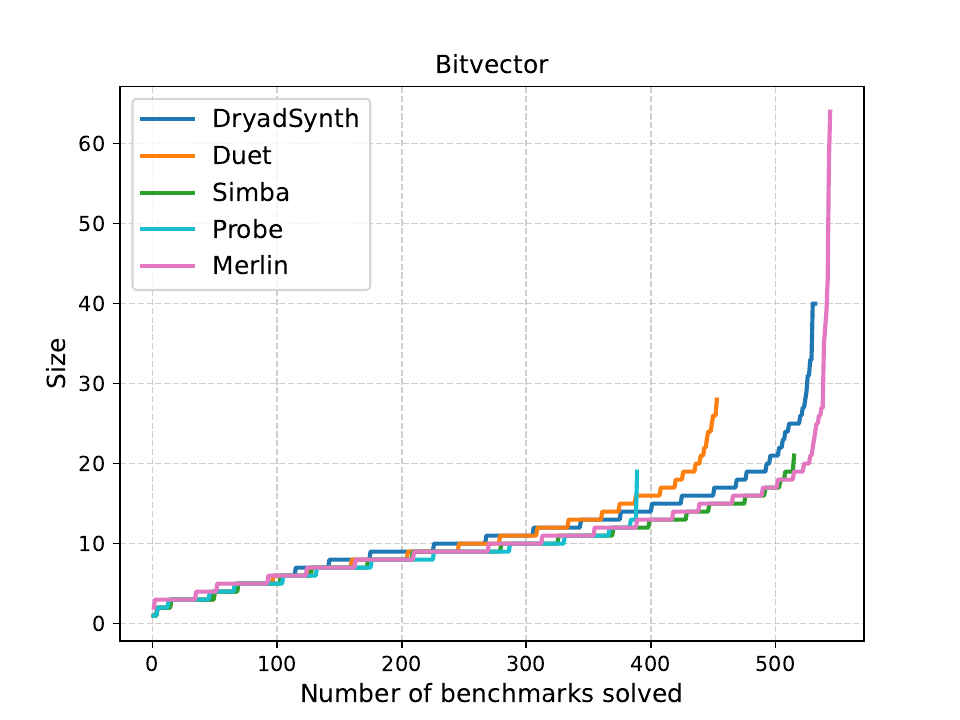}
  \caption{Size comparison for the bitvector domain.}
  \label{fig:bv_sizes}
\end{subfigure}\begin{subfigure}{.5\textwidth}
  \centering
  \includegraphics[width=1\linewidth]{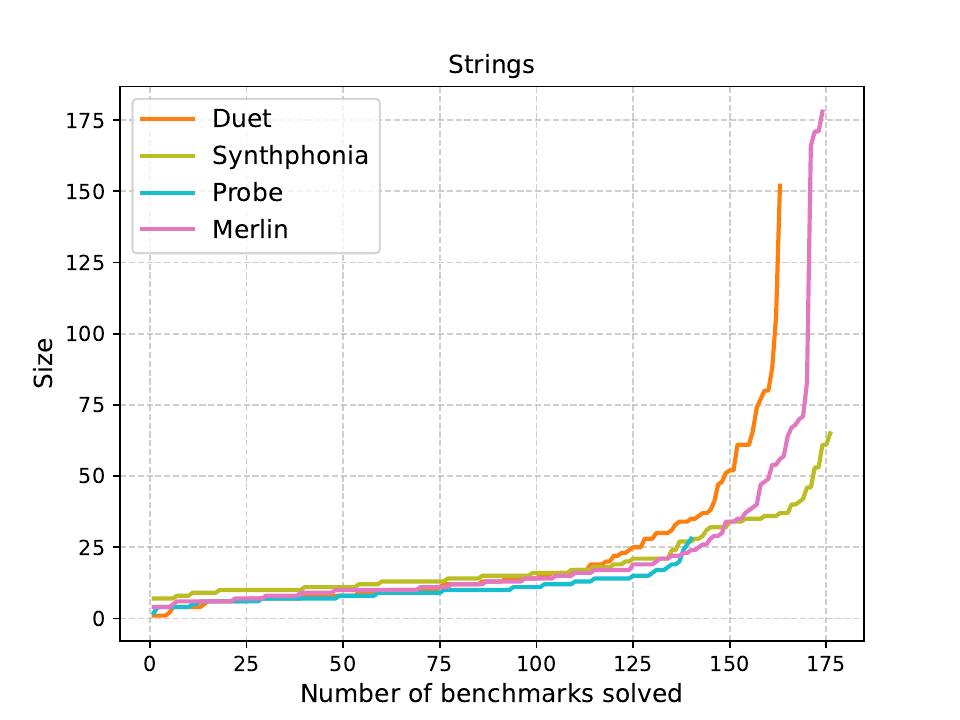}
  \caption{Size comparison for the string domain.}
  \label{fig:string_sizes}
\end{subfigure}
\caption{Solution sizes for the bitvector and string benchmarks.}
\label{fig:sizes}
\end{figure}

\begin{figure}
\centering
\begin{subfigure}{.4\textwidth}
  \centering
  \includegraphics[width=1\linewidth]{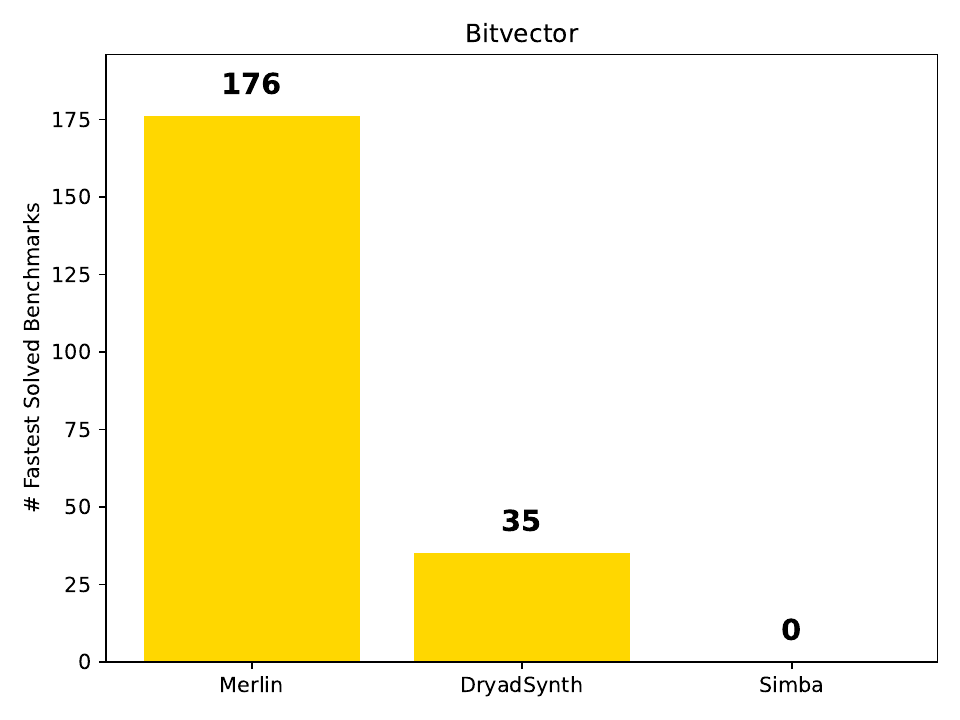}
  \caption{Bitvector domain.}
  \label{fig:bv_fastest}
\end{subfigure}\begin{subfigure}{.4\textwidth}
  \centering
  \includegraphics[width=1\linewidth]{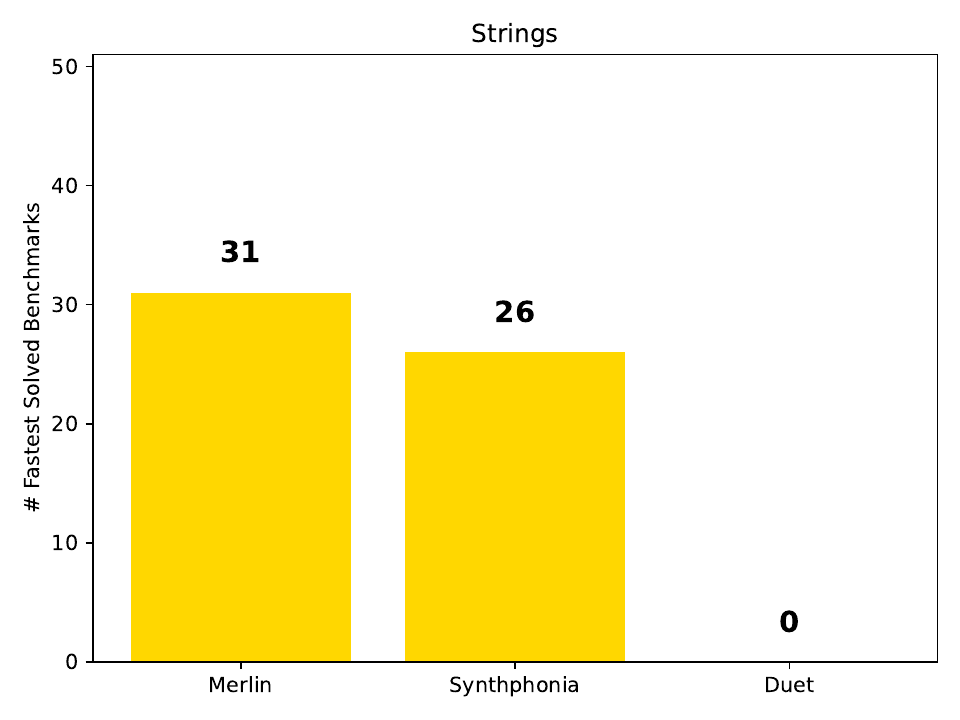}
  \caption{String domain.}
  \label{fig:string_fastest}
\end{subfigure}
\caption{Fastest solved non-trivial benchmarks for the bitvector and string domains.}
\end{figure}

We compare the quality of the benchmark solutions.
\Cref{fig:sizes} depicts the sizes of the benchmark solutions.
The size of the program is an indicator for its quality: 
Following Occam's razor, smaller solutions are better than bigger solutions.
Bigger solutions might be a sign for overfitting the given examples.

\Cref{fig:bv_sizes} shows the sizes for the bitvector benchmarks.
Solutions from Duet and DryadSynth are generally larger than solutions of the other tools.
Probe and Simba generally produce the smallest solutions. 
\toolname{} produces solutions of similar size as Probe and Simba but is also able to produce bigger solutions for harder benchmarks.

\Cref{fig:string_sizes} shows the solution sizes for the string domain.
Again, Probe generally produces very small solutions and Duet tends to generate bigger solutions.
While generating sligthly more complex solutions for easier benchmarks, Synthphonia produces smaller solutions even for harder benchmarks.
\toolname{} initially also produces smaller solutions, but the solution size increases drastically for more complex problems.
\toolname{} is only able to produce such big solutions through its learning mechanism.
Synthphonia has more advanced deduction capabilities.
This advanced deduction makes it produce more compact solutions for complex problems.

\Cref{fig:blaze_sizes} shows the solutions sizes for the Blaze benchmarks.
Overall, the sizes of Blaze's solutions are comparable to the sizes of \toolname{}'s solutions.
Note that Blaze always returns a cost-minimal program, according to some costs defined in the grammar.
Blaze's algorithm could be configured to return the size-minimal program instead. \section{Proofs}
\begin{lemma}
    An orimetric $\functionDistanceMeasure$ induces an equivalence relation $\inducedequivof{\functionDistanceMeasure}$ with 
    \begin{equation*}
        \aFunc \inducedequivof{\functionDistanceMeasure} \aFuncp 
        \quad \text{ if } \quad
        \functionDistanceMeasureOf{\aFunc}{\aFuncp} = 0
        \ .
    \end{equation*}
\end{lemma}
\begin{proof}
    We first show reflexivity. 
    Because $\functionDistanceMeasureOf{\aFunc}{\aFunc} = 0$ by definition, we have $\aFunc \inducedequivof{\functionDistanceMeasure} \aFuncp$.

    Next, we show transitivity.
    We show $\aFunc \inducedequivof{\functionDistanceMeasure} \aFuncpp$.
    We have $\aFunc \inducedequivof{\functionDistanceMeasure} \aFuncp$ and $\aFuncp \inducedequivof{\functionDistanceMeasure} \aFuncpp$.
    This means that $\functionDistanceMeasureOf{\aFunc}{\aFuncp} = 0$ and $\functionDistanceMeasureOf{\aFuncp}{\aFuncpp} = 0$.
    Therefore, $\functionDistanceMeasureOf{\aFunc}{\aFuncp} + \functionDistanceMeasureOf{\aFuncp}{\aFuncpp} = 0$. 
    Because $\functionDistanceMeasureOf{\aFunc}{\aFuncpp} \geq 0$ 
    and 
    $\functionDistanceMeasureOf{\aFunc}{\aFuncpp} \leq \functionDistanceMeasureOf{\aFunc}{\aFuncp} + \functionDistanceMeasureOf{\aFuncp}{\aFuncpp} = 0$,
    we get $\functionDistanceMeasureOf{\aFunc}{\aFuncpp} = 0$.
    Thus, $\aFunc \inducedequivof{\functionDistanceMeasure} \aFuncpp$ holds.

    Lastly, we show symmetry.
    We show $\aFuncp \inducedequivof{\functionDistanceMeasure} \aFunc$.
    We have $\aFunc \inducedequivof{\functionDistanceMeasure} \aFuncp$.
    This means that  
    $\functionDistanceMeasureOf{\aFunc}{\aFuncp} = 0$.
    By definition, we get 
    $\functionDistanceMeasureOf{\aFuncp}{\aFunc} = 0$.
    And thus, 
    $\aFuncp \inducedequivof{\functionDistanceMeasure} \aFunc$
    holds.
\end{proof}

\subsubsection*{Proof of \Cref{lem:bu_closed}}
\begin{proof}[\unskip\nopunct]
   We show that bottom-up enumerable sets are closed under arbitrary unions.
   Let $\aprogSet_1, \aprogSet_2, \ldots$ be bottom-up enumerable sets.
   Let $\aprogSet = \bigcup_{i} \aprogSet_i$ be its union.
   Further, let $\aprog$ be a program in $\aprogSet$ that has children.
   That means $\aprog$ is of shape 
   $\aprog = \anOperatorOf{\aprog_1, \ldots, \aprog_n}$.
   Since $\aprog \in \aprogSet$ there must be a set $\aprogSet_i$ with $\aprog \in \aprogSet_i$.
   Because every $\aprogSet_i$ is bottom-up enumerable, we have
   $\aprog_j \in \aprogSet_i$ for every $1 \leq j \leq n$.
   From $\aprogSet_i \subseteq \aprogSet$ follows that $\aprog_j \in \aprogSet$.
\end{proof}

\subsubsection*{Proof of \Cref{Lemma:BottomUpEnumeration}}
\begin{proof}[\unskip\nopunct]
We first show
$\enumof{\enumerationOrder{}}{\aprogSet} \subseteq \buof{\aprogSet}$.
By definition, $\enumof{\enumerationOrder{}}{\aprogSet}$ returns a bottom-up enumerable set that is a subset of $\aprogSet$.
Because $\buof{\aprogSet}$ is the largest bottom-up enumerable set, the inclusion trivially holds.

Next, we show 
$\enumof{\enumerationOrder{}}{\aprogSet} \supseteq \buof{\aprogSet}$
Towards a contradiction, assume that there is a program $\aprog \in \buof{\aprogSet}$ that does not exist in $\enumof{\enumerationOrder{}}{\aprogSet}$, 
$\aprog \not\in \enumof{\enumerationOrder{}}{\aprogSet}$.
Let $i \in \nat$ be the index of $\aprog$ in $\enumerationOrder{}$.
For $\aprogSet$ not to be in $\enumof{\enumerationOrder{}}{\aprogSet}$ means that there is a subprogram $\aprogp$ of $\aprog$, $\aprogp \subprogrel \aprog$, that was not enumerated before $\aprog$.
Thus, its index $j$ must be greater than $i$.
Therefore, $\aprog \enumerationOrderNotEq{} \aprogp$ although $\aprogp \subprogrel \aprog$.
This is a contradiction to $\enumerationOrder{}$ being a bottom-up enumeration-order.
\end{proof}

\begin{lemma}\label[lemma]{lem:factorize_rep_system}
Given a search space $\aprogSet$ and an enumeration order $\enumerationOrder{}$ and a metric $\functionDistanceMeasure$, 
the function
$\factorizeFuncOf{\aprogSet}{\functionDistanceMeasure}{\enumerationOrder{}}$ returns a representative system 
$\semanticsOf{\factorizeFuncOf{\aprogSet}{\functionDistanceMeasure}{\enumerationOrder{}}}$
for
$\factorize{\semanticsOf{\aprogSet}}{\inducedequivof{\functionDistanceMeasure}}$ with 
$\semanticsOf{\aprogSetp} = \setCollector{\semanticsOf{\aprog}}{\aprog \in \aprogSetp}$.
\end{lemma}
\begin{proof}
   Let $\aprogSetp = 
\factorizeFuncOf{\aprogSet}{\functionDistanceMeasure}{\enumerationOrder{}}$.
The representative system is then $\semanticsOf{\aprogSetp}$.
Let $\aprog$ be a program of $\aprogSet$.
Further, let $\classof{\semanticsOf{\aprog}}$ be its equivalence class.

We first show that there is a program that represents $\aprog$, i.e. there is a program $\aprogp \in \aprogSetp$ with 
$\classof{\semanticsOf{\aprogp}} = \classof{\semanticsOf{\aprog}}$.
If $\aprog \in \aprogSetp$, we are done.
If not, there must be a program $\aprogpp \in \aprogSetp$ with $\semanticsOf{\aprogpp}  \equiv \semanticsOf{\aprog}$.
This also means that $\classof{\semanticsOf{\aprogpp}} = \classof{\semanticsOf{\aprog}}$.

Next, we show that there is exactly one representative for each class.
Towards a contradiction, assume there are programs $\aprog, \aprogp \in \aprogSetp$ with 
$\classof{\semanticsOf{\aprog}} = \classof{\semanticsOf{\aprogp}}$
and $\aprog \neq \aprogp$.
Since $\enumerationOrder{}$ is total, we either have 
$\aprog \enumerationOrderNotEq{} \aprogp$ or 
$\aprogp \enumerationOrderNotEq{} \aprog$.
Assume the first case. The proof for the other case is analogous to the following proof.
Since 
$\classof{\semanticsOf{\aprog}} = \classof{\semanticsOf{\aprogp}}$
we have $\semanticsOf{\aprog} \equiv \semanticsOf{\aprogp}$.
Because $\aprog \enumerationOrderNotEq{} \aprogp$, $\aprogp$ cannot be in $\aprogSetp$. 
\end{proof}

\subsubsection*{Proof of \Cref{Lemma:Enumeration}}
\begin{proof}[\unskip\nopunct]
   The first part of the lemma immediately follows from \Cref{lem:factorize_rep_system}.

   We now show the second part of the lemma.
   Let 
$\aprogSetp = 
\factorizeFuncOf{\aprogSet}{\functionDistanceMeasure}{\enumerationOrder{}}$.
Let $\aprog$ be a program of $\aprogSetp$.
Assume it has children, otherwise there is nothing to show.
This means, $\aprog$ is of the following shape:
$\aprog = \anOperatorOf{\aprog_1, \ldots, \aprog_n}$.
We now have to show that each $\aprog_i$ is in $\aprogSetp$.
Towards a contradiction assume there is a $\aprog_i \not\in \aprogSetp$.
That means there is a program $\aprogp_i \in \aprogSetp$ with 
$\aprogp_i \enumerationOrderNotEq{} \aprog_i$ and 
$\semanticsOf{\aprogp_i} \equiv \semanticsOf{\aprog_i}$.
Let $\aprogp = 
\anOperatorOf{\aprog_1, \ldots, \aprogp_i, \ldots, \aprog_n}$ be the result of replacing $\aprog_i$ by $\aprogp_i$ in $\aprog$.
Because $\ametric$ is a congruence, we have 
$\semanticsOf{\aprog} \inducedequivof{\ametric} \semanticsOf{\aprogp}$.
Since $\aprogp_i \enumerationOrderNotEq{} \aprog_i$, and 
$\enumerationOrder{}$ is a precongruence, 
we have 
$\aprogp \enumerationOrderNotEq{} \aprog$.
Thus, $\aprog$ cannot be a program of $\aprogSetp$.
This is a contradiction.
\end{proof}

\subsubsection*{Proof of \Cref{lem:progress}}
\begin{proof}[\unskip\nopunct]
    An orimetric $\functionDistanceMeasurep$ relates less objects than the orimetric $\functionDistanceMeasure$, if the following holds for all $\anelema, \anelemb \in \aset$: 
    $\functionDistanceMeasurepOf{\anelema}{\anelemb} = 0$ 
    implies
    $\functionDistanceMeasureOf{\anelema}{\anelemb} = 0$.
    Towards a contradiction, assume that the program $\aprog$ is discovered twice.
    That means, for the two orimetrics $\functionDistanceMeasure$ and $\functionDistanceMeasurep$, $\groundTruthFunc$ and $\semanticsOf{\aprog}$
    have distance $0$:
    $\functionDistanceMeasureOf{\semanticsOf{\aprog}}{\groundTruthFunc} = 0$ and 
    $\functionDistanceMeasurepOf{\semanticsOf{\aprog}}{\groundTruthFunc} = 0$.
    Otherwise, $\findProgFunc$ would not have returned the programs.
    Let $\ametricpar{\aprog}$ be the metric resulting from the refinement that happened immediately after 
    $\aprog$ was first found using $\functionDistanceMeasure$.
    Then $\ametricpar{}$ is a refinement of $\functionDistanceMeasure$.
    Therefore, in $\inducedequivof{\ametricpar{}}$, $\aprog$ and $\groundTruthFunc$ are not related,
    $\ametricparof{\aprog}{\semanticsOf{\aprog}}{\groundTruthFunc} \neq 0$.
    Since $\functionDistanceMeasurep$ is either $\ametricpar{\aprog}$ or some later refinement,
    $\functionDistanceMeasurep \subseteq \ametricpar{\aprog}$ holds.
    This means, the equivalence of $\ametricpar{\aprog}$ relates more programs than the induced equivalence of $\functionDistanceMeasurep$.
    But we assumed $\functionDistanceMeasurepOf{\semanticsOf{\aprog}}{\groundTruthFunc} = 0$. 
    That would imply that 
    $\ametricparof{\aprog}{\semanticsOf{\aprog}}{\groundTruthFunc} = 0$.
    This is a contradiction.
\end{proof}

\subsubsection*{Proof of \Cref{lem:cegar:termination}}
\begin{proof}[\unskip\nopunct]
Let $\aprogSet$ be the search space returned after $\pruneFunc$ and $\factorizeFunc$ are executed.
The search space $\aprogSet$ is complete wrt.\ $\functionDistanceMeasure$.
This means there is a program $\aprogp \in \aprogSet$ with
$\functionDistanceMeasureOf{\semanticsOf{\aprogp}}{\semanticsOf{\aprog}} = 0$.
Thus, at some point $\aprogp$ is the current program in $\findProgFunc$ and will be returned.
\end{proof}

\subsubsection*{Proof of \Cref{lem:lifting}}
\begin{proof}[\unskip\nopunct]
    Let $\ametrictilde$ be the orimetric that will be lifted,
    and $\ametricpar{\asety}$ the lifted metric.

    We first show reflexivity.
    \begin{equation*}
    \ametricparof{\asety}{\aFunc}{\aFunc} 
    = 
    \sum_{\anelemy \in \asety}
    \ametrictildeof{\aFuncOf{\anelemy}}{\aFuncOf{\anelemy}}
    \end{equation*}
    Because $\ametrictilde$ is an orimetric, we have 
    $\ametrictildeof{\aFuncOf{\anelemy}}{\aFuncOf{\anelemy}} = 0$.
    This carries over to the sum.

    Next, we show symmetry at zero.
    Assume 
    $
    \ametricparof{\asety}{\aFunc}{\aFuncp} 
    = 0
    $.
    This means 
    $
    \sum_{\anelemy \in \asety}
    \ametrictildeof{\aFuncOf{\anelemy}}{\aFuncpOf{\anelemy}}
    = 0
    $.
    Therefore, we have 
    $
    \ametrictildeof{\aFuncOf{\anelemy}}{\aFuncpOf{\anelemy}} = 0
    $
    for every $\anelemy \in \asety$.
    Because $\ametrictilde$ is an orimetric, that means 
    $
    \ametrictildeof{\aFuncpOf{\anelemy}}{\aFuncOf{\anelemy}} = 0
    $
    for every $\anelemy \in \asety$.
    This, again, carries over to the sum.
    Thus,
    $
    \ametricparof{\asety}{\aFuncp}{\aFunc} 
    = 0
    $.

    Last, we show the triangle inequality.
    \begin{align*}
        \ametricparof{\asety}{\aFunc}{\aFuncpp} 
        &\ =
        \sum_{\anelemy \in \asety}
        \ametrictildeof{\aFuncOf{\anelemy}}{\aFuncppOf{\anelemy}}
        \\
        &\ 
        \leq 
        \sum_{\anelemy \in \asety}
        \ametrictildeof{\aFuncOf{\anelemy}}{\aFuncpOf{\anelemy}}
        +
        \ametrictildeof{\aFuncpOf{\anelemy}}{\aFuncppOf{\anelemy}}
        \tag{$\triangle$ of $\ametrictilde$}
        \\
        &\
        = 
        \sum_{\anelemy \in \asety}
        \ametrictildeof{\aFuncOf{\anelemy}}{\aFuncpOf{\anelemy}}
        +
        \sum_{\anelemy \in \asety}
        \ametrictildeof{\aFuncpOf{\anelemy}}{\aFuncppOf{\anelemy}}
        \\
        &\ 
        =
        \ametricparof{\asety}{\aFunc}{\aFuncp} 
        +
        \ametricparof{\asety}{\aFuncp}{\aFuncpp} 
    \end{align*}

\end{proof}

\subsubsection*{Proof of \Cref{lem:lift_quasi}}
\begin{proof}[\unskip\nopunct]
Note that two equivalence classes are equal, if a representative $\anelema$ of the first equivalence class and a representative of the second equivalence class $\anelemb$ have distance zero.

Let $\ametrictilde$ be the orimetric that will be lifted,
and $\ametricpar{\inputexamplesp}$ the lifted orimetric with 
$\inputexamplesp \subseteq \inputexamples$.
We first show unambiguity wrt.\ $\groundTruthFuncs$.
For this, let $\groundTruthFunc, \groundTruthFuncp$ be ground truth functions of $\groundTruthFuncs$.
We show that 
$\ametricparof{\inputexamplesp}{\groundTruthFunc}{\groundTruthFuncp} = 0$.
\begin{align*}
\ametricparof{\inputexamplesp}{\groundTruthFunc}{\groundTruthFuncp}
&\ 
= 
\sum_{\aninput \in \inputexamplesp}
\ametrictildeof{\groundTruthFuncOf{\aninput}}{\groundTruthFuncpOf{\aninput}}
=
0
\end{align*}
We explain the last equality.
Since $\aninput$ is an input example and thus the output is given by the specification, 
we have
$
\groundTruthFuncOf{\aninput} = \groundTruthFuncpOf{\aninput}
= \anoutput
$.
Then, by reflexivity of $\ametrictilde$, we get 
$\ametrictildeof{\anoutput}{\anoutput} = 0$.

We show the second part of the lemma.
Let $\ametrictilde$ be the quasimetric that will be lifted,
and $\ametricpar{\inputexamples}$ the lifted orimetric.
Let $\aFunc\not\in \groundTruthFuncs$ be a function that is not a ground truth function.
Further, let $\groundTruthFunc \in \groundTruthFuncs$ be a ground truth function.
We show that 
$\ametricparof{\inputexamples}{\aFunc}{\groundTruthFunc} \neq 0$.
Towards a contradiction, assume that
$\ametricparof{\inputexamples}{\aFunc}{\groundTruthFunc} = 0$.
We have 
\begin{align*}
\ametricparof{\inputexamples}{\aFunc}{\groundTruthFunc}
&\ 
= 
\sum_{\aninput \in \inputexamples}
\ametrictildeof{\aFuncOf{\aninput}}{\groundTruthFuncOf{\aninput}}
= 0
\end{align*}
For this to be true, each addend must be $0$.
Because $\ametrictilde$ is a quasimetric, 
that means that $\aFuncOf{\aninput} = \groundTruthFuncOf{\aninput}$ for every input example $\aninput \in \inputexamples$. 
Thus, $\aFunc$ satisfies the specification.
Therefore, $\aFunc \in \groundTruthFunc$.
This is a contradiction.
\end{proof}

\subsubsection*{Size Hack}
\begin{lemma}\label[lemma]{lem:size_hack}
    Given an orimetric $\functionDistanceMeasure$, then $\functionDistanceMeasurep$ also is an orimetric.
\end{lemma}
\begin{proof}
The formal definition of $\functionDistanceMeasurep$ is:
\begin{equation*}
    \functionDistanceMeasurepOf{\semanticsOf{\aprog}}{\aFunc} = 
    \begin{cases}
        \functionDistanceMeasureOf{\semanticsOf{\aprog}}{\aFunc} 
        & \text{, }\functionDistanceMeasureOf{\semanticsOf{\aprog}}{\aFunc} < \radius
        \\
        \radius - \epsilon 
        & \text{, }\sizeOf{\aprog} \leq \asizeThreshold \vee \exists \aprogp \in \aprogSet: \sizeOf{\aprogp} \leq \asizeThreshold \wedge \semanticsOf{\aprog} \inducedequivof{\functionDistanceMeasure} \semanticsOf{\aprogp}
        \\
        \radius & \text{, otherwise.}
    \end{cases}
\end{equation*}

As $\epsilon$, we take half of the minimum distance between two functions, and the minimum difference between the distance of two functions and $\radius$.
For the minimum to exist, we restrict the function space to functions expressible by programs of whose size is less than an exorbitant size, say 1 million, plus $\groundTruthFunc$.
    We call the function space $\allFuncs'$.
    Then $\epsilon$ is formally defined as:
    \begin{equation*}
        \epsilon = 
        \frac{
            \min(\setCollector{
                \functionDistanceMeasureOf{\aFunc}{\aFuncp}, 
                \radius - \functionDistanceMeasureOf{\aFunc}{\aFuncp}
            }{\aFunc, \aFuncp \in \allFuncs'} 
            \setminus 
            \reals_{\leq 0}
            ) 
        }{2}    
    \end{equation*}

    We first show reflexivity.
    Because, $\functionDistanceMeasure$ is an orimetric, we have 
    $\functionDistanceMeasureOf{\aFunc}{\aFunc} = 0$.
    Because $0 < \radius$, we have $\functionDistanceMeasurepOf{\aFunc}{\aFunc} = 0$.

    Next, we show symmetry at zero.
    If $\functionDistanceMeasurepOf{\aFunc}{\aFuncp} = 0$,
    then we know that 
    $\functionDistanceMeasureOf{\aFunc}{\aFuncp} = 0$.
    Because $\functionDistanceMeasure$ is an orimetric, we have that
    $\functionDistanceMeasureOf{\aFuncp}{\aFunc} = 0$.
    And therefore
    $\functionDistanceMeasurepOf{\aFuncp}{\aFunc} = 0$
    holds.

    Last, we show the triangle inequality.
    Consider $\functionDistanceMeasurepOf{\semanticsOf{\aprog}}{\aFunc}$.
    We differentiate two cases. 
    The first case is when $\functionDistanceMeasureOf{\semanticsOf{\aprog}}{\aFunc} < \radius$.
    The second case is when $\functionDistanceMeasureOf{\semanticsOf{\aprog}}{\aFunc} \geq \radius$.

    We start with the first case, so $\functionDistanceMeasureOf{\semanticsOf{\aprog}}{\aFunc} < \radius$.
    Then, we have $\functionDistanceMeasureOf{\semanticsOf{\aprog}}{\aFunc} = 
    \functionDistanceMeasurepOf{\semanticsOf{\aprog}}{\aFunc} < \radius$, and thus we have 
    $\functionDistanceMeasurepOf{\semanticsOf{\aprog}}{\aFunc} \leq \radius - \epsilon$.
    Consider now 
    $\functionDistanceMeasureOf{\semanticsOf{\aprog}}{\semanticsOf{\aprogp}} + 
    \functionDistanceMeasureOf{\semanticsOf{\aprogp}}{\aFunc}$. 
    If both addends are $< \radius$, we get the triangle equality directly from the triangle inequality of $\functionDistanceMeasure$.
    If one addend is $\geq \radius$, then the result in $\functionDistanceMeasurep$ is $\geq \radius - \epsilon$ and thus also greater than or equal to 
    $\functionDistanceMeasurepOf{\semanticsOf{\aprog}}{\aFunc}$.

    In the second case we have $\functionDistanceMeasureOf{\semanticsOf{\aprog}}{\aFunc} \geq \radius$. 
    Then, $\functionDistanceMeasurepOf{\semanticsOf{\aprog}}{\aFunc} \geq \radius - \epsilon$.
    We have two subcases.
    The first subcase is that $\sizeOf{\aprog} \leq \asizeThreshold$ or there exists $\aprogp \in \aprogSet$ with $\sizeOf{\aprogp} \leq \asizeThreshold$  and $\semanticsOf{\aprog} \inducedequivof{\functionDistanceMeasure} \semanticsOf{\aprogp}$.
    The second subcase is the negation of the first subcase.

    We start with the first subcase.
    We have that $\functionDistanceMeasurepOf{\semanticsOf{\aprog}}{\aFunc} = \radius - \epsilon$.
    That means we have 
    \begin{equation*}
    \functionDistanceMeasurepOf{\semanticsOf{\aprog}}{\aFunc} 
    \leq 
    \functionDistanceMeasureOf{\semanticsOf{\aprog}}{\aFunc} 
    \leq 
    \functionDistanceMeasureOf{\semanticsOf{\aprog}}{\semanticsOf{\aprogp}}
    + 
    \functionDistanceMeasureOf{\semanticsOf{\aprogp}}{\aFunc}
    \ .
    \end{equation*}
    Again, if both addends are $< \radius$, the triangle inequality follows directly.
    If one addend is $\geq \radius$, then the result in $\functionDistanceMeasurep$ is $\geq \radius - \epsilon$ and thus also greater than or equal to 
    $\functionDistanceMeasurepOf{\semanticsOf{\aprog}}{\aFunc}$.

    In the second subcase we have 
    $\functionDistanceMeasurepOf{\semanticsOf{\aprog}}{\aFunc} = \radius$
    and 
    $\functionDistanceMeasureOf{\semanticsOf{\aprog}}{\aFunc} \geq \radius$
    .
    This means 
    \begin{equation*}
        \radius 
        =
        \functionDistanceMeasurepOf{\semanticsOf{\aprog}}{\aFunc}
        \leq 
        \functionDistanceMeasureOf{\semanticsOf{\aprog}}{\aFunc}
        \leq 
        \functionDistanceMeasureOf{\semanticsOf{\aprog}}{\semanticsOf{\aprogp}}
        +
        \functionDistanceMeasureOf{\semanticsOf{\aprogp}}{\aFunc}.
    \end{equation*}
    Again, if both addends are $< \radius$, the triangle inequality follows directly.
    If not, consider the addends in $\functionDistanceMeasurep$: 
    $
        \functionDistanceMeasurepOf{\semanticsOf{\aprog}}{\semanticsOf{\aprogp}}
        $ and $
        \functionDistanceMeasurepOf{\semanticsOf{\aprogp}}{\aFunc}
        $
    .
    At least one of the addends is $\geq \radius - \epsilon$.
    Since it also is $\leq \radius$, we have the two cases:
    Either it is equal to $\radius - \epsilon$ or it is equal to $\radius$.
    In the latter case, we are done, because 
    $
        \radius 
        =
        \functionDistanceMeasurepOf{\semanticsOf{\aprog}}{\aFunc} 
        $.

    In the former case, 
    towards a contradiction, assume the first addend is equal to $\radius - \epsilon$.
    So, we have 
    $\functionDistanceMeasurepOf{\semanticsOf{\aprog}}{\semanticsOf{\aprogp}} = \radius - \epsilon$.
    For this to be true, 
    $\sizeOf{\aprog} \leq \asizeThreshold$ or there exists $\aprogpp \in \aprogSet$ with $\sizeOf{\aprogpp} \leq \asizeThreshold$  and $\semanticsOf{\aprog} \inducedequivof{\functionDistanceMeasure} \semanticsOf{\aprogpp}$.
    Then we would not be in this case. 
    This is a contradiction.

    Next, assume 
    the second addend is equal to $\radius - \epsilon$:
    $\functionDistanceMeasurepOf{\semanticsOf{\aprogp}}{\aFunc} = \radius - \epsilon$.
    Then, 
    $\sizeOf{\aprogp} \leq \asizeThreshold$ or there exists $\aprogpp \in \aprogSet$ with $\sizeOf{\aprogpp} \leq \asizeThreshold$  and $\semanticsOf{\aprogp} \inducedequivof{\functionDistanceMeasure} \semanticsOf{\aprogpp}$.
    For 
     \begin{equation*}
        \functionDistanceMeasurepOf{\semanticsOf{\aprog}}{\aFunc}
        \leq 
        \functionDistanceMeasurepOf{\semanticsOf{\aprog}}{\semanticsOf{\aprogp}}
        +
        \functionDistanceMeasurepOf{\semanticsOf{\aprogp}}{\aFunc}
    \end{equation*}   
    to be violated, we would need 
    $\functionDistanceMeasurepOf{\semanticsOf{\aprog}}{\semanticsOf{\aprogp}} = 0$ because all other distances are greater than $\epsilon$.
    This, would mean that $\functionDistanceMeasureOf{\semanticsOf{\aprog}}{\semanticsOf{\aprogp}} = 0$.
    Because $\inducedequivof{\functionDistanceMeasure}$ is transitive, 
    we have $\semanticsOf{\aprogpp} \inducedequivof{\functionDistanceMeasure} \semanticsOf{\aprog}$.
    This, however, has the consequence that there exists the same $\aprogpp$ with $\sizeOf{\aprogpp} \leq \asizeThreshold$.
    Therefore, 
    $\functionDistanceMeasurepOf{\semanticsOf{\aprog}}{\aFunc}$ should have been  $\radius - \epsilon$.
    This is a contradiction to 
    $\functionDistanceMeasurepOf{\semanticsOf{\aprog}}{\aFunc} = \radius$.
\end{proof}
 \section{Details on Deduction}\label{sec:deduction}
In the following, we describe how to efficiently check if a successor of a program is determined through deduction.
The explanation will reason about settings with one input-output example $(\aninput,\anoutput)$ only.
This is just to ease the notation and can easily be lifted to settings with multiple input-output examples.

For all deduction techniques, if we enumerated the program $\aprog$, we first compute its output $\anoutput_{\aprog} = \semanticsOfAppliedTo{\aprog}{\aninput}$.
Also, in order to make the deduction efficient, we maintain sets of programs that have been enumerated.
They are usually indexed by the programs' output values for easy lookup.

\subsection{Deduction for Bitvectors}
In the following, we describe how to efficiently check if a successor of a program is determined through deduction.
The explanation will reason about settings with one input-output example $(\aninput,\anoutput)$ only.
This is just to ease the notation and can easily be lifted to settings with multiple input-output examples.
For all deduction techniques, if we enumerated the program $\aprog$, we first compute its output $\anoutput_{\aprog} = \semanticsOfAppliedTo{\aprog}{\aninput}$.
Also, in order to make the deduction efficient, we maintain sets of programs that have been enumerated.
They are usually indexed by the programs' output values for easy lookup.

For bitvectors, we have nine sketches.
Four of them are also used by DryadSynth:
$\andop(\sketchhole, \sketchhole)$, 
$\orop(\sketchhole, \sketchhole)$, 
$\addop(\sketchhole, \sketchhole)$,
and 
$\xorop(\sketchhole, \sketchhole)$.
The remaining five sketches are new:
$\mulop(\sketchhole, \sketchhole)$, 
$\notop(\addop(\sketchhole, \sketchhole))$, 
$\negop(\addop(\sketchhole, \sketchhole))$, 
$\notop(\xorop(\sketchhole, \sketchhole))$,
and 
$\negop(\xorop(\sketchhole, \sketchhole))$.

\paragraph{$\andop(\sketchhole, \sketchhole)$}
For the $\andop(\sketchhole, \sketchhole)$ sketch, we keep a set of programs $\aprogSet_1$ that return a bitvector which is bitwise greater than the output bitvector $\anoutput$.
If we now enumerate a program $\aprog$ that returns the bitvector $\anoutput_{\aprog}$, we first check whether the program is a viable candidate for the $\andop$ operator by checking whether $\anoutput_{\aprog}$ is bitwise greater than $\anoutput$.
If that is the case, we search through $\aprogSet_1$ to find a program $\aprog_1$ for which $\semanticsOfAppliedTo{\andop(\aprog, \aprog_1)}{\aninput} = \anoutput$ holds.
If there is such a program, $\andop(\aprog, \aprog_1)$ is one of the immediate successors of $\aprog$.
\paragraph{$\orop(\sketchhole, \sketchhole)$}
For the $\orop(\sketchhole, \sketchhole)$ sketch, we keep a set of programs $\aprogSet_1$ that return a bitvector which is bitwise less than the output bitvector $\anoutput$.
The deduction procedure is analogous to the one for the $\andop(\sketchhole, \sketchhole)$ sketch.
\paragraph{$\mulop(\sketchhole, \sketchhole)$}
For the $\mulop(\sketchhole, \sketchhole)$ sketch, we keep a set of programs $\aprogSet_1$ that return a bitvector.
If we now enumerate a program $\aprog$ that returns the bitvector $\anoutput_{\aprog}$, we compute candidate inputs for the $\mulop$ operator.
To get viable candidates, we compute all values $\anoutput_1$ 
that solve the equation $\anoutput_{\aprog} \times \anoutput_1 = \anoutput$.
This involves finding the multiplicative inverse of $\anoutput_{\aprog}$ in the integer ring represented by the bitvectors using the extended euclidean algorithm.
For each candidate $\anoutput_1$, we check if $\aprogSet_1$ contains a program $\aprog_1$ whose output is $\anoutput_1$.
If there is such a program, $\mulop(\aprog, \aprog_1)$ is one of the immediate successors of $\aprog$.
\paragraph{$\addop(\sketchhole, \sketchhole)$}
For the $\addop(\sketchhole, \sketchhole)$ sketch, we also keep a set of programs $\aprogSet_1$ that return a bitvector.
If we now enumerate a program $\aprog$ that returns $\anoutput_{\aprog}$, we can directly compute the value $\anoutput_1$ we need such that $\anoutput_{\aprog} + \anoutput_1 = \anoutput$ by calculating $\anoutput_1 = \anoutput - \anoutput_{\aprog}$.
Then, we look up whether there exists a program $\aprog_1$ in $\aprogSet_1$ 
whose output is $\anoutput_1$.
If there is such a program, $\addop(\aprog, \aprog_1)$ is the successor of $\aprog$.

\paragraph{Remaining Sketches}
The deduction for the remaining sketches is analogous to the deduction for the $\addop(\sketchhole, \sketchhole)$ sketch.
For each sketch, we can compute the value $\anoutput_1$ needed to complete the sketch.
Then, we look up if there is a program in $\aprogSet_1$ whose output is $\anoutput_1$. 
If so, we fill the sketch holes with $\aprog$ and $\aprog_1$ and determine the resulting program to be one of the immediate successors.

\subsection{Deduction for Strings}
For Strings, we have three sketches which use one operator each:
$\concatOp(\sketchhole, \sketchhole)$, $\replaceOp(\sketchhole, \sketchhole, \sketchhole)$, and $\substrOp(\sketchhole, \sketchhole, \sketchhole)$.

\paragraph{$\replaceOp(\sketchhole, \sketchhole, \sketchhole)$}
The operator $\replaceOp$ takes as arguments three strings: the first is the string in which the replacement should happen, namely the first occurrence of the second argument, should it exist, will be replaced by the third argument. 
We keep two sets of programs that are enumerated earlier: $\aprogSet_1$ and $\aprogSet_2$.
In $\aprogSet_1$, we keep all programs that return a string.
In $\aprogSet_2$, we keep all programs that return a substring of the expected output $\anoutput$.
If there is a program $\aprog_1 \in \aprogSet_1$ whose output $\anoutput_1 = \semanticsOfAppliedTo{\aprog_1}{\aninput}$ returns a substring of $\anoutput_{\aprog}$
and there also is a program $\aprog_2 \in \aprogSet_2$ for which 
$\semanticsOfAppliedTo{\replaceOpOf{\aprog}{\aprog_1}{\aprog_2}}{\aninput} = \anoutput$ holds,
then the one of the immediate successors in the enumeration order is
$\replaceOpOf{\aprog}{\aprog_1}{\aprog_2}$.
We check the same for the other positions in the sketch, i.e.\ we search programs $\aprog_1, \aprog_2$ with 
$\semanticsOfAppliedTo{\replaceOpOf{\aprog_1}{\aprog}{\aprog_2}}{\aninput} = \anoutput$
or 
$\semanticsOfAppliedTo{\replaceOpOf{\aprog_1}{\aprog_2}{\aprog}}{\aninput} = \anoutput$.

For the sketch $\concatOpOf{\sketchhole}{\sketchhole}$, one necessary precondition for the programs that fill the sketch holes is that they must produce substrings of the target output.
For strings that do not satisfy the precondition, we penalize them with a big constant $\bigconst \in \reals$.
Otherwise, the distance between two strings is the difference of their length.
Formally, we define the quasimetric $\distanceMeasure$ as follows:
\begin{equation*}
\distanceMeasureOf{\astring}{\astringp} = 
\begin{cases}
    \lengthOf{\astringp} - \lengthOf{\astring}
    &, \substrOf{\astring}{\astringp}
    \\
    \bigconst + 
    \absOf{\lengthOf{\astringp} - \lengthOf{\astring}}
    &, \textit{else} \ .
\end{cases}
\end{equation*}

The key insight from the perspective of oriented metrics is that the first program should produce a superstring of the output
We keep three sets.
In the first set $\aprogSet_1 \subseteq \completeLanguageOf{\agrammar} \times \completeLanguageOf{\agrammar}$ we keep all pairs of programs that we have explored, where the first program returns a superstring $\anoutput_s$ of the output and the second program returns the indices of $\anoutput$ in $\anoutput_s$. 
In the second set $\aprogSet_2 \subseteq \completeLanguageOf{\agrammar} \times \nat$ we keep programs that output a superstring $\anoutput_s$ of $\anoutput$ alongside with the starting index of $\anoutput$ in $\anoutput_s$.
In the third set $\aprogSet_3 \subseteq \completeLanguageOf{\agrammar}$ we keep all programs that return an integer.
We have one approach for programs that return strings and one for programs that returns integers:

If $\aprog$ returns a superstring $\anoutput_{\aprog}$ of $\anoutput$, 
then let $k$ be the index of $\anoutput$ in $\anoutput_{\aprog}$.
Next, we look up whether there is a program $\aprog_k \in \aprogSet_3$ that returns $k$.
If so, we look for a program $\aprog_l \in \aprogSet_3$ which specifies the correct length, i.e.\ $\semanticsOfAppliedTo{\substrOpOf{\aprog}{\aprog_k}{\aprog_l}}{\aninput} = \anoutput$. 
If there is such a program, $\substrOpOf{\aprog}{\aprog_k}{\aprog_l}$ is one of the immdediate successors of $\aprog$.

If $\aprog$ returns an integer, we check if there is a program pair $(\aprog_s, \aprog_k) \in \aprogSet_1$ for which the current program returns the correct length,
i.e.\ $\semanticsOfAppliedTo{\substrOpOf{\aprog_s}{\aprog_k}{\aprog}}{\aninput} = \anoutput$ holds.
In that case, $\substrOpOf{\aprog_s}{\aprog_k}{\aprog}$ is one of the successor programs.

Otherwise, we check if the current program returns a suitable starting index for a program we have enumerated before.
If $\aprog$ returns the integer $k$,
We look for a pair $(\aprog_s, k) \in \aprogSet_2$. 
If it exists, we search through $\aprogSet_3$ again to find a program $\aprog_l$ with 
$\semanticsOfAppliedTo{\substrOpOf{\aprog_s}{\aprog}{\aprog_l}}{\aninput} = \anoutput$.
If there is such a program, 
$\substrOpOf{\aprog_s}{\aprog}{\aprog_l}$
is the successor of $\aprog$. \section{Example: Instantiation of ESolver and Observational Equivalence Factorization}
As an example, we instantiate ESolver in our framework.
To recap, ESolver uses size based enumeration with OE factorization.
Size based enumeration discovers programs in increasing size.
Here, the size is measured as the number of nodes in the AST representing the program.

Following the instantiation in \Cref{sec:lifting}, we only need to define a quasimetric.
For bitvectors, one could use the hamming distance as the quasimetric, for example.
We call the resulting metric $\distanceMeasureForExamples{\inputexamples}$.

Now that we have a metric in place, the next step is to define an enumeration order $\enumerationOrder{}$.
We use a size based enumeration order.
As described in \Cref{sec:overview}, ESolver only factorizes the search space with OE. 
Thus, for pruning, we set the radius $\radius$ of the ball to infinity.
$\pruneFunc$ therefore returns the identity: 
$\pruneFuncOf{\aprogSet}{\functionDistanceMeasure}{\groundTruthFunc} = \aprogSet$.
The function $\factorizeFunc$ is defined as above.

Since $\distanceMeasureForExamples{\inputexamples}$ respects congruence, we have that if 
the search space $\completeLanguageOf{\agrammar}$ is bottom-up enumerable, $\factorizeFuncOf{\completeLanguageOf{\agrammar}}{\distanceMeasureForExamples{\inputexamples}}{\enumerationOrder{}}$ is bottom-up enumerable and complete wrt.\ $\completeLanguageOf{\agrammar}$ and $\distanceMeasureForExamples{\inputexamples}$.

Because $\distanceMeasureForExamples{\inputexamples}$ is precise and unambiguous, $\refineFunc$ and $\updateEnumOrder$ will not be reached (\Cref{lem:cegar:termination}).
Thus, we omit defining a $\refineFunc$ or a $\updateEnumOrder$ function for this example.
The resulting instantiation is standard size based enumeration with OE factorization as is done in ESolver. %
\clearpage{}

\end{document}